\documentclass[a4paper,USenglish,cleveref,thm-restate]{lipics-v2021}
\usepackage{doi}
\usepackage[numbers, sort&compress]{natbib}
\usepackage{amsmath,amsthm,enumitem,graphicx,xcolor}
\usepackage{framed}

\hideLIPIcs
\nolinenumbers

\newtheorem{problem}{Problem}

\newcommand{\alg}{\mathcal{A}}
\newcommand{\dist}{\operatorname{dist}}
\newcommand{\poly}{\operatorname{poly}}

\newcommand{\MIS}{Maximal Independent Set}

\title{Distributed Computation with Local Advice}

 \author{Alkida Balliu}{Gran Sasso Science Institute (GSSI), Italy}{}{}{}
 \author{Sebastian Brandt}{CISPA Helmholtz Center for Information Security, Germany}{}{}{}
 \author{Fabian Kuhn}{University of Freiburg, Germany}{}{}{}
 \author{Krzysztof Nowicki}{Unaffiliated}{}{}{}
 \author{Dennis Olivetti}{Gran Sasso Science Institute (GSSI), Italy}{}{}{}
 \author{Eva Rotenberg}{IT University of Copenhagen, Denmark}{}{}{}
 \author{Jukka Suomela}{Aalto University, Finland}{}{}{}

\authorrunning{A. Balliu, S. Brandt, F. Kuhn, K. Nowicki, D. Olivetti, E. Rotenberg, J. Suomela}



\funding{This work was supported by the MUR (Italy) Department of Excellence 2023 - 2027 for GSSI, the PNRR MIUR research project GAMING “Graph Algorithms and MinINg for Green agents” (PE0000013, CUP D13C24000430001), the Independent Research Fund Denmark grant 2020-2023 (9131-00044B) ``Dynamic Network Analysis'' and the VILLUM Foundation grant (VIL37507) ``Efficient Recomputations for Changeful Problems''.}

\acknowledgements{This work was initiated at Dagstuhl Seminar 23071: ``From Big Data Theory to Big Data Practice''.
We would like to thank Sameep Dahal, Laurent Feuilloley, and Augusto Modanese for insightful discussions, and the anonymous reviewers for their helpful comments on previous versions of this work.}

\begin{document}
	
\maketitle

\begin{abstract}
	Algorithms with advice have received ample attention in the distributed and online settings, and they have recently proven useful also in dynamic settings. In this work we study \emph{local computation with advice}: the goal is to solve a graph problem $\Pi$ with a distributed algorithm in $T(\Delta)$ communication rounds, for some function $T$ that only depends on the maximum degree $\Delta$ of the graph, and the key question is how many bits of advice per node are needed.	

Some of our results regard \emph{Locally Checkable Labeling problems} (LCLs), which is an important family of problems that includes various coloring and orientation problems on finite-degree graphs. 
These are constraint-satisfaction graph problems that can be defined with a finite set of valid input/output-labeled neighborhoods.

	Our main results are:
	\begin{enumerate}
		\item 
		Any \emph{locally checkable labeling problem} can be solved with only $1$ bit of advice per node in graphs with \emph{sub-exponential growth} (the number of nodes within radius $r$ is sub-exponential in $r$; for example, grids are such graphs). Moreover, we can make the set of nodes that carry advice bits arbitrarily sparse. As a corollary, any locally checkable labeling problem admits a \emph{locally checkable proof} with $1$ bit per node in graphs with sub-exponential growth.
		\item The assumption of sub-exponential growth is complemented by a conditional lower bound: assuming the \emph{Exponential-Time Hypothesis}, there are locally checkable labeling problems that cannot be solved in general with any constant number of bits per node.
		\item In any graph we can find an \emph{almost-balanced orientation} (indegrees and outdegrees differ by at most one) with $1$ bit of advice per node, and again we can make the advice arbitrarily sparse. As a corollary, we can also \emph{compress an arbitrary subset of edges} so that a node of degree $d$ stores only $d/2 + 2$ bits, and we can \emph{decompress} it locally, in $T(\Delta)$ rounds.
		\item In any graph of maximum degree $\Delta$, we can find a $\Delta$-coloring (if it exists) with $1$ bit of advice per node, and again, we can make the advice arbitrarily sparse.
		\item In any $3$-colorable graph, we can find a $3$-coloring with $1$ bit of advice per node. As a corollary, in bounded-degree graphs there is a locally checkable proof that certifies $3$-colorability with $1$ bit of advice per node, while prior work shows that this is not possible with a \emph{proof labeling scheme} (PLS), which is a more restricted setting where the verifier can only see up to distance $1$.
	\end{enumerate}
	Our work shows that for many problems the key threshold is not whether we can achieve $1$ bit of advice per node, but whether we can make the advice arbitrarily sparse. To formalize this idea, we develop a general framework of \emph{composable} schemas that enables us to build algorithms for local computation with advice in a modular fashion: once we have (1)~a schema for solving $\Pi_1$ and (2)~a schema for solving $\Pi_2$ assuming an oracle for $\Pi_1$, we can also compose them and obtain (3)~a schema that solves $\Pi_2$ without the oracle. It turns out that many natural problems admit composable schemas, all of them can be solved with only $1$ bit of advice, and we can make the advice arbitrarily sparse.
\end{abstract}

\section{Introduction}

Our work explores \emph{what can and cannot be computed locally with the help of advice}. Our main focus is understanding advice in the context of classic local graph problems, such as vertex coloring.

While computation with different forms of advice has been explored in a wide range of distributed settings \cite{fraigniaud2007mst-advice,fraigniaud2010advice,dobrev2012exploration,glacet2017leader-advice,dereniowski2012maps-advice,gorain2018exploration-advice,miller2015rendezvous-advice,miller2016election-advice,fusco2011tradeoffs-advice,ilcinkas2010broadcasting-advice,nisse2009searching-advice,fraigniaud2008tree-advice,komm2015treasure-advice,miller2015treasure-hunt,fusco2016topology-advice}, there is hardly any prior work on solving classic local graph problems. A rare example of prior work is \cite{distr-comp-with-advice} from 2007, which studied the question of how much advice is necessary to break Linial's \cite{Linial92} lower bound for coloring cycles.

We initiate a systematic study of exactly how much advice is needed in the context of a wide range of graph problems. As we will see in this work, the exploration of the advice complexity of graph problems opens up connections with many other topics---it is linked with \emph{distributed proofs} \cite{korman06distributed,korman07distributed,korman10proof,korman10constructing,goos16lcp,feuilloley2017survey}, \emph{distributed decompression}, the notion of \emph{order-invariant algorithms} \cite{what-can-be-computed-locally}, and also with the \emph{exponential-time hypothesis} \cite{eth} in computational complexity theory.

\subsection{Local computation with advice}

Let us first formalize the setting we study:
\begin{framed}
	\noindent
	A graph problem $\Pi$ can be \emph{solved with $\beta$ bits of advice} if there exists a $T(\Delta)$-round distributed algorithm $\alg$, such that for any graph $G$ that admits a solution to $\Pi$, there is an assignment of $\beta$-bit labels on vertices, such that the output of $\alg$ on the labeled graph is a solution to $\Pi$.
\end{framed}
\noindent
We note that on graphs that do not admit a solution to $\Pi$, Algorithm $\alg$ is allowed to behave arbitrarily.
For instance, when we consider the $3$-coloring problem, Algorithm $\alg$ may produce an arbitrary output on graphs that are not $3$-colorable.
 
We will work in the usual LOCAL model of distributed computing. In an $n$-node graph, the nodes are labeled with unique identifiers from $\{1, 2, \dotsc, \poly(n)\}$. We emphasize that the advice may depend on the assignment of identifiers, and algorithm $\alg$ can freely make use of both the advice and the identifiers.

Now we seek to understand this question:
\begin{framed}
	\noindent
	What is the smallest $\beta$ such that $\Pi$ can be solved with $\beta$ bits of advice?
\end{framed}
\noindent
Note that if $\Pi$ is, for example, the $3$-coloring problem, it is trivial to solve with $\beta = 2$ bits of advice per node, as we can directly encode the solution. The key question is how much better we can do.

Our work shows that for many problems the key threshold is not whether we can achieve $1$ bit of advice per node, but whether we can make the advice \emph{arbitrarily sparse}, that is, make the ratio between $1$s and $0$s assigned to the nodes of the graph to be an arbitrarily small constant. This is particularly useful, as it enables us to \emph{compose} multiple sparse advice schemes, so that it suffices to use just one bit per node in total (we will elaborate on this in \cref{ssec:intro-composability}). Hence, a large part of this paper addresses the following question:
\begin{framed}
	\noindent
	Which problems admit arbitrarily sparse advice?
\end{framed}
The notion of the sparsity of the advice is discussed in more detail later in the paper. In this context, our paper shows that some problems can be solved with arbitrarily sparse advice. On the other hand, we also show that assuming the Exponential-Time Hypothesis, for any constant $c$, there exist problems that cannot be solved with $c$ bits of advice. Finally, there are some problems like 3-coloring that can be solved with 1 bit of advice, but where it is not clear whether it can be solved with arbitrarily sparse advice.
We discuss all of those points in more detail in the remaining part of this section.

\subsection{Contribution 1: LCLs in bounded-growth graphs}

Locally checkable labeling problems (LCL), first introduced by \citet*{what-can-be-computed-locally} in the 1990s, are one of the most extensively studied families of problems in the theory of distributed graph algorithms. These are graph problems that can be specified by giving a \emph{finite} set of valid local neighborhoods. Many key problems such as vertex coloring, edge coloring, maximal independent set, maximal matching, sinkless orientation, and many other splitting and orientation problems are examples of LCLs, at least when restricted to bounded-degree graphs. Thanks to the extensive research effort since 2016, we now understand very well the landscape of \emph{all} LCL problems and their computational complexities across different models of distributed computing \cite{balliu18lcl-complexity,balliu20almost-global,Chang2019,Ghaffari2018,balliu20lcl-randomness,fischer17sublogarithmic,Rozhon2019,brandt16lll,chang19exponential,ghaffari17distributed,balliu21lcl-congest}.

In \cref{ssec:subexp-growth}, we design a schema that allows us to solve \emph{any} LCL problem with just one bit of advice in graphs with a sub-exponential growth (the number of nodes in a radius-$r$ neighborhood is sub-exponential in $r$):
\begin{framed}
	\noindent
	Any LCL problem can be solved with $1$ bit of advice per node in sub-exponential growth graphs.
\end{framed}
\noindent
Furthermore, we show that the encoding can be made arbitrarily sparse. Note that e.g.\ grids have polynomial growth while e.g.\ regular trees have exponential growth, so the result is applicable in grids but not in trees.

\subsubsection{Application: locally checkable proofs in bounded-growth graphs}\label{sssec:intro-lcp-bounded-growth}

One prominent application of this result is its connection with \emph{distributed proofs} \cite{korman06distributed,korman07distributed,korman10proof,korman10constructing,feuilloley2017survey}, and in particular with \emph{locally checkable proofs} \cite{goos16lcp}. Consider any LCL $\Pi$. Assume that our task is to prepare a distributed proof that shows that in a graph $G$ there exists a feasible solution of $\Pi$ (for example, if $\Pi$ is the task of 10-coloring, then the task is to certify that the chromatic number of $G$ is at most 10). Now if $G$ has sub-exponential growth, we can use our result from \cref{ssec:subexp-growth} to prepare a $1$-bit advice that enables the algorithm to find a solution of $\Pi$. Our advice is the proof: to verify it, we simply try to recover a solution with the help of the advice, and then check that the output is feasible in all local neighborhoods (recall that $\Pi$ is locally checkable). We obtain the following corollary:
\begin{framed}
	\noindent
	Any LCL problem admits a locally checkable proof with $1$ bit per node in graphs with sub-exponential growth.
\end{framed}
\noindent
Note that this is not a proof labeling scheme as defined in \cite{korman06distributed,korman07distributed,korman10proof,korman10constructing}, as the verifier running at node $u$ may need to see more than just the identifier and the proof label of $u$ and the proof labels of $u$'s immediate neighbors. However, in bounded-degree graphs it is a locally checkable proof (LCP) as defined in \cite{goos16lcp}; for a fixed $\Delta$ the verification radius is a constant $T(\Delta)$.

So to summarize, if we can solve some LCL problem $\Pi$ with $b$ bits of advice per node, then we also have an LCP for the graph property ``G admits a feasible solution to $\Pi$'' with $b$-bit proofs per node.
The converse is not true: Consider the LCL problem $\Pi$ that encodes the task ``orient edges so that each node has indegree equal to outdegree''. Now to prove ``$G$ admits a feasible solution to $\Pi$'' one can use a $0$-bit LCP, where even-degree nodes accept and odd-degree nodes reject. However, this LCP does not help us at all if we would like to solve $\Pi$ with the help of advice. In this sense distributed computation with advice is a harder problem than local proofs.

It is also good to note that one can have LCPs for graph properties that are not of the form ``G admits a feasible solution to some LCL $\Pi$.'' For example, planarity is such a property. Such LCPs are (to our knowledge) not directly connected with computation with local advice.

\subsection{Contribution 2: LCLs in general graphs}

At this point a natural question is whether the assumption about bounded growth is necessary. Could we solve all LCL problems in all graphs with $1$ bit of advice? In \cref{sec:structural} we show that the answer is likely to be no:
\begin{framed}
	\noindent
	Fix any $\beta$. If all LCL problems can be solved locally with at most $\beta$ bits of advice, then the Exponential-Time Hypothesis (ETH) is false.
\end{framed}

The intuition here is that if some LCL problem $\Pi$ can be solved with, say, $1$ bit of advice per node with some local algorithm $\alg$, then we could solve it with a centralized sequential algorithm as follows: check all $2^n$ possible assignments of advice, apply $\alg$ to decode the advice, and see if the solution is feasible. The total running time (from the centralized sequential perspective) would be $2^n \cdot n \cdot s(n)$, where $s(n)$ is the time we need to simulate $\alg$ at one node. Then we need to show that assuming the Exponential-Time Hypothesis, this is too fast for some LCL problem $\Pi$. However, the key obstacle is that it may be computationally expensive to simulate $\alg$, as it might perform arbitrarily complicated calculations that depend on the numerical values of the unique identifiers, and we cannot directly bound $s(n)$.

Hence, we need to show that $\alg$ can be made cheap to simulate. The key ingredient is the following technical result, which we prove using a Ramsey-type argument that is inspired by the proof of \citet*{what-can-be-computed-locally}:
\begin{framed}
	\noindent
	Assume that problem $\Pi$ can be solved with $\beta$ bits of advice per node, using some algorithm $\alg$. Then the same problem can be also solved with $\beta$ bits of advice using an \emph{order-invariant} algorithm $\alg'$, whose output does not depend on the numerical values of the identifiers but only on their relative order.
\end{framed}
\noindent
The key point here is that (for bounded-degree graphs) $\alg'$ can be represented as a finite lookup table; hence the simulation of $\alg'$ is cheap, and we can finally make a formal connection to the Exponential-Time Hypothesis. We refer to \cref{sec:structural} for more details.

\subsection{Contribution 3: balanced orientations}

In \cref{ssec:balanced-orientation} we move on to a specific graph problem: we study the task of finding balanced and almost-balanced orientations. The goal is to orient the edges so that for each node indegree and outdegree differ by at most $1$. This is a hard problem to solve in a distributed setting, while slightly more relaxed versions of the problem admit efficient (but not constant-time) algorithms \cite{degree-splitting}.

Here it is good to note that if we could place our advice on \emph{edges}, then trivially one bit of advice per edge would suffice (simply use the single bit to encode whether the edge is oriented from lower to higher identifier). However, we are here placing advice on \emph{nodes}, and encoding the orientation of each incident edge would require a number of bits proportional to the maximum degree. Surprisingly, we can do it, in any graph:
\begin{framed}
	\noindent
	We can find almost-balanced orientations with $1$ bit of advice per node.
\end{framed}
\noindent
Again, we can make the advice arbitrarily sparse.

\subsubsection{Application: distributed decompression}

Equipped with the advice schema for solving almost-balanced orientations, we can now make a formal connection to what we call \emph{distributed decompression}. Here the task is to encode some graph labeling so that it can be decompressed locally (in $T(\Delta)$ rounds).

Local decompression is closely linked with local computation with advice. If we can compress some solution to $\Pi$ with only $\beta$ bits per node, and decompress it locally, then we can also solve $\Pi$ with $\beta$ bits of advice per node. Furthermore, if $\Pi$ is a problem such that for any graph there is only one feasible solution, then the two notions coincide.

We will now show yet another connection between local decompression and local computation with advice. Consider the task of compressing an \emph{arbitrary subset of edges} $X \subseteq E$. In a trivial encoding, we label each node $v$ of degree $d$ with a $d$-bit string that indicates which of the incident edges are present in $X$. On the other hand, we need a total of $|E|$ bits in order to distinguish all subsets of the edge-set $E$. In particular, for $d$-regular graphs, this means we need at least $d/2$ bits per node to recover an arbitrary subset of edges.

It turns out that once we can solve almost-balanced orientations, we can also compress a subset of edges efficiently. We simply use $1$ bit of advice per node to encode an almost-balanced orientation. Now a node of degree $d$ has outdegree $\delta \le \lceil d/2 \rceil$, and it can simply store a $\delta$-bit vector that indicates which of its outgoing edges are in $X$. Overall, we will need $\lceil d/2 \rceil +1$, i.e. $\le d/2 + 2$, bits per node:
\begin{framed}
	\noindent
	We can encode an arbitrary set of edges $X \subseteq E$ so that a node of degree $d$ only needs to store $\lceil d/2 \rceil + 1$ bits, and we can decompress $X$ locally, in $T(\Delta)$ rounds.
\end{framed}

\subsection{Contribution 4: vertex \texorpdfstring{\boldmath $\Delta$}{Delta}-coloring}

In \cref{ssec:delta-col}, we study the problem of
$\Delta$-coloring
graphs of maximum degree $\Delta$:
\begin{framed}
	\noindent
	In any graph of maximum degree $\Delta$, we can find a $\Delta$-coloring (if it exists) with $1$ bit of advice per node.
\end{framed}
\noindent
Again, we can make the advice arbitrarily sparse.

Our schema for encoding $\Delta$-colorings consists of three steps. First, we compute a vertex coloring with $O(\Delta^2)$ colors, with the help of advice. Then we reduce the number of colors down to $\Delta + 1$, using the algorithm by \cite{FHK16,BarenboimEG18,MausT22}. Finally, we follow the key idea of the algorithm by \citet*{PS92} to turn ${(\Delta+1)}$-coloring into a $\Delta$-coloring, and again we will need some advice to make this part efficient.

\subsection{Contribution 5: vertex 3-coloring}

So far we have seen primarily results of two flavors: many problems can be solved with $1$ bit of advice so that we can make the advice arbitrarily sparse, while there are also some problems that require arbitrarily many bits of advice.

We now turn our attention to a problem that seems to lie right at the boundary of what can be done with only $1$ bit per node: vertex $3$-coloring in any $3$-colorable graph. Note that this is a problem that is hard to solve without advice not only in the distributed setting (it is a global problem) but also in the centralized setting (it is an NP-hard problem).

In the centralized setting, $1$ bit of advice per node makes the problem easy. To see this, we can simply use the bit to indicate which nodes are of color $3$. Then the rest of the graph has to be bipartite, and we can simply find a proper $2$-coloring in polynomial time.

In the distributed setting, the trivial solution does not work: $2$-coloring in bipartite graphs is still a global problem. Nevertheless, in \cref{ssec:3-col-3-col} we show that $3$-coloring is still doable with $1$ bit of advice:
\begin{framed}
	\noindent
	In any $3$-colorable graph, we can find a $3$-coloring with $1$ bit of advice per node.
\end{framed}
Our encoding essentially uses one bit to encode one of the color classes, but we adjust the encoding slightly so that throughout the graph there are \emph{local hints} that help us to also choose the right parity for the region that we need to $2$-color.

Here, our encoding genuinely needs one bit per node (it just barely suffices); we cannot make our advice arbitrarily sparse.

\subsubsection{Application: locally checkable proofs for 3-coloring}

Following the same idea as in \cref{sssec:intro-lcp-bounded-growth}, we can now also certify that a graph is $3$-colorable with a proof that only takes $1$ bit per node. Furthermore, the verifier only needs to see up to distance $T(\Delta)$, and hence if $\Delta = O(1)$, this is a locally checkable proof in the sense of \cite{goos16lcp}:
\begin{framed}
	\noindent
	There is a locally checkable proof that certifies $3$-colorability with $1$ bit per node in graphs of maximum degree $\Delta = O(1)$.
\end{framed}
Now it is interesting to compare this with \emph{proof labeling schemes} (PLS). In essence, a PLS is an LCP in which the verifier only sees the identifier of the present node, and the proof labels within radius $1$. The above result provides a $1$-bit LCP but not a $1$-bit PLS.

This connects directly with the work done by \citet*{lower-bound-constant-size-local-cert} and \citet*{BFZ24}, which study the same question: how many bits per node are needed to certify $k$-colorability. Here \citet{BFZ24} shows that $\Theta(\log k)$ bits per node are necessary for a PLS that certifies $k$-colorability, and \cite{lower-bound-constant-size-local-cert} shows that in particular $3$-colorability cannot be certified with $1$ bit per node with any PLS (while $2$ bits per node is trivial). Our work complements the latter result, and provides a separation between PLSs and LCPs in this setting: now we know that while $1$ bit per node does not suffice to certify $3$-colorability with any PLS, it is sufficient for an LCP (with the caveat that we need to be in a bounded-degree graph).

\subsection{Key technique: composability framework}\label{ssec:intro-composability}

We already discussed some of the proof ingredients above. However, there is one additional technique that we use in many of our algorithms: the framework of \emph{composable schemas}.

It turns out that for many problems, it is easier to work with advice schemas in which only a few nodes carry advice bits, but they may carry many bits of advice. In \Cref{def:advice-schema} we give the formal definition of such a schema, and in \Cref{def:composable} we give the formal definition of \emph{composable} schemas, which satisfy the additional property that the ratio between the total number of bits held by the nodes, and the total number of nodes, can be made arbitrarily small in every large-enough neighborhood.

While the definition is a bit technical, it has two key properties, which we discuss in more detail in \cref{app:composability}:
\begin{enumerate}
	\item As the name suggests, composable schemas can be easily composed, in the following sense: once we have (1)~a composable schema for solving $\Pi_1$ and (2)~a composable schema for solving $\Pi_2$ assuming an oracle for $\Pi_1$, we can also compose them and obtain (3)~a composable schema that solves $\Pi_2$ without the oracle. This way we can solve problems in a modular fashion, in essence using schemas as ``subroutines.''
	\item A composable schema can be then encoded with only $1$ bit of advice per node, and we can make the advice arbitrarily sparse.
\end{enumerate}
For example, our algorithms for finding almost-balanced orientations and $\Delta$-coloring with advice are based on the framework of composable schemas.

\subsection{Open questions}

Our work suggests a number of open questions:
\begin{enumerate}
	\item Our negative result in \cref{sec:structural} assumes the Exponential-Time Hypothesis. Is it possible to prove an unconditional lower bound without such assumptions? Can we exhibit a concrete LCL problem $\Pi$ that is unconditionally hard?
	\item We conjecture that the advice for vertex $3$-coloring in $3$-colorable graphs cannot be made arbitrarily sparse. Is this true?
	\item We also conjecture that for a sufficiently large $d$, vertex $d$-coloring in $d$-colorable graphs cannot be encoded with one bit per node. Is this true? This is closely linked with the discussion on the trade-off conjecture in e.g.\ \cite{lower-bound-constant-size-local-cert,BFZ24}.
	\item Our schema for distributed decompression is asymptotically optimal, but there is room for improving additive constants. Here is a concrete open question: Let $G = (V,E)$ be a $3$-regular graph, and let $X \subseteq E$ be an arbitrary set of edges. Is it possible to encode $X$ using only $2$ bits per node so that it can be decompressed locally? (Note that $1$ bit per node is trivially impossible, while $3$ bits per node is trivial. 
	If we delete one edge from each connected component, an encoding with $2$ bits per node follows from $2$-degeneracy.)
\end{enumerate}

\section{Additional related work}

\emph{Computing with advice} is not really a well-defined area of research, as one can consider many kinds of advice and various models of computations.
In this paper, we focus on existential advice and the distributed LOCAL model: the advice is provided by an all-knowing oracle, and it is designed to enable algorithms with low locality (i.e., low time complexity). 
More broadly, one can consider a variant in which the restriction on the size of advice is more relaxed, but it needs to be computable by an oracle with some limitations (realizable in some model of computation), which then can be used as a building block of algorithms. 

\subparagraph{Advice and locality in distributed computing.}
\citet*{redundancy-in-distr-proofs} consider the framework of Proof Labeling Schemes (PLS). 
In this framework, 
there is a \emph{prover} and a \emph{verifier}: the prover is a centralized entity that assigns labels to the nodes, while the verifier is a distributed algorithm that in $T$ rounds is able to verify the validity of the collection of labels.
If the predicate we want to investigate holds true, then there must exist an assignment of labels (or certificates) such that each node, in $T$ rounds, ``accepts'' the given labeling;
while if the predicate we want to investigate does not hold, then for any given labeling assignment there must exist a node that rejects it. 
The authors study the tradeoff between the size of the labels given at each node and the round-complexity of the verifier.
In their work, they, too, discover how this tradeoff is remarkably different depending on the growth rate of the graph. 

\citet*{distr-comp-with-advice} investigate the problem of distributedly $3$-coloring a cycle. Without advice, it is known that this problem requires $\Omega(\log^* n)$ rounds \cite{Linial92}, and the paper seeks to understand whether advice can help to break this barrier. In the context studied in \cite{distr-comp-with-advice}, each node receives some advice as input, and the total advice is measured as the sum of the lengths of the bit-strings given to all nodes. The authors show that, for any constant $k$, $O(n / \log^{(k)} n)$ bits of total advice are not sufficient to beat the $\Omega(\log^* n)$ lower bound, where $\log^{(k)} n$ denotes $k$ recursive iterations of $\log n$.

\subparagraph{Algorithm design.}
\citet*{PARNAS2007183} show that one can use distributed algorithms to derive (by simulation) Local Computation Algorithms. Quite naturally, one can use the same approach to derive parallel graph algorithms or graph algorithms for dynamic data sets (or graph algorithms for any model of computation that can leverage the fact that the output is defined by some small neighborhood of a vertex or edge). However, as is, the reduction from \cite{PARNAS2007183} did not immediately imply many new state-of-the-art algorithms, as the size of the neighborhood grows exponentially with locality, and, for many problems, the locality of the fastest-known algorithms is fairly large. One of the ways to address this issue is to rely on analysis of algorithms with advice. 

\citet*{ChristiansenNR23} use distributed algorithms with advice to design faster dynamic algorithms for vertex coloring, by combining dynamic graph orientations with ideas for simulating distributed algorithms on the resulting directed graph. 

While this approach in principle can be used in the design of parallel algorithms, so far it is used implicitly or in a somewhat trivial way (e.g.\ \citet*{ghaffari_et_al:LIPIcs.DISC.2020.34} mention that 4-coloring of a graph can be trivially used to compute a maximal independent set).

\subparagraph{Advice in other models of computation with partial knowledge.}
More generally, one can consider the impact of advice that captures some global knowledge about the whole input (or distribution of inputs) on the performance of algorithms. Similarly to the case of distributed computing, one can consider existential advice and advice that can be realized in a somewhat practical model of computation.

\citet*{EmekFKR09} introduced the notion of advice in the context of online algorithms. The authors studied the impact of advice in the competitive ratio of some classical online problems, namely \emph{metrical task systems} and \emph{$k$-server}, and they show tradeoffs between the number of bits of advice and the achievable competitive ratio. We refer to the survey by \citet*{BoyarFKLMsurvey} for more work on online algorithms with advice.

\citet*{DBLP:books/cu/20/MitzenmacherV20} show that algorithms using advice realized by machine learning models can be used to provide algorithms that on average perform better than their traditional counterparts while keeping the same worst-case bounds. However, this particular advice model allows querying an oracle at each step. As such, in the context of online algorithms, it is in some sense both weaker than the existential advice for online algorithms (as it is produced by some machine learning model) and stronger as it gives a lot of bits of advice, and has access to the prefix of the input.

\section{Preliminaries}

Let $G = (V,E)$ be a graph, where $V$ is the set of nodes and $E$ is the set of edges. We denote with $n$ the number of nodes of the graph, and with $\Delta$ its maximum degree. We may use the notation $V(G)$ to denote the set of nodes of $G$, and $E(G)$ to denote the set of edges of $G$. With $\dist_G(u,v)$ we denote the distance between $u \in V$ and $v \in V$ in $G$, that is, the length of the shortest path between $u$ and $v$ in $G$, where the length of a path is the number of edges of the path. For an integer $k$, the power graph $G^k$ of a graph $G$ is the graph that contains the same nodes as $G$, and there is an edge $\{u,v\}$ in $G^k$ if and only if $1 \le \dist_G(u,v) \le k$. Two nodes $u$ and $v$ are neighbors in $G$ if there is an edge $\{u,v\} \in E(G)$.

A maximal independent set (MIS) of $G$ is a subset $S$ of nodes of $G$ satisfying that no nodes of $S$ are neighbors in $G$, and that all nodes in $V \setminus S$ have at least one neighbor that is in $S$.
An $(\alpha,\beta)$-ruling set is a subset $S$ of nodes of $G$ satisfying that each pair of nodes from $S$ have distance at least $\alpha$, and that all nodes in $V \setminus S$ have at least one node in $S$ at distance at most $\beta$. Observe that \MIS{} and $(2,1)$-ruling set is the same problem. By $\beta$-ruling set we denote the $(2,\beta)$-ruling set problem.

One important tool that we will use is the Lovász Local Lemma, which states the following.

\begin{lemma}[Lovász Local Lemma \cite{shearer,SPENCER197769,lll}]\label{lem:lll}
	Let $\{E_1,\ldots, E_k\}$ be a set of events satisfying the following conditions: each event $E_i$ depends on at most $d$ other events; each event $E_i$ happens with probability at most $p$. Then, if $e p d \le 1$, there is non-zero probability that none of the events occur.
\end{lemma}

\subsection{LOCAL model}
In the LOCAL model of distributed computation, each node of an $n$-node graph is equipped with a unique ID, typically in $\{1,\ldots,n^c\}$, for some constant $c$. Initially, each node in the graph knows its own ID, its degree (i.e., the number of neighboring nodes), the maximum degree $\Delta$ in the graph, and the total number of nodes. Then, the computation proceeds in synchronous rounds where at each round each node exchanges messages with its neighbors and performs some local computation. The size of the messages can be arbitrarily large and the local computation can be arbitrarily heavy. The runtime of an algorithm is defined as the number of rounds it requires such that all nodes terminate and produce an output. 

\subsection{Locally checkable labelings}
We formally define an LCL problem $\Pi$ as a tuple  $(\Sigma_{\mathrm{in}},\Sigma_{\mathrm{out}},C, r)$ where each element of the tuple is defined as follows. The parameters $\Sigma_{\mathrm{in}}$ and $\Sigma_{\mathrm{out}}$ are finite sets of input and output labels, respectively. The parameter $r$ is a positive integer called \emph{checkability radius} of $\Pi$, i.e., it determines how far in the graph each node needs to check in order to verify the validity of a given solution. The parameter $C$ determines the constraints of $\Pi$, that is, $C$ is a finite set of labeled graphs $H$ containing a vertex of eccentricity at most $r$, where
each edge-endpoint pair $(uv,v)\in E_H\times V_H$ 
has a label $\ell_{\mathrm{in}}\in \Sigma_{\mathrm{in}}$ and a label $\ell_{\mathrm{out}}\in \Sigma_{\mathrm{out}}$.

Let $\Pi=(\Sigma_{\mathrm{in}},\Sigma_{\mathrm{out}},C, r)$ be an LCL problem and let $G=(V,E)$ be a graph where each edge-endpoint pair $(uv,v)$ is labeled with a label from $\Sigma_{\mathrm{in}}$. The task of solving $\Pi$ on $G$ requires labeling each edge-endpoint pair $(uv,v)\in E\times V$ with a label in $\Sigma_{\mathrm{out}}$ such that, for each node $v\in V$ it holds that the graph induced by the nodes at distance at most $r$ from $v$ and edges that have at least one endpoint at distance at most $r-1$ from $v$ is isomorphic to some (labeled) graph in $C$.

\subsection{Advice schema}
We now formally define the notion of advice schema.
\begin{definition}[Advice Schema]\label{def:advice-schema}
	A $(\mathcal{G},\Pi,\beta,T)$-advice schema is a function $f$ that receives as input a (possibly input-labeled) graph $G = (V,E,I) \in \mathcal{G}$ (where $I$ is an input for the nodes and/or the edges of the graph), and outputs a function $\ell_G := f(G)$ that satisfies the following. 
	\begin{itemize}
		\item The function $\ell_G$ maps each node $v \in V$ into a bit-string of length at most $\beta$, where $\beta$ is a function of $\Delta$.
		\item There exists a LOCAL algorithm $\mathcal{A}$ that, for each $G \in \mathcal{G}$, if we label each node $v \in V$ with the bit-string $\ell_G(v)$, then $\mathcal{A}$ runs in at most $T$ rounds and outputs a valid solution for $\Pi$, where $T$ is a function of $\Delta$.
	\end{itemize}
\end{definition}
Moreover, we distinguish three possible types of advice schemas:
\begin{enumerate}
	\item If all nodes of the graph receive bit-strings of the same length, then the schema is called \emph{uniform fixed-length}.
	\item If a subset of nodes receives bit-strings of the same length, and all the others receive bit-strings of length $0$, then the schema is called \emph{subset fixed-length}.
	\item If a subset of nodes receives a bit-string of possibly different positive lengths, and all the others receive bit-strings of length $0$, then the schema is called \emph{variable-length}. Observe that an advice schema, as defined in \Cref{def:advice-schema} and without further assumptions, is variable-length.
\end{enumerate}
In the second and third cases, the nodes having non-zero bit-strings assigned are called \emph{bit-holding} nodes. Moreover, observe that Type 1 is a special case of Type 2, which is a special case of Type 3.

For uniform fixed-length advice schemas where all nodes receive one bit, we define the notion of sparsity, which captures the fact that the ratio between nodes receiving a $1$ and nodes receiving a $0$ can be made arbitrarily small.
\begin{definition}[Sparse schema]
A uniform fixed-length $(\mathcal{G},\Pi,\beta,T)$-advice schema is called $\varepsilon$-\emph{sparse} if the following holds:
\begin{itemize}
	\item $\beta = 1$, that is, each node receives exactly one bit.
	\item For any graph $G \in \mathcal{G}$, let $n_0$ be the number of nodes to which the schema (i.e., the function $f$ of \Cref{def:advice-schema}) assigns a $0$, and let $n_1$ be the number of nodes to which $f$ assigns a $1$. Then, $\frac{n_1}{n_0 + n_1} \le \varepsilon$.
\end{itemize}
With abuse of notation, we call a uniform fixed-length $(\mathcal{G},\Pi,\beta,T)$-advice schema (where $T$ is a function that receives as input $\varepsilon$ and returns a function of $\Delta$) \emph{sparse} if, for any constant $\varepsilon > 0$, there exists an $\varepsilon$-sparse  $(\mathcal{G},\Pi,\beta,T(\varepsilon))$-advice schema.
\end{definition}

\subsection{Composability}
The main idea that we will use to devise our advice schemas is the following. Given some problem $\Pi$, it may be cumbersome to directly define an advice schema for it. Instead, it may be easier to do this operation gradually. In more detail, many problems can be solved as follows: first, find a solution for some subproblem; then, use the solution for the subproblem in order to solve the problem more easily. As a running example, consider the following problem $\Pi$.
\begin{itemize}
	\item The input is a bipartite graph where all nodes have even degree.
	\item It is required to output a coloring of the edges, say red and blue, such that each node has the same number of red and blue incident edges.
\end{itemize}
Consider the following three problems.
\begin{itemize}
	\item Let $\Pi_\mathrm{v}$ be the problem of computing a $2$-coloring of the nodes, say, black and white.
	\item Let $\Pi_\mathrm{e}$ be the problem of outputting a $2$-coloring of the edges, such that each node has the same number of red and blue incident edges, assuming that we are given as input a $2$-coloring of the nodes and a balanced orientation of the edges, and assuming that all nodes have even degree.
	\item Let $\Pi_\mathrm{o}$ be the problem of orienting all edges such that each node has the same number of incoming and outgoing edges, assuming that all nodes have even degree.
\end{itemize}
Consider the following algorithm. First, solve $\Pi_\mathrm{v}$ and $\Pi_\mathrm{o}$. Then, solve $\Pi_\mathrm{e}$, by coloring \emph{red} the edges oriented from black to white, and by coloring \emph{blue} the edges oriented from white to black. Observe that this algorithm solves $\Pi$.

In other words, we decomposed $\Pi$ into three different subproblems: two of them are ``hard'' problems ($\Pi_\mathrm{v}$ and $\Pi_\mathrm{o}$), for which some advice is needed if we want to solve them fast, and one ($\Pi_\mathrm{e}$) is trivial once we are given as input the solution for the other two. The idea now is to devise two different advice schemas for $\Pi_\mathrm{v}$ and $\Pi_\mathrm{o}$ separately.
\begin{itemize}
	\item For $\Pi_\mathrm{v}$ there is a trivial advice schema:  use $1$ bit to encode the $2$-coloring of the nodes.
	\item For $\Pi_\mathrm{o}$ the advice schema is non-trivial, and we will see how it can be done in \Cref{ssec:balanced-orientation} by using $1$ bit per node.
\end{itemize}
While both schemas use only $1$ bit per node, we cannot directly combine these two schemas into a single schema that uses $1$ bit per node. For this reason, we will devise our schemas in a more high-level way. In particular, for most of the results of this paper, we will not directly provide a uniform fixed-length advice schema that uses one bit per node, but we will start by devising variable-length schemas. Variable-length schemas provide a simpler way to encode information: we can choose a subset of nodes and assign bit-strings to them, without worrying about how these strings will be later encoded by giving one bit to each node. Moreover, the variable-length schemas that we provide satisfy a property called \emph{composability}, which we will formally define later. Intuitively, this property requires the ability to make the set of bit-holding nodes of a variable-length schema to be arbitrarily sparse. For example, for $\Pi_\mathrm{v}$, we do not really need to provide the color of all nodes: we assign $1$ bit to a sparse set of nodes (encoding their color), and to all other nodes we do not assign any bit. The nodes that have no bit assigned can still recover a $2$-coloring by simple propagation. For $\Pi_\mathrm{o}$, things are more complicated, but a composable variable-length schema can be obtained as well.

Once we have a composable schema for all subproblems that are used to solve our problem of interest, we apply, as a black box, a lemma (see \Cref{lem:compose}) to obtain a single variable-length schema for our problem. Then, again as a black box, we convert such a schema into a uniform fixed-length schema that uses a single bit per node (see \Cref{lem:composable-to-1bit}).

Summarizing, many of the results that we provide about schemas that use a single bit per node are based on the following idea:
\begin{itemize}
	\item We first decompose a problem of interest into many subproblems;
	\item We show that, for each of the subproblems, it is possible to devise a variable-length schema that satisfies some desirable properties;
	\item We prove that such properties imply that we can compose many variable-length schemas into a single one that satisfies the same properties (see \Cref{lem:compose});
	\item We prove that the resulting variable-length schema implies a uniform fixed-length schema that uses a single bit per node (see \Cref{lem:composable-to-1bit}).
\end{itemize}
Note that the last two steps are done in a problem-independent way. The conditions that allow to combine multiple variable-length schemas are expressed in \Cref{def:composable}, that, on a high level, states the following. We want our variable-length schema to be tunable as a function of three parameters $\alpha$, $\gamma$, and $c$. The schema needs to satisfy that in each $\alpha$-radius neighborhood there are at most $\gamma$ bit-holding nodes, and that the number of bits held by these nodes is upper bounded by $c \alpha / \gamma^3$. Hence, a composable schema is a collection of schemas, one for each choice of parameters $\alpha$, $\gamma$, and~$c$.
\begin{definition}[Composability]\label{def:composable}
	A $(\mathcal{G},\Pi,\gamma_0, A, T)$-composable advice schema is a collection $\mathcal{S}$ of advice schemas satisfying the following.
	For any constant $c > 0$, any $\gamma \ge \gamma_0$, and any $\alpha \ge A(c,\gamma)$, there exists $\beta \le c \alpha / \gamma^3$ such that:
	\begin{itemize}
		\item The collection $\mathcal{S}$ contains a variable-length $(\mathcal{G},\Pi,\beta,T(\alpha,\Delta))$-advice schema $S$.
		\item For each $G \in \mathcal{G}$, the assignment given by $S$ to the nodes of $G$ satisfies that, in each $\alpha$-radius neighborhood of $G$, there are at most $\gamma_0$ bit-holding nodes.
	\end{itemize}
\end{definition}

\section{Structural results and lower bounds}\label{sec:structural}

Our main goal here is to show that if we were able to solve all LCLs with some constant number of bits of advice per node, it would violate the \emph{Exponential-Time Hypothesis} (ETH) \cite{eth}. However, to do that we need to first show a structural result, which is also of independent interest: algorithms that make use of advice can be made \emph{order-invariant}, i.e., they do not need the numerical values of the identifiers.

\subsection{Order-invariance}\label{ssec:order-invariance}

In general, in an advice schema our advice bits may depend on the numerical values of the identifiers, and the algorithm $\alg$ that solves the problem may make use of the advice bits and the numerical values of the identifiers. However, here we will show that we can essentially for free eliminate the dependency on the numerical values of the identifiers.

A distributed algorithm is called \emph{order-invariant} \cite{what-can-be-computed-locally} if the output remains the same if we change the numerical values of the identifiers but preserve their relative order. Put otherwise, an order-invariant algorithm may make use of the graph structure, local inputs (which in our case includes the advice), and the relative order of the identifiers, but not the numerical values of the identifiers.

\citet*{what-can-be-computed-locally} showed that constant-time distributed algorithms that solve LCL problems can be made order-invariant, and since then order-invariant algorithms have played a key role in many lower-bound results related to constant-time distributed algorithms, see e.g.\ \cite{goos13local-approximation,goos15mep-lb}.

We say that graph problem $\Pi$ is \emph{component-wise-defined} if the following holds: a valid solution for graph $G$ is also a valid solution for a connected component that is isomorphic to $G$. Put otherwise, we can always solve $\Pi$ by splitting the graph into connected components and solving $\Pi$ separately in each component. In essence all graph problems of interest (especially in the distributed setting) are component-wise defined, but one can come up with problems that do not satisfy this property (a trivial example being the task of outputting the number of nodes in the graph).

We say that graph problem $\Pi$ is a \emph{finite-input} problem if the nodes and edges are either unlabeled, or they are labeled with labels from some finite set $\Sigma_{\mathrm{in}}$.

Note that Theorem~\ref{thm:order-invariant} below is applicable to LCL problems (they are component-wise defined, and in any LCL we have a bounded $\Delta$ and bounded alphabets $\Sigma_{\mathrm{in}},\Sigma_{\mathrm{out}}$), but also to many other problems that are not locally checkable. We emphasize that we will assume some bound on the maximum degree $\Delta$, but if we have an encoding schema that works for any $\Delta$, we can apply the following result separately for each $\Delta$.
\begin{theorem}\label{thm:order-invariant}
	Fix a maximum degree $\Delta$ and the number of advice bits $\beta$. Assume that $\Pi$ is a finite-input component-wise defined graph problem, in which the task is to label nodes with labels from some finite set $\Sigma_{\mathrm{out}}$. Assume that we can solve some graph problem $\Pi$ with some distributed algorithm $\alg$ in $T$ rounds using $\beta$ bits of advice. Then we can also solve $\Pi$ with an order-invariant algorithm $\alg'$ in $T$ rounds using $\beta$ bits of advice.
\end{theorem}

\begin{proof}
	Let $E$ be the \emph{encoder} that, given a graph $G$ (with some unique identifiers), produces $\beta$-bit advice strings that $\alg$ can then use to solve $\Pi$.

	In the first steps, we follow the basic idea of \citet*{what-can-be-computed-locally} to manipulate $\alg$. 
	Algorithm $\alg$ is a mapping from labeled radius-$T$ neighborhoods to local outputs; we write here $N^T(v)$ for the radius-$T$ neighborhood of node $v$, and we let $s = \Delta^{T+1}$ be an upper bound on the number of nodes in $N^T(v)$.

	In general, we can encode all information in $N^T(v)$ as follows:
	\begin{enumerate}
		\item The \emph{set $X \subseteq \{ 1,2,\dotsc \}$ of unique identifiers} present in the neighborhood, with $|X| \le s$.
		\item The \emph{structure} $S$ of the neighborhood, which includes the graph topology, the relative order of the identifiers, local inputs, and the advice bits.
	\end{enumerate}
	So we can reinterpret $\alg$ as a function that maps a pair $(X, S)$ to the local output $\alg(X, S) \in \Sigma_{\mathrm{out}}$. But we can equally well interpret $\alg$ as a function $A$ that maps $X$ to a \emph{type} $A(X) = F$, where $F$ is a function that maps $S$ to the local output $F(S) \in \Sigma_{\mathrm{out}}$; we simply let $A(X)(S) = \alg(X, S)$. The second interpretation turns out to be convenient.

	Notice that we can always pad $X$ with additional identifiers (that use values larger than any value in $X$), so that we will have $|X| = s$.

	Then we make the key observation: \emph{there are only finitely many different structures} $S$. This follows from the assumptions that $\Delta$ is a constant, $\beta$ is a constant, $\Pi$ is a finite-input problem, and $T$ is a constant. Furthermore, $\Sigma_{\mathrm{out}}$ is finite, so $F$ is a mapping from a finite set to a finite set, and follows that \emph{there are only finitely many different types} $F$. Let $m$ be the number of possible types, and identify the types with $1, 2, \dots, m$.
	
	With this interpretation, $A$ defines a \emph{coloring} (in the Ramsey-theoretic sense) of all $s$-subsets of natural numbers with $m$ colors, that is, it assigns to each subset $X \subseteq \{1,2,\dotsc\}$ with $|X| = s$ some \emph{color} $A(X) \in \{1,2,\dotsc,m\}$. By applying the infinite multigraph version of Ramsey's theorem, it follows that there exists an infinite set of natural numbers $I$ and a single \emph{canonical type} $F^*$ such that for any $X \subseteq I$ we have $A(X) = F^*$. Set $I$ is called a \emph{monochromatic set}.

	Now let us stop for a moment and digest what we have learned so far: if our unique identifiers came from the set $I$, then we will have $\alg(X,S) = A(X)(S) = F^*(S)$. That is, $\alg$ will then ignore the numerical values of the identifiers and only pay attention to the structure $S$. However, this is a big if; in general our identifiers can be arbitrary, and even adversarial.

	Let us now continue; we will now modify the encoder $E$. We construct a new encoder $E'$ that works as follows. Assume we are given a graph $G$, together with some assignment of unique identifiers. The encoder $E'$ first constructs graph $G_1$ by renumbering the identifiers (but preserving their order) so that all identifiers in $G_1$ come from the monochromatic set $I$.

	Now we would like to apply encoder $E$ to $G_1$, but we cannot. Our encoder may assume that the identifiers in an $n$-node graph come from the set $\{1,2,\dotsc,\poly(n)\}$, while now we have made our identifiers astronomically large. But to fix this we exploit the fact that $\Pi$ is component-wise solvable. We simply construct a new graph $H$ that consists of sufficiently many copies of $G$, such that the first copy is $G_1$, with unique identifiers coming from $I$, while in the other copies we assign identifiers from $\{1,2,\dotsc\} \setminus I$. This way we can arrange things so that $H$ has $N$ nodes, for some (very large) number $N$, and the identifiers are assigned from $\{1,2,\dotsc,N\}$.

	Now $H$ is a valid instance, and we can feed it to the encoder $E$, which will label it with $\beta$-bit labels so that if we apply $\alg$ to $H$ and these labels, we will correctly solve $\Pi$ in each component of $H$. In particular, we will correctly solve $\Pi$ in $G_1$. Moreover, in component $G_1$, algorithm $\alg$ will apply order-invariant algorithm $F^*$ in all neighborhoods, ignoring the numerical values of the identifiers.

	However, we needed an encoding for our \emph{original} graph $G$, with the \emph{original} set of identifiers. To do that, we proceed as follows. We make the above thought experiment, to construct the advice for $H$. Then we simply copy the advice bits from component $G_1$ to the original graph $G$.

	If we now applied $\alg$ to solve $\Pi$ in $G$, it does not necessarily work. However, if we apply $\alg'(X,S) = F^*(S)$ to solve $\Pi$ in $G$, it will behave in exactly the same way as applying $\alg$ in $G_1$. Hence, $\alg'$ will also solve $\Pi$ correctly in $G$. Furthermore, $\alg'$ is by construction order-invariant.
\end{proof}

We note that the encoder $E'$ constructed above is not practical or efficient (even if the original encoder $E$ is). As we will see next, merely knowing that $\alg'$ exists can be sufficient.

\subsection{Hardness assuming the Exponential-Time Hypothesis}\label{ssec:eth}

In \cref{ssec:subexp-growth} we show that in graphs with sub-exponential growth, any LCL can be solved with one bit of advice per node. We now show that one bit does not suffice in general, assuming the Exponential-Time Hypothesis \cite{eth}.

Recall that the Exponential-Time Hypothesis states that there is a positive constant $\delta > 0$ such that 3-SAT cannot be solved in time $O(2^{\delta n})$, where $n$ is the number of variables in the 3-SAT instance. 
\begin{theorem}\label{thm:eth}
	Fix any $\beta$. Assume that for every LCL problem $\Pi$ there is some $T$ such that, on any input, $\Pi$ can be solved in $T$ rounds with $\beta$ bits of advice. Then the Exponential-Time Hypothesis is false.
\end{theorem}

\begin{proof}
	Suppose there is some $\beta$ such that all LCL problems can be solved with $\beta$ bits of advice. We show how to then solve 3-SAT in time $O(2^{\varepsilon n} \cdot n \cdot f(\varepsilon,\beta))$ for an arbitrarily small $\varepsilon > 0$, violating the Exponential-Time Hypothesis.

	Consider some 3-SAT instance $\phi$, with $n$ variables and $m$ clauses. By applying the sparsification lemma by \cite{sparsification-lemma}, we can assume w.l.o.g.\ that $m = O(n)$.

	Now turn $\phi$ into a bipartite graph, with nodes representing variables on one side and nodes representing clauses on the other side. Each clause is connected by edges with the three variables it contains. Each node is labeled by its type (variable or clause) and each edge $vc$ is labeled to indicate whether the variable $v$ is negated in the clause $c$. This results in a bipartite graph with $n + m = O(n)$ nodes, and each clause-node has degree at most $3$. Hence, the total number of edges is at most $3m$.
	
	Variable-nodes may have arbitrarily high degrees, but the sum of their degrees is bounded by $3m$. We replace each variable-node that has degree $d > 3$ by a cycle of $d$ variable-nodes, with the new edges labeled by equality constraints. This results in a graph in which all nodes have degree at most $3$, and the total number of nodes is bounded by $3m + m = O(n)$. Let $G_0$ be the resulting graph, let $n_0 = O(n)$ be the number of nodes in $G_0$, let $\Delta_0 = 3$ be the maximum degree, and let $\ell_0 = O(1)$ be the number of node and edge labels that we used to encode the instance.

	Now it is easy to define an LCL problem $\Pi_0$ such that a solution of $\Pi_0$ in graph $G_0$ can be interpreted as a satisfying assignment of the variables in formula $\phi$, and vice versa.

	Now assume that we have defined an LCL problem $\Pi_i$ and a graph $G_i$ with $n_i$ nodes, maximum degree $\Delta_i$, and $\ell_i$ labels; the base case $i = 0$ was presented above. Define a new LCL problem $\Pi_{i+1}$ and a new graph $G_{i+1}$ as follows. We contract edges in $G_i$ to construct $G_{i+1}$ so that we satisfy two properties: each node in $G_{i+1}$ represents $O(1)$ nodes of $G_i$, but the number of nodes is $n_{i+1} \le n_{i} / 2$. This can be achieved by e.g.\ greedily contracting edges, favoring the edges whose endpoints currently represent the smallest number of nodes of $G_i$.

	We label the nodes of $G_{i+1}$ so that given the input labels of the nodes, we can also recover the original graph $G_{i}$. As each node represents a bounded number of original nodes, and the original graph had maximum degree $\Delta_i$, the new graph will have maximum degree $\Delta_{i+1}$ that only depends on $i$. Also, the number of labels $\ell_{i+1}$ will be bounded by a constant that only depends on $i$. To see that everything can be indeed encoded with constant-size labels, note that each new node $v$ in $G_{i+1}$ can have its own \emph{local numbering} of the original nodes $v'$ that were contracted to $v$, that is, each node $v'$ can be represented as a pair $(v,a)$, where $v$ is a new node and $a$ is a sequence number. This way, the new label of node $v$ can use constant-size triples of the form $(a,b,x)$ to encode that graph $G_i$ had an edge from $(v,a)$ to $(v,b)$ with label $x$, and the new edge label of edge $(u,v)$ can also use constant-size triples of the form $(a,b,x)$ to encode that graph $G_i$ had an edge from $(u,a)$ to $(v,b)$ with label $x$. Finally, the new node label of $v$ will use constant-size pairs of the form $(a,x)$ to encode that the node label of $(v,a)$ in $G_i$ was $x$. This way we can use constant-size node and edge labels in $G_{i+1}$ to encode the structure and node and edge labels of $G_i$.

	Now we can define an LCL problem $\Pi_{i+1}$ such that a valid solution of $\Pi_{i+1}$ in $G_{i+1}$ can be mapped to a valid solution of $\Pi_i$ in $G_i$, and hence eventually to a satisfying assignment of $\phi$, and conversely a satisfying assignment of $\phi$ can be turned into a valid solution of $\Pi_{i+1}$. In essence, the output label of a node in $G_{i+1}$ captures the output labels of all $O(1)$ nodes of $G_i$ that it represents.

	Continuing this way for $\Theta(\log(\beta/\varepsilon))$ steps, we can construct a graph $G_i$ with fewer than $\varepsilon n / \beta$ nodes. Furthermore, there is some LCL problem $\Pi_i$ such that $\phi$ is a yes-instance if and only if there is a valid solution of $\Pi_i$ in $G_i$.

	By assumption, there exists a distributed algorithm $\alg$ that solves $\Pi_i$ with $\beta$ bits of advice per node. Using \cref{thm:order-invariant}, we can also assume that $\alg$ is order-invariant. Hence, $\alg$ is a finite function, and we can compute $\alg$ in constant time (where the constant depends on $i$, which depends on $\beta$ and $\varepsilon$, but is independent of $n$).

	Now, we simply try out all possible strings of advice; there are at most $2^{\varepsilon n}$ such strings. For each advice combination, we try to apply $\alg$ to solve $\Pi_i$ in $G_i$; we simulate $\alg$ at each of the $n_i = O(n)$ nodes. If and only if $\phi$ is satisfiable, we will find an advice string such that $\alg$ succeeds in solving $\Pi_i$. The overall running time is $O(2^{\varepsilon n} \cdot n \cdot f(\varepsilon,\beta))$.
\end{proof}

\makeatletter
\renewcommand{\bibsection}{\section*{\refname}}
\makeatother
\bibliographystyle{plainnat}
\bibliography{refs}

\begin{thebibliography}{57}
\providecommand{\natexlab}[1]{#1}
\providecommand{\url}[1]{\texttt{#1}}
\expandafter\ifx\csname urlstyle\endcsname\relax
  \providecommand{\doi}[1]{doi: #1}\else
  \providecommand{\doi}{doi: \begingroup \urlstyle{rm}\Url}\fi

\bibitem[Ard{\'e}vol~Mart{\'\i}nez et~al.(2023)Ard{\'e}vol~Mart{\'\i}nez, Caoduro, Feuilloley, Narboni, Pournajafi, and Raymond]{lower-bound-constant-size-local-cert}
Virginia Ard{\'e}vol~Mart{\'\i}nez, Marco Caoduro, Laurent Feuilloley, Jonathan Narboni, Pegah Pournajafi, and Jean-Florent Raymond.
\newblock A lower bound for constant-size local certification.
\newblock \emph{Theoretical Computer Science}, 971:\penalty0 114068, 2023.
\newblock \doi{10.1016/J.TCS.2023.114068}.

\bibitem[Balliu et~al.(2018)Balliu, Hirvonen, Korhonen, Lempi{\"a}inen, Olivetti, and Suomela]{balliu18lcl-complexity}
Alkida Balliu, Juho Hirvonen, Janne~H. Korhonen, Tuomo Lempi{\"a}inen, Dennis Olivetti, and Jukka Suomela.
\newblock New classes of distributed time complexity.
\newblock In Ilias Diakonikolas, David Kempe, and Monika Henzinger, editors, \emph{Proceedings of the 50th Annual {ACM} {SIGACT} Symposium on Theory of Computing, {STOC} 2018, Los Angeles, CA, USA, June 25-29, 2018}, pages 1307--1318. {ACM}, 2018.
\newblock \doi{10.1145/3188745.3188860}.

\bibitem[Balliu et~al.(2020)Balliu, Brandt, Olivetti, and Suomela]{balliu20lcl-randomness}
Alkida Balliu, Sebastian Brandt, Dennis Olivetti, and Jukka Suomela.
\newblock How much does randomness help with locally checkable problems?
\newblock In Yuval Emek and Christian Cachin, editors, \emph{{PODC} '20: {ACM} Symposium on Principles of Distributed Computing, Virtual Event, Italy, August 3-7, 2020}, pages 299--308. {ACM}, 2020.
\newblock \doi{10.1145/3382734.3405715}.

\bibitem[Balliu et~al.(2021{\natexlab{a}})Balliu, Brandt, Olivetti, and Suomela]{balliu20almost-global}
Alkida Balliu, Sebastian Brandt, Dennis Olivetti, and Jukka Suomela.
\newblock Almost global problems in the {LOCAL} model.
\newblock \emph{Distributed Computing}, 34\penalty0 (4):\penalty0 259--281, 2021{\natexlab{a}}.
\newblock \doi{10.1007/S00446-020-00375-2}.

\bibitem[Balliu et~al.(2021{\natexlab{b}})Balliu, Censor-Hillel, Maus, Olivetti, and Suomela]{balliu21lcl-congest}
Alkida Balliu, Keren Censor-Hillel, Yannic Maus, Dennis Olivetti, and Jukka Suomela.
\newblock Locally checkable labelings with small messages.
\newblock In Seth Gilbert, editor, \emph{35th International Symposium on Distributed Computing, {DISC} 2021, October 4-8, 2021, Freiburg, Germany (Virtual Conference)}, volume 209 of \emph{LIPIcs}, pages 8:1--8:18. Schloss Dagstuhl - Leibniz-Zentrum f{\"{u}}r Informatik, 2021{\natexlab{b}}.
\newblock \doi{10.4230/LIPICS.DISC.2021.8}.

\bibitem[Barenboim et~al.(2022)Barenboim, Elkin, and Goldenberg]{BarenboimEG18}
Leonid Barenboim, Michael Elkin, and Uri Goldenberg.
\newblock Locally-iterative distributed {$(\Delta+1)$}-coloring and applications.
\newblock \emph{Journal of the ACM}, 69\penalty0 (1):\penalty0 5:1--5:26, 2022.
\newblock \doi{10.1145/3486625}.

\bibitem[Bousquet et~al.(2024)Bousquet, Feuilloley, and Zeitoun]{BFZ24}
Nicolas Bousquet, Laurent Feuilloley, and S{\'e}bastien Zeitoun.
\newblock Local certification of local properties: Tight bounds, trade-offs and new parameters.
\newblock In Olaf Beyersdorff, Mamadou~Moustapha Kant{\'e}, Orna Kupferman, and Daniel Lokshtanov, editors, \emph{41st International Symposium on Theoretical Aspects of Computer Science, {STACS} 2024, March 12-14, 2024, Clermont-Ferrand, France}, volume 289 of \emph{LIPIcs}, pages 21:1--21:18. Schloss Dagstuhl - Leibniz-Zentrum f{\"{u}}r Informatik, 2024.
\newblock \doi{10.4230/LIPICS.STACS.2024.21}.

\bibitem[Boyar et~al.(2017)Boyar, Favrholdt, Kudahl, Larsen, and Mikkelsen]{BoyarFKLMsurvey}
Joan Boyar, Lene~M. Favrholdt, Christian Kudahl, Kim~S. Larsen, and Jesper~W. Mikkelsen.
\newblock Online algorithms with advice: {A} survey.
\newblock \emph{ACM Computing Surveys}, 50\penalty0 (2):\penalty0 19:1--19:34, 2017.
\newblock \doi{10.1145/3056461}.

\bibitem[Brandt et~al.(2016)Brandt, Fischer, Hirvonen, Keller, Lempi{\"a}inen, Rybicki, Suomela, and Uitto]{brandt16lll}
Sebastian Brandt, Orr Fischer, Juho Hirvonen, Barbara Keller, Tuomo Lempi{\"a}inen, Joel Rybicki, Jukka Suomela, and Jara Uitto.
\newblock A lower bound for the distributed {L}ov{\'a}sz local lemma.
\newblock In Daniel Wichs and Yishay Mansour, editors, \emph{Proceedings of the 48th Annual {ACM} {SIGACT} Symposium on Theory of Computing, {STOC} 2016, Cambridge, MA, USA, June 18-21, 2016}, pages 479--488. {ACM}, 2016.
\newblock \doi{10.1145/2897518.2897570}.

\bibitem[Chang and Pettie(2019)]{Chang2019}
Yi-Jun Chang and Seth Pettie.
\newblock A time hierarchy theorem for the {LOCAL} model.
\newblock \emph{SIAM Journal on Computing}, 48\penalty0 (1):\penalty0 33--69, 2019.
\newblock \doi{10.1137/17M1157957}.

\bibitem[Chang et~al.(2019)Chang, Kopelowitz, and Pettie]{chang19exponential}
Yi-Jun Chang, Tsvi Kopelowitz, and Seth Pettie.
\newblock An exponential separation between randomized and deterministic complexity in the {LOCAL} model.
\newblock \emph{SIAM Journal on Computing}, 48\penalty0 (1):\penalty0 122--143, 2019.
\newblock \doi{10.1137/17M1117537}.

\bibitem[Christiansen et~al.(2023)Christiansen, Nowicki, and Rotenberg]{ChristiansenNR23}
Aleksander Bj{\o}rn~Grodt Christiansen, Krzysztof Nowicki, and Eva Rotenberg.
\newblock Improved dynamic colouring of sparse graphs.
\newblock In Barna Saha and Rocco~A. Servedio, editors, \emph{Proceedings of the 55th Annual {ACM} Symposium on Theory of Computing, {STOC} 2023, Orlando, FL, USA, June 20-23, 2023}, pages 1201--1214. {ACM}, 2023.
\newblock \doi{10.1145/3564246.3585111}.

\bibitem[Dereniowski and Pelc(2012)]{dereniowski2012maps-advice}
Dariusz Dereniowski and Andrzej Pelc.
\newblock Drawing maps with advice.
\newblock \emph{Journal of Parallel and Distributed Computing}, 72\penalty0 (2):\penalty0 132--143, 2012.
\newblock \doi{10.1016/J.JPDC.2011.10.004}.

\bibitem[Dobrev et~al.(2012)Dobrev, Kr{\'a}lovic, and Markou]{dobrev2012exploration}
Stefan Dobrev, Rastislav Kr{\'a}lovic, and Euripides Markou.
\newblock Online graph exploration with advice.
\newblock In Guy Even and Magn{\'u}s~M. Halld{\'o}rsson, editors, \emph{Structural Information and Communication Complexity - 19th International Colloquium, {SIROCCO} 2012, Reykjavik, Iceland, June 30-July 2, 2012, Revised Selected Papers}, volume 7355 of \emph{Lecture Notes in Computer Science}, pages 267--278. Springer, 2012.
\newblock \doi{10.1007/978-3-642-31104-8_23}.

\bibitem[Emek et~al.(2009)Emek, Fraigniaud, Korman, and Ros{\'e}n]{EmekFKR09}
Yuval Emek, Pierre Fraigniaud, Amos Korman, and Adi Ros{\'e}n.
\newblock Online computation with advice.
\newblock In Susanne Albers, Alberto Marchetti-Spaccamela, Yossi Matias, Sotiris~E. Nikoletseas, and Wolfgang Thomas, editors, \emph{Automata, Languages and Programming, 36th International Colloquium, {ICALP} 2009, Rhodes, Greece, July 5-12, 2009, Proceedings, Part {I}}, volume 5555 of \emph{Lecture Notes in Computer Science}, pages 427--438. Springer, 2009.
\newblock \doi{10.1007/978-3-642-02927-1_36}.

\bibitem[Erd{\H{o}}s and Lov{\'a}sz(1975)]{lll}
P.~Erd{\H{o}}s and L.~Lov{\'a}sz.
\newblock Problems and results on {$3$}-chromatic hypergraphs and some related questions.
\newblock In \emph{Infinite and finite sets ({C}olloq., {K}eszthely, 1973; dedicated to {P}. {E}rd\H{o}s on his 60th birthday), {V}ols. {I}, {II}, {III}}, volume Vol. 10 of \emph{Colloq. Math. Soc. J\'{a}nos Bolyai}, pages 609--627. North-Holland, Amsterdam-London, 1975.

\bibitem[Feuilloley and Fraigniaud(2017)]{feuilloley2017survey}
Laurent Feuilloley and Pierre Fraigniaud.
\newblock Survey of distributed decision, 2017.

\bibitem[Feuilloley et~al.(2021)Feuilloley, Fraigniaud, Hirvonen, Paz, and Perry]{redundancy-in-distr-proofs}
Laurent Feuilloley, Pierre Fraigniaud, Juho Hirvonen, Ami Paz, and Mor Perry.
\newblock Redundancy in distributed proofs.
\newblock \emph{Distributed Computing}, 34\penalty0 (2):\penalty0 113--132, 2021.
\newblock \doi{10.1007/S00446-020-00386-Z}.

\bibitem[Fischer and Ghaffari(2017)]{fischer17sublogarithmic}
Manuela Fischer and Mohsen Ghaffari.
\newblock Sublogarithmic distributed algorithms for {L}ov{\'a}sz local lemma, and the complexity hierarchy.
\newblock In Andr{\'e}a~W. Richa, editor, \emph{31st International Symposium on Distributed Computing, {DISC} 2017, October 16-20, 2017, Vienna, Austria}, volume~91 of \emph{LIPIcs}, pages 18:1--18:16. Schloss Dagstuhl - Leibniz-Zentrum f{\"{u}}r Informatik, 2017.
\newblock \doi{10.4230/LIPICS.DISC.2017.18}.

\bibitem[Fraigniaud et~al.(2007)Fraigniaud, Korman, and Lebhar]{fraigniaud2007mst-advice}
Pierre Fraigniaud, Amos Korman, and Emmanuelle Lebhar.
\newblock Local {MST} computation with short advice.
\newblock In Phillip~B. Gibbons and Christian Scheideler, editors, \emph{{SPAA} 2007: Proceedings of the 19th Annual {ACM} Symposium on Parallelism in Algorithms and Architectures, San Diego, California, USA, June 9-11, 2007}, pages 154--160. {ACM}, 2007.
\newblock \doi{10.1145/1248377.1248402}.

\bibitem[Fraigniaud et~al.(2008)Fraigniaud, Ilcinkas, and Pelc]{fraigniaud2008tree-advice}
Pierre Fraigniaud, David Ilcinkas, and Andrzej Pelc.
\newblock Tree exploration with advice.
\newblock \emph{Information and Computation}, 206\penalty0 (11):\penalty0 1276--1287, 2008.
\newblock \doi{10.1016/J.IC.2008.07.005}.

\bibitem[Fraigniaud et~al.(2009)Fraigniaud, Gavoille, Ilcinkas, and Pelc]{distr-comp-with-advice}
Pierre Fraigniaud, Cyril Gavoille, David Ilcinkas, and Andrzej Pelc.
\newblock Distributed computing with advice: information sensitivity of graph coloring.
\newblock \emph{Distributed Computing}, 21\penalty0 (6):\penalty0 395--403, 2009.
\newblock \doi{10.1007/S00446-008-0076-Y}.

\bibitem[Fraigniaud et~al.(2010)Fraigniaud, Ilcinkas, and Pelc]{fraigniaud2010advice}
Pierre Fraigniaud, David Ilcinkas, and Andrzej Pelc.
\newblock Communication algorithms with advice.
\newblock \emph{Journal of Computer and System Sciences}, 76\penalty0 (3-4):\penalty0 222--232, 2010.
\newblock \doi{10.1016/J.JCSS.2009.07.002}.

\bibitem[Fraigniaud et~al.(2016)Fraigniaud, Heinrich, and Kosowski]{FHK16}
Pierre Fraigniaud, Marc Heinrich, and Adrian Kosowski.
\newblock Local conflict coloring.
\newblock In Irit Dinur, editor, \emph{{IEEE} 57th Annual Symposium on Foundations of Computer Science, {FOCS} 2016, 9-11 October 2016, Hyatt Regency, New Brunswick, New Jersey, {USA}}, pages 625--634. {IEEE} Computer Society, 2016.
\newblock \doi{10.1109/FOCS.2016.73}.

\bibitem[Fusco and Pelc(2011)]{fusco2011tradeoffs-advice}
Emanuele~G. Fusco and Andrzej Pelc.
\newblock Trade-offs between the size of advice and broadcasting time in trees.
\newblock \emph{Algorithmica}, 60\penalty0 (4):\penalty0 719--734, 2011.
\newblock \doi{10.1007/S00453-009-9361-9}.

\bibitem[Fusco et~al.(2016)Fusco, Pelc, and Petreschi]{fusco2016topology-advice}
Emanuele~G. Fusco, Andrzej Pelc, and Rossella Petreschi.
\newblock Topology recognition with advice.
\newblock \emph{Information and Computation}, 247:\penalty0 254--265, 2016.
\newblock \doi{10.1016/J.IC.2016.01.005}.

\bibitem[Ghaffari and Su(2017)]{ghaffari17distributed}
Mohsen Ghaffari and Hsin-Hao Su.
\newblock Distributed degree splitting, edge coloring, and orientations.
\newblock In Philip~N. Klein, editor, \emph{Proceedings of the Twenty-Eighth Annual {ACM-SIAM} Symposium on Discrete Algorithms, {SODA} 2017, Barcelona, Spain, Hotel Porta Fira, January 16-19}, pages 2505--2523. {SIAM}, 2017.
\newblock \doi{10.1137/1.9781611974782.166}.

\bibitem[Ghaffari et~al.(2018)Ghaffari, Harris, and Kuhn]{Ghaffari2018}
Mohsen Ghaffari, David~G. Harris, and Fabian Kuhn.
\newblock On derandomizing local distributed algorithms.
\newblock In Mikkel Thorup, editor, \emph{59th {IEEE} Annual Symposium on Foundations of Computer Science, {FOCS} 2018, Paris, France, October 7-9, 2018}, pages 662--673. {IEEE} Computer Society, 2018.
\newblock \doi{10.1109/FOCS.2018.00069}.

\bibitem[Ghaffari et~al.(2020{\natexlab{a}})Ghaffari, Grunau, and Jin]{ghaffari_et_al:LIPIcs.DISC.2020.34}
Mohsen Ghaffari, Christoph Grunau, and Ce~Jin.
\newblock Improved {MPC} algorithms for mis, matching, and coloring on trees and beyond.
\newblock In Hagit Attiya, editor, \emph{34th International Symposium on Distributed Computing, {DISC} 2020, October 12-16, 2020, Virtual Conference}, volume 179 of \emph{LIPIcs}, pages 34:1--34:18. Schloss Dagstuhl - Leibniz-Zentrum f{\"{u}}r Informatik, 2020{\natexlab{a}}.
\newblock \doi{10.4230/LIPICS.DISC.2020.34}.

\bibitem[Ghaffari et~al.(2020{\natexlab{b}})Ghaffari, Hirvonen, Kuhn, Maus, Suomela, and Uitto]{degree-splitting}
Mohsen Ghaffari, Juho Hirvonen, Fabian Kuhn, Yannic Maus, Jukka Suomela, and Jara Uitto.
\newblock Improved distributed degree splitting and edge coloring.
\newblock \emph{Distributed Computing}, 33\penalty0 (3-4):\penalty0 293--310, 2020{\natexlab{b}}.
\newblock \doi{10.1007/S00446-018-00346-8}.

\bibitem[Ghaffari et~al.(2021)Ghaffari, Hirvonen, Kuhn, and Maus]{GHKM18}
Mohsen Ghaffari, Juho Hirvonen, Fabian Kuhn, and Yannic Maus.
\newblock Improved distributed {\(\Delta\)}-coloring.
\newblock \emph{Distributed Computing}, 34\penalty0 (4):\penalty0 239--258, 2021.
\newblock \doi{10.1007/S00446-021-00397-4}.

\bibitem[Glacet et~al.(2017)Glacet, Miller, and Pelc]{glacet2017leader-advice}
Christian Glacet, Avery Miller, and Andrzej Pelc.
\newblock Time vs. information tradeoffs for leader election in anonymous trees.
\newblock \emph{ACM Transactions on Algorithms}, 13\penalty0 (3):\penalty0 31:1--31:41, 2017.
\newblock \doi{10.1145/3039870}.

\bibitem[G{\"o}{\"o}s and Suomela(2016)]{goos16lcp}
Mika G{\"o}{\"o}s and Jukka Suomela.
\newblock Locally checkable proofs in distributed computing.
\newblock \emph{Theory of Computing}, 12\penalty0 (1):\penalty0 1--33, 2016.
\newblock \doi{10.4086/TOC.2016.V012A019}.

\bibitem[G{\"o}{\"o}s et~al.(2013)G{\"o}{\"o}s, Hirvonen, and Suomela]{goos13local-approximation}
Mika G{\"o}{\"o}s, Juho Hirvonen, and Jukka Suomela.
\newblock Lower bounds for local approximation.
\newblock \emph{Journal of the ACM}, 60\penalty0 (5):\penalty0 39:1--39:23, 2013.
\newblock \doi{10.1145/2528405}.

\bibitem[G{\"o}{\"o}s et~al.(2017)G{\"o}{\"o}s, Hirvonen, and Suomela]{goos15mep-lb}
Mika G{\"o}{\"o}s, Juho Hirvonen, and Jukka Suomela.
\newblock Linear-in-{$\Delta$} lower bounds in the {LOCAL} model.
\newblock \emph{Distributed Computing}, 30\penalty0 (5):\penalty0 325--338, 2017.
\newblock \doi{10.1007/S00446-015-0245-8}.

\bibitem[Gorain and Pelc(2019)]{gorain2018exploration-advice}
Barun Gorain and Andrzej Pelc.
\newblock Deterministic graph exploration with advice.
\newblock \emph{ACM Transactions on Algorithms}, 15\penalty0 (1):\penalty0 8:1--8:17, 2019.
\newblock \doi{10.1145/3280823}.

\bibitem[Ilcinkas et~al.(2010)Ilcinkas, Kowalski, and Pelc]{ilcinkas2010broadcasting-advice}
David Ilcinkas, Dariusz~R. Kowalski, and Andrzej Pelc.
\newblock Fast radio broadcasting with advice.
\newblock \emph{Theoretical Computer Science}, 411\penalty0 (14-15):\penalty0 1544--1557, 2010.
\newblock \doi{10.1016/J.TCS.2010.01.004}.

\bibitem[Impagliazzo and Paturi(1999)]{eth}
Russell Impagliazzo and Ramamohan Paturi.
\newblock Complexity of k-sat.
\newblock In \emph{Proceedings of the 14th Annual {IEEE} Conference on Computational Complexity, Atlanta, Georgia, USA, May 4-6, 1999}, pages 237--240. {IEEE} Computer Society, 1999.
\newblock \doi{10.1109/CCC.1999.766282}.

\bibitem[Impagliazzo et~al.(2001)Impagliazzo, Paturi, and Zane]{sparsification-lemma}
Russell Impagliazzo, Ramamohan Paturi, and Francis Zane.
\newblock Which problems have strongly exponential complexity?
\newblock \emph{Journal of Computer and System Sciences}, 63\penalty0 (4):\penalty0 512--530, 2001.
\newblock \doi{10.1006/JCSS.2001.1774}.

\bibitem[Komm et~al.(2015)Komm, Kr{\'a}lovic, Kr{\'a}lovic, and Smula]{komm2015treasure-advice}
Dennis Komm, Rastislav Kr{\'a}lovic, Richard Kr{\'a}lovic, and Jasmin Smula.
\newblock Treasure hunt with advice.
\newblock In Christian Scheideler, editor, \emph{Structural Information and Communication Complexity - 22nd International Colloquium, {SIROCCO} 2015, Montserrat, Spain, July 14-16, 2015, Post-Proceedings}, volume 9439 of \emph{Lecture Notes in Computer Science}, pages 328--341. Springer, 2015.
\newblock \doi{10.1007/978-3-319-25258-2_23}.

\bibitem[Korman and Kutten(2006)]{korman06distributed}
Amos Korman and Shay Kutten.
\newblock On distributed verification.
\newblock In Soma Chaudhuri, Samir~R. Das, Himadri~S. Paul, and Srikanta Tirthapura, editors, \emph{Distributed Computing and Networking, 8th International Conference, {ICDCN} 2006, Guwahati, India, December 27-30, 2006}, volume 4308 of \emph{Lecture Notes in Computer Science}, pages 100--114. Springer, 2006.
\newblock \doi{10.1007/11947950_12}.

\bibitem[Korman and Kutten(2007)]{korman07distributed}
Amos Korman and Shay Kutten.
\newblock Distributed verification of minimum spanning trees.
\newblock \emph{Distributed Computing}, 20\penalty0 (4):\penalty0 253--266, 2007.
\newblock \doi{10.1007/S00446-007-0025-1}.

\bibitem[Korman et~al.(2010{\natexlab{a}})Korman, Kutten, and Peleg]{korman10proof}
Amos Korman, Shay Kutten, and David Peleg.
\newblock Proof labeling schemes.
\newblock \emph{Distributed Computing}, 22\penalty0 (4):\penalty0 215--233, 2010{\natexlab{a}}.
\newblock \doi{10.1007/S00446-010-0095-3}.

\bibitem[Korman et~al.(2010{\natexlab{b}})Korman, Peleg, and Rodeh]{korman10constructing}
Amos Korman, David Peleg, and Yoav Rodeh.
\newblock Constructing labeling schemes through universal matrices.
\newblock \emph{Algorithmica}, 57\penalty0 (4):\penalty0 641--652, 2010{\natexlab{b}}.
\newblock \doi{10.1007/S00453-008-9226-7}.

\bibitem[Linial(1992)]{Linial92}
Nathan Linial.
\newblock Locality in distributed graph algorithms.
\newblock \emph{SIAM Journal on Computing}, 21\penalty0 (1):\penalty0 193--201, 1992.
\newblock \doi{10.1137/0221015}.

\bibitem[Maus and Tonoyan(2022)]{MausT22}
Yannic Maus and Tigran Tonoyan.
\newblock Linial for lists.
\newblock \emph{Distributed Computing}, 35\penalty0 (6):\penalty0 533--546, 2022.
\newblock \doi{10.1007/S00446-022-00424-Y}.

\bibitem[Miller and Pelc(2015{\natexlab{a}})]{miller2015rendezvous-advice}
Avery Miller and Andrzej Pelc.
\newblock Fast rendezvous with advice.
\newblock \emph{Theoretical Computer Science}, 608:\penalty0 190--198, 2015{\natexlab{a}}.
\newblock \doi{10.1016/J.TCS.2015.09.025}.

\bibitem[Miller and Pelc(2015{\natexlab{b}})]{miller2015treasure-hunt}
Avery Miller and Andrzej Pelc.
\newblock Tradeoffs between cost and information for rendezvous and treasure hunt.
\newblock \emph{Journal of Parallel and Distributed Computing}, 83:\penalty0 159--167, 2015{\natexlab{b}}.
\newblock \doi{10.1016/J.JPDC.2015.06.004}.

\bibitem[Miller and Pelc(2016)]{miller2016election-advice}
Avery Miller and Andrzej Pelc.
\newblock Election vs. selection: How much advice is needed to find the largest node in a graph?
\newblock In Christian Scheideler and Seth Gilbert, editors, \emph{Proceedings of the 28th {ACM} Symposium on Parallelism in Algorithms and Architectures, {SPAA} 2016, Asilomar State Beach/Pacific Grove, CA, USA, July 11-13, 2016}, pages 377--386. {ACM}, 2016.
\newblock \doi{10.1145/2935764.2935772}.

\bibitem[Mitzenmacher and Vassilvitskii(2020)]{DBLP:books/cu/20/MitzenmacherV20}
Michael Mitzenmacher and Sergei Vassilvitskii.
\newblock Algorithms with predictions.
\newblock In Tim Roughgarden, editor, \emph{Beyond the Worst-Case Analysis of Algorithms}, pages 646--662. Cambridge University Press, 2020.
\newblock \doi{10.1017/9781108637435.037}.

\bibitem[Naor and Stockmeyer(1995)]{what-can-be-computed-locally}
Moni Naor and Larry~J. Stockmeyer.
\newblock What can be computed locally?
\newblock \emph{SIAM Journal on Computing}, 24\penalty0 (6):\penalty0 1259--1277, 1995.
\newblock \doi{10.1137/S0097539793254571}.

\bibitem[Nisse and Soguet(2009)]{nisse2009searching-advice}
Nicolas Nisse and David Soguet.
\newblock Graph searching with advice.
\newblock \emph{Theoretical Computer Science}, 410\penalty0 (14):\penalty0 1307--1318, 2009.
\newblock \doi{10.1016/J.TCS.2008.08.020}.

\bibitem[Panconesi and Srinivasan(1992)]{PS92}
Alessandro Panconesi and Aravind Srinivasan.
\newblock Improved distributed algorithms for coloring and network decomposition problems.
\newblock In S.~Rao Kosaraju, Mike Fellows, Avi Wigderson, and John~A. Ellis, editors, \emph{Proceedings of the 24th Annual {ACM} Symposium on Theory of Computing, May 4-6, 1992, Victoria, British Columbia, Canada}, pages 581--592. {ACM}, 1992.
\newblock \doi{10.1145/129712.129769}.

\bibitem[Parnas and Ron(2007)]{PARNAS2007183}
Michal Parnas and Dana Ron.
\newblock Approximating the minimum vertex cover in sublinear time and a connection to distributed algorithms.
\newblock \emph{Theoretical Computer Science}, 381\penalty0 (1-3):\penalty0 183--196, 2007.
\newblock \doi{10.1016/J.TCS.2007.04.040}.

\bibitem[Rozhon and Ghaffari(2020)]{Rozhon2019}
V{\'a}clav Rozhon and Mohsen Ghaffari.
\newblock Polylogarithmic-time deterministic network decomposition and distributed derandomization.
\newblock In Konstantin Makarychev, Yury Makarychev, Madhur Tulsiani, Gautam Kamath, and Julia Chuzhoy, editors, \emph{Proceedings of the 52nd Annual {ACM} {SIGACT} Symposium on Theory of Computing, {STOC} 2020, Chicago, IL, USA, June 22-26, 2020}, pages 350--363. {ACM}, 2020.
\newblock \doi{10.1145/3357713.3384298}.

\bibitem[Shearer(1985)]{shearer}
James~B. Shearer.
\newblock On a problem of {Spencer}.
\newblock \emph{Combinatorica}, 5\penalty0 (3):\penalty0 241--245, 1985.
\newblock \doi{10.1007/BF02579368}.

\bibitem[Spencer(1977)]{SPENCER197769}
Joel Spencer.
\newblock Asymptotic lower bounds for ramsey functions.
\newblock \emph{Discrete Mathematics}, 20:\penalty0 69--76, 1977.
\newblock \doi{10.1016/0012-365X(77)90044-9}.

\end{thebibliography}

\newpage
\appendix

\section{LCLs on graphs with sub-exponential growth}\label{ssec:subexp-growth}
In this section, we show that on graphs of sub-exponential growth, we can solve any LCL problem in constant time by using $1$ bit per node of advice (we remind the reader that, as per definition of LCLs, the class of graphs that we consider are the ones with constant maximum degree, hence throughout this section we assume $\Delta=O(1)$). 
In more detail, we will prove the following theorem, and devote the rest of the section to proving it.

\begin{theorem}\label{thm:subexp}
	Let $\Pi$ be an LCL problem, and let $\mathcal{G}$ be a family of graphs of maximum degree $\Delta$ of sub-exponential growth in which $\Pi$ is solvable. Then, there exists a uniform fixed-length sparse $(\mathcal{G},\Pi,1,O(1))$-advice schema.
\end{theorem}

We start by giving a formal definition of sub-exponential growth. For a node $v$, we denote with $N_{\le d}(v)$ the set of nodes at distance at most $d$ from $v$, and we denote with $N_{= d}(v)$ the set of nodes at distance exactly $d$ from $v$. 

\begin{definition}[Sub-exponential growth]\label{def:subexp-growth}
	A family $\mathcal{G}$ of graphs is of sub-exponential growth if, for any constant $c>0$, there exists a constant $x_0$ such that, for any $x\ge x_0$, for any $G=(V,E)\in\mathcal{G}$, for any $v\in V$, $|N_{\le x}(v)|\le 2^{c\cdot x}$.
\end{definition}

\subparagraph{High-level ideas.}
On a high level, we would like to exploit that, on graphs with sub-exponential growth, it is possible to somehow compute a clustering such that the number of nodes on the border of each cluster is much less than the number of nodes inside each cluster. Ideally, this property would allow us to encode the solution of the nodes of the border of a cluster by using the nodes inside that cluster. Then, inside each cluster, by brute force, we can complete the partial solution assigned to the border of the cluster, and hence each node would know its own part of the solution, as desired. However, things are not so simple: (1) not  only we need to encode the solution of the nodes that are in the borders of the clusters, but we also need to be able to encode a clustering of the graph, and all this by using only one bit per node; (2) we need to make sure that the decoding process is able to decode the right thing, and this requires that the nodes that encode some information and that are inside two different clusters are far enough from each other. This latter issue forces us to prove some stronger properties of graphs with sub-exponential growth, that is that it is possible to compute a clustering such that the number of nodes that are \emph{very} close to the center of the cluster is at least as large as the number of nodes in the border (\Cref{lem:subexp-nr-nodes}). In more detail, we prove that, for each node $v$, there exists some value $\alpha$ such that, for a cluster $C_v$ centered at $v$, the set $N_{\le \alpha}$ of nodes of $C_v$ that are at distance at most $\alpha$ from $v$ is such that all nodes in $N_{\le \alpha}$ are far enough from the border and they are at least as many as the ones in the border. We will use only nodes in $N_{\le \alpha}$ to encode any useful information for us. In order to encode the clustering we do the following. First, we compute a distance-$d$ coloring of the graph (for some fixed $d$), that is, a coloring of the nodes such that nodes with the same color are at distance strictly larger than $d$. Then, we proceed through color classes in some order and at each phase $i$ we assign nodes to some cluster centered at a node with color $i$ (the nodes that will not be assigned to any cluster in this phase will be processed later). In this case, we say that $i$ is the color of the cluster. After we are done with color-class $i$, we remove the nodes that are part of some cluster and we repeat the procedure by considering the next color on the remaining subgraph. The distance coloring ensures that we avoid any kind of collisions with other centers of the same color.
In order to encode the color of a cluster $C_v$ centered at $v$ we use a path that starts at node $v$ and is contained in $N_{\le \alpha}$ (and hence it is far enough from the border of the cluster). Then, we use an independent set of the nodes in $N_{\le \alpha}$ in order to encode the solution of the border (\Cref{lem:subexp-nr-nodes} ensures that these nodes are enough for our purposes). The fact that we use only nodes in $N_{\le \alpha}$ to encode information ensures that nodes that carry information regarding different clusters are far enough from each other (even if the clusters are adjacent) and hence we avoid wrong decodings. In order to be able to distinguish different encodings inside a cluster (e.g., the color of the cluster, the output of the border, the center of the cluster) we use the following scheme. Consider the connected components induced by the nodes having as input a $1$: the $1$s that form an independent set are the ones that encode the solution of the border; the other $1$s are used to encode the cluster center and its color, where four $1$s encode the center of the cluster, two $1$s encode a $0$ and three $1$s encode a $1$.

We start by proving a technical lemma regarding graphs of sub-exponential growth, which will be useful later.

\begin{lemma}\label{lem:subexp-nr-nodes}
	Let $\mathcal{G}$ be a family of graphs of sub-exponential growth that have degree bounded by $\Delta$. 
	For any integer $r>0$, let $c:=\log(1 + 1/\Delta^{r})/(3r)$, and let $x_r:=\max\{4r, x_0\}$, where $x_0$ is the value given by \Cref{def:subexp-growth} for $\mathcal{G}$, when using parameter $c$. Then, for all $G = (V,E) \in \mathcal{G}$, and all $v \in V$, there exists a value $\alpha$ in $\{x_r,\ldots,2 x_r\}$ such that $|N_{\le \alpha}(v)| \ge \Delta^r |N_{=\alpha+r}(v)|$.
\end{lemma}
\begin{proof}
	Assume, for a contradiction, that, for any $\alpha$ in $\{x_r,\ldots,2 x_r\}$, it holds that $|N_{\le \alpha}(v)| < \Delta^r |N_{=\alpha+r}(v)|$. We prove, by induction, that, under this assumption, $|N_{\le x_r + kr}(v)|\ge (1 + 1/\Delta^r)^k \cdot |N_{\le x_r}(v)|$, for any $k$ satisfying $0\le kr \le x_r$. For the base case $k=0$, we need to prove that  $|N_{\le x_r}(v)|\ge |N_{\le x_r}(v)|$, which trivially holds. Assume that the statement holds for $k$, we show that it holds for $k+1$. We can lower bound $|N_{\le x_r + (k+1)r}(v)|$ with $|N_{\le x_r + kr}(v)| + |N_{= x_r + (k+1)r}(v)|\ge |N_{\le x_r + kr}(v)| +  |N_{\le x_r + kr}(v)|/\Delta^r = |N_{\le x_r + kr}(v)| (1 + 1/\Delta^r) \ge (1 + 1/\Delta^r)^k \cdot |N_{\le x_r}(v)|\cdot (1 + 1/\Delta^r) = (1 + 1/\Delta^r)^{k+1} \cdot |N_{\le x_r}(v)|$, proving the inductive step. Moreover, by the assumption of sub-exponential growth, we have that $|N_{\le x_r + kr}(v)|\le 2^{c(x_r + kr)}$. By combining the obtained bounds, we obtain the following:
	\[
	(1 + 1/\Delta^r)^k \cdot |N_{\le x_r}(v)|\le |N_{\le x_r + kr}(v)|\le 2^{c(x_r + kr)}.
	\]
	Hence, we get a contradiction if 
	\[
	(1 + 1/\Delta^r)^k \cdot |N_{\le x_r}(v)| >  2^{c(x_r + kr)}.
	\]
	Since $|N_{\le x_r}(v)|\ge 1$, it is sufficient to show that there exists a value of $k$ satisfying $0\le kr \le x_r$, such that
	\[
	(1 + 1/\Delta^r)^k >  2^{c(x_r + kr)}, \text{ which is implied by } k>\frac{x_r}{\log(1 + 1/\Delta^{r})/c - r}.
	\]
	Since $c=\log(1 + 1/\Delta^{r})/(3r)$, we get that, in order to get our contradiction, it is sufficient to give an integer $k$ satisfying $k>x_r/(2r)$ and $0\le kr\le x_r$. For this purpose, we take $k:=\lceil \frac{x}{2r}\rceil + 1$. Observe that $k>x_r/(2r)$ is clearly satisfied. Also, $kr\le (\lceil \frac{x_r}{2r}\rceil + 1)r\le (\frac{x_r}{2r} + 2)r=x_r/2 + 2r \le x_r$, where the last inequality holds since $x_r\ge 4r$ (recall that $x_r$ is defined as $\max\{4r, x_0\}$).
\end{proof}

\subparagraph{A clustering of the graph.}
In order to show the desired advice schema, we first compute a clustering of the graph as follows. Let $r$ be a large-enough constant to be fixed later. Let $c$ and $x$ be the values given by \Cref{lem:subexp-nr-nodes} as a function of $r$. Let $G\in\mathcal{G}$. At first, we compute a distance-$(5x)$ coloring on $G$. Since $G$ is a graph of sub-exponential growth, such a distance coloring can be done by using $2^{c\cdot 5x}$ colors (this is because $x\ge x_0$, where $x_0$ is the value given by \Cref{def:subexp-growth} when using parameter $c$).
Then, we process nodes by color classes, in ascending order, and, at each step $i$, we assign some nodes to some clusters. Let $G_i$ be the subgraph of $G$ induced by nodes that, at the end of step $i-1$, do not belong to any cluster yet. Let $v$ be a node of color $i$. We consider two possible cases: either, in $G_{i}$ it holds that $|N_{=2x}(v)|=0$, or in $G_{i}$ it holds that $|N_{=2x}(v)|>0$. For each node in the latter case, we compute a cluster as follows. Let $\alpha_v$ be the value of $\alpha$ given by \Cref{lem:subexp-nr-nodes} for node $v$. We create a cluster containing $v$ and all the nodes of $G_i$, that, in $G_i$, have distance at most $\alpha_v + r$ from $v$. We call $v$ the \emph{center} of this cluster. Observe that, at the end of this cluster-formation procedure, each node $v$ that does not belong to any cluster satisfies that, within distance $2x$, it sees the whole connected component containing $v$ induced by nodes that do not belong to any cluster. 

\subparagraph{Encoding the clustering.}
We provide an assignment of bits to the nodes of $G$ that encodes the clustering. We will later show how to add other bits on top of such an assignment and encode a solution of $\Pi$ in such a way that it is clear which bits encode the clustering and which bits encode the solution of $\Pi$. For each cluster $C$, let $v$ be its center, and let $i$ be the color of $v$. For any $u\in C$, let $d_u$ be the distance of $u$ from $v$. By the construction of $C$, there must exist $u\in C$ such that $d_u=y:=\lfloor x/2\rfloor$. Consider an arbitrary path $P_v=(v_1,\ldots, v_y)$ of length exactly $y$ where $v_1=v$ and $v_y=u$. This is the path that we will use for encoding the color of $v$. Let $B$ be the bit-string representation of $i$, which is of length $\lceil\log i \rceil$.
We modify $B$ by replacing each $0$ with the bit-string $110$, and each $1$ with the bit-string $1110$, obtaining the bit-string $B'$. Then, we define $B''$ as the concatenation of $11110110$ with $B'$ and with $0$. In total, we use at most $4|B|+9=4 \lceil\log i \rceil+9\le 4 \lceil\log (2^{c\cdot 5x})\rceil + 9=20cx+9$ bits. By setting $r$ large enough, we obtain a small-enough $c$ and a large-enough $x$ such that $20cx+9\le \lfloor x/2\rfloor =y=|P_v|$.
Hence, we can assign the $j$th bit of $B''$ to the node $v_j$ of $P_v$. Finally, we assign the bit $0$ to all nodes that did not get any bit from the previous bit-assignment. Note that all nodes within distance $\alpha_v +r$ (in $G_i$) from $v$ that are not part of $P_v$ have their bit set to $0$. 

Let $C_v$ be the set of nodes that belong to the cluster of $v$. We prove a property that will be useful later. 
\begin{lemma}\label{lem:not-neighbors}
	Let $u$ and $w$ be two cluster centers. Let $u'$ be an arbitrary node in $C_u \cap N_{\le \alpha_u}(u)$, and let $w'$ be an arbitrary node in $C_w \cap N_{\le \alpha_w}(w)$. Then, $u'$ and $w'$ are not neighbors in $G$.
\end{lemma} 
\begin{proof}
	If $u$ and $w$ have the same color, then they are at distance at least $5x$, and hence all nodes in $N_{\le \alpha_u}(u)$ are at distance at least $x$ from nodes in $N_{\le \alpha_w}(w)$, and hence the property holds. Therefore, assume, w.l.o.g., that $u$ has color $i$ and $w$ has color $j$, for $i<j$. Then, in $G_i$, all nodes in $C_w \cap N_{\le \alpha_w}(w)$ are at distance at least $r >0$ from all nodes in $C_u \cap N_{\le \alpha_u}(u)$, since the cluster $C_u$ contains all nodes within distance $\alpha_u + r$ from $u$ in $G_i$. This implies that they are at distance strictly greater than zero also in $G$.
\end{proof}
\begin{corollary}\label{lem:paths-not-neighbors}
	For any two cluster centers $u$ and $w$, it must hold that, in $G$, there is no edge connecting a node in $P_u$ with a node in $P_w$.
\end{corollary} 
\begin{proof}
	The proof follows from \Cref{lem:not-neighbors} and by the fact that each path contains only nodes within distance $y<\alpha$ from the center of its cluster.
\end{proof}

\subparagraph{Decoding the clustering.}
We show that nodes can decode the clustering in $2^{O(x)}=O(1)$ rounds, by using a recursive procedure that spends $O(x)$ rounds for each color. 
Before diving into showing this, we observe that, for a family of graphs $\mathcal{G}$, nodes can compute the appropriate values of $r$, $c$, $x_0$ and $x$ needed for the clustering, and hence we assume these values to be hard-coded in the algorithm.

Assume that we already decoded the clusters having centers of color $1,\ldots, i-1$. We now show how to decode clusters having centers of color $i$. Each node gathers its $2x$-radius neighborhood. In this way, each node $v$ knows which of the nodes within distance $2x$ from $v$ are part of some cluster with center of color at most $i-1$, and which nodes within distance $2x$ do not belong to any cluster yet. Hence, each node $v$ knows which nodes within distance $2x$ are part of $G_i$. In the following, all computation will be done with regard to $G_i$. 

Let $S$ be the set of nodes $v$ of $G_i$ satisfying the following: 
\begin{itemize}
	\item $|N_{=2x}(v)|>0$; 
	\item for any $0\le j\le y$, there is at most one node in $N_{=j}(v)$ that has its bit set to $1$;
	\item for any $y+1\le j\le x$, nodes in $N_{=j}(v)$ have their bit set to $0$;
	\item there exists a path $P=(v_1,\ldots,v_y)$ satisfying that, for all $j$, the node $v_j$ is at distance exactly $j-1$ from $v$, and that contains all nodes in $N_{\le y}(v)$ that have their bits set to $1$;
	\item the bit-string corresponding to nodes $v_1,\ldots,v_8$ is equal to $11110110$;
	\item nodes ${v_9,\ldots, v_y}$ form a bit-string of the form $(110|1110)^* \,0 \,0^*$.
\end{itemize}
By the encoding of the clusters, it is clear that all centers of color $i$ must satisfy these properties. Consider the set $S'\subseteq S$ of nodes satisfying that, by decoding the bit-string on the path, they obtain color $i$. We prove that nodes in $S'$ are all and only the centers of clusters of color $i$. Recall that, when we encoded clusters having a center $v$ of color $i$, we assigned a $0$ to all nodes within distance $\alpha_v +r$ (in $G_i$) from $v$ that are not part of $P_v$. Hence, it is clear that, if a center $v$ has color $i$, then $v\in S'$. Let $v$ be an arbitrary node in $S'$, and let $P=(v_1,\ldots,v_y)$ be a path satisfying the properties above. Then we show that $v$ must be a center of color $i$. By \Cref{lem:paths-not-neighbors}, nodes belonging to different paths form different connected components, and hence, if four nodes form a connected component containing four $1$s, then one of these nodes must be the center of a cluster. This implies that one node among $v_1,\ldots,v_4$ must be the center of a cluster. By the above properties, we get that no node among $v_2, v_3, v_4$ can be in $S'$, and hence $v_1=v$ is indeed a center of a cluster. Also, by the above properties, there are no other centers within distance $x$ from $v$. We therefore get that the bit-string encoded on $P$ is exactly the color of $v$, and hence that the color of $v$ is $i$.

Each node $v$ in $S'$ can compute $\alpha_v$ by exploring its $2x$-radius neighborhood, and hence by spending additional $\alpha_v+r=O(x)$ rounds, it can inform all nodes that belong to the cluster of $v$.

\subparagraph{\boldmath Encoding a solution for $\Pi$.}
In the following, we modify the bit assignment that encodes the clustering in order to also encode a solution for $\Pi$. First of all, recall that the assignment for the clustering satisfies that each connected component induced by nodes having their bit set to $1$ has size at least $2$. In order to avoid ambiguity when decoding, we assign $1$s to an independent set of nodes that are not neighbors of nodes that have already a $1$ assigned.

Let $\ell$ be a function mapping each node-edge pair of $G$ into a label of $\Pi$, such that the labeling obtained by $\ell$ is a solution for $\Pi$. We process clusters by colors in ascending order. Let $C_v$ be a cluster centered at $v$ with color $i$. In the following, all computation is done w.r.t.\ $G_i$. Let $\bar{r}$ be the checkability radius of $\Pi$. Let $S_v$ be the set of nodes within distance $\bar{r}$ (in $G_i$) from at least one node in $N_{=\alpha_v + r}(v)$. Note that $|S_v| \le N_{=\alpha_v + r}(v) \Delta^{\bar{r}}$. Hence, by \Cref{lem:subexp-nr-nodes}, $|S_v| \le |N_{\le \alpha_v}(v)| \Delta^{\bar{r}} / \Delta^r$. Consider the labeling given by $\ell$ to all node-edge pairs incident to the nodes in $S_v$. We can encode such a labeling by using a bit-string of length $k|N_{\le \alpha_v}(v)| \Delta^{\bar{r}+1} / \Delta^r$, where $k$ is a constant that depends on the alphabet of $\Pi$. Let $B$ be such a bit-string. Let $Y$ be the set containing the nodes of $N_{\le \alpha_v}(v)$ that have their bit set to $1$, plus all their neighbors. Let $Z = N_{\le \alpha_v}(v) \setminus Y$.
The nodes that have their bits set to $1$ are at most $y = \lfloor x/2 \rfloor$, and each such node has at most $\Delta -1$ neighbors that have a $0$. 
Recall that $v$ has at least one node at distance exactly $2x$ (as otherwise $v$ would not be a cluster center), and that $N_{\le \alpha_v}(v)$ contains all nodes within distance $\alpha_v \ge x$ from $v$. 
We get that there exists an injective function mapping nodes that have a $1$ to nodes in $Z$. Thus, $|Z| \ge |Y| / \Delta$, and hence $|N_{\le \alpha_v}(v)| = |Z| + |Y| \ge |Y| / \Delta + |Y| = |Y|(1 + 1/\Delta)$, which implies $|Y| \le |N_{\le \alpha_v}(v)|  / (1+1/\Delta)$. Thus, we get that $|Z| \ge |N_{\le \alpha_v}(v)| - |Y| \ge |N_{\le \alpha_v}(v)| / (\Delta+1)$. Hence, any independent set in $Z$ has size at least $|N_{\le \alpha_v}(v)| / (\Delta+1)^2$. By picking $r$ large enough, we have that
\[
|B| \le k|N_{\le \alpha_v}(v)| \Delta^{\bar{r}+1} / \Delta^r \le |N_{\le \alpha_v}(v)| / (\Delta+1)^2.
\]

We consider the nodes in $Z$ in order of their IDs, and we greedily compute a \MIS{} $Z'$. Consider the nodes in $Z'$ sorted by their IDs. We assign to the $j$th node of $Z'$ the $j$th bit of $B$.
By \Cref{lem:not-neighbors}, the nodes that received a $1$ for the encoding of the solution of $\Pi$ form an independent set in $G$.

\subparagraph{\boldmath Decoding a solution for $\Pi$.}
At first, nodes decode the clustering by ignoring the $1$s that form an independent set. Then, clusters are processed in ascending order of their color. Each cluster center $v$ can compute $Z$, and hence $Z'$, and sort the nodes of $Z$ by their IDs. In this way, $v$ can recover the bit-string $B$. The solution encoded in $B$ is assigned to the node-edge pairs of the nodes in $S_v$. Then, the nodes in $C_v$ can complete the solution inside the cluster by brute force by e.g., have $v$ gather the whole cluster, compute a solution (that respects the partial one) for all nodes in the cluster of $v$, and send the computed solution to all the nodes in the cluster (note that, by the construction, a solution for the cluster that agrees with the partially assigned solution exists). Finally, nodes that do not belong to any cluster can complete a solution by brute force because they are in connected components of constant diameter (again, by the construction, a valid completion exists). In total, this procedure takes $2^{O(x)}=O(1)$ rounds.

\subparagraph{Sparsity.}
We now show that we can make the ratio between nodes holding a $1$ and those holding a $0$ an arbitrarily small constant. There are two places in the schema where we assign $1$s, and those are the following.
\begin{itemize}
	\item When we encode the color of the center of the clusters: in this case, the number of $1$s in a cluster $C_v$ is upper bounded by $h=20 c x + 9$, and there are at least $y-h$ nodes in $P_v$ holding $0$s; by choosing $c$ small enough, we can obtain any desired ratio between the number of $1$s and the number of $0$s.
	
	\item When we encode the solution for $\Pi$: in this case, by picking $r$ large enough, we can make the encoded bit-string $B$ arbitrarily small compared to the size of $Z'$, that is, the nodes that encode $B$; note that the nodes in $Z'$ that do not take part in the encoding of $B$ have a $0$ assigned.
\end{itemize}

\section{Balanced orientations}\label{ssec:balanced-orientation}
In this section, we consider the problem of computing balanced orientations, that is, an orientation of the edges such that:
\begin{itemize}
	\item If a node has even degree, the number of incoming edges is the same as the number of outgoing edges.
	\item If a node has odd degree, then the number of incoming edges and the number of outgoing edges differ by $1$.
\end{itemize}
We start by considering the case in which all nodes have even degree. Note that, in this case, a balanced orientation always exists (though, without advice, the problem of computing such an orientation requires $\Omega(n)$ rounds, e.g., in a cycle). Later, we will show how to extend this result to the more general case. Finally, we will show that balanced orientations can be used to solve variants of the edge coloring problem.
We start by proving the following lemma.
\begin{lemma}\label{lem:balanced-orientation}
	Let $\mathcal{G}_\Delta$ be the family of graphs of maximum degree $\Delta$, where it is also satisfied that all nodes have an even degree. Let $\Pi_\Delta$ be the problem of orienting the edges of a given graph $G \in \mathcal{G}_\Delta$ such that each node $v$ satisfies $\deg_{\mathrm{in}}(v) = \deg_{\mathrm{out}}(v)$, where $\deg_{\mathrm{in}}(v)$ and $\deg_{\mathrm{out}}(v)$ are, respectively, the number of incoming and outgoing edges of $v$. Then, there exists a $(\mathcal{G}_\Delta, \Pi_\Delta, \gamma_0,A,T)$-composable advice schema, where $\gamma_0 = O(1)$, $A(c,\gamma) = \Theta(\gamma^3)$, and $T(\alpha,\Delta) = \Delta^{O(\alpha)}$.
\end{lemma}
By applying \Cref{lem:composable-to-1bit}, which allows us to convert the composable advice schema of \Cref{lem:balanced-orientation} into a fixed-length one, we obtain the following.
\begin{corollary}\label{cor:balanced-orientation-even}
	Let $\mathcal{G}_\Delta$ be the family of graphs of maximum degree $\Delta$, where it is also satisfied that all nodes have an even degree. Let $\Pi_\Delta$ be the problem of orienting the edges of a given graph $G \in \mathcal{G}_\Delta$ such that each node $v$ satisfies $\deg_{\mathrm{in}}(v) = \deg_{\mathrm{out}}(v)$.
	Then, there exists a uniform fixed-length sparse $(\mathcal{G}_\Delta, \Pi_\Delta,1,\Delta^{O(1)})$-advice schema.
\end{corollary}
We devote the rest of the section to prove \Cref{lem:balanced-orientation}.
Let $G$ be a graph in $\mathcal{G}_\Delta$.
We start with a simple observation. The balanced orientation problem can be easily solved in constant time if we are given an oracle for consistently orienting cycles. In fact, consider the following algorithm:
\begin{enumerate}
	\item Starting from $G$, nodes construct a virtual graph $G'$ where all nodes have degree $2$, as follows. Each node $v$ of degree $2d$ makes $d$ copies of itself. The copy number $i \in \{1,\ldots,d\}$ is incident to the $(2 i - 1)$-th and $2i$-th edges of $v$, taken in some arbitrary fixed order (e.g., by sorting the neighbors of $v$ by their IDs). Note that $G'$ is a collection of cycles.
	\item Use the oracle to consistently orient all cycles of $G'$.
	\item The orientation of the edges of $G'$ induces an orientation of the edges of $G$. Since each node of $G'$ has exactly one outgoing and one incoming edge, then each node $v$ of $G$ of degree $2d$ obtains exactly $d$ outgoing edges and $d$ incoming edges.
\end{enumerate}
Let $\beta = \gamma_0 = 2$. We show that, for any $\gamma \ge \gamma_0$, for any constant $c$, and for any $\alpha \ge \max\{\gamma^3 \beta / c, \gamma^3 \beta\}$, there exists a variable-length $(\mathcal{G}_\Delta, \Pi_\Delta,\beta, \Delta^{O(\alpha)} )$-advice schema satisfying that in each $\alpha$-radius neighborhood there are at most $\gamma_0$ bit-holding nodes. Since this condition satisfies the requirements of \Cref{def:composable}, we would obtain \Cref{lem:balanced-orientation}.

Let $G''$ be the graph obtained from $G'$ by removing all cycles of length at most $r$, for some parameter $r$ to be fixed later. 
Let $f : V(G'') \rightarrow V(G)$ be the function defined as follows: for each node $v$ in $G''$ that is a copy of node $u$ in $G$, $f$ maps $v$ to $u$.
We prove that there exists a subset $S''$ of the nodes of $G''$ satisfying the following properties, where $S := \{ f(v) \mid v \in S'' \}$.
\begin{enumerate}
	\item $S''$ is an $r$-dominating set of $G''$, that is, all nodes that are part of $G''$ are within distance $r$ (in $G''$) from a node of $S''$.
	\item If $v'',u'' \in S''$, then $\dist_G(f(v''),f(u'')) \ge 3\alpha$.
\end{enumerate}
If such a subset $S''$ exists, then we can encode an orientation of $G''$ (and hence produce an advice schema) as follows. Let $v'' \in S''$. By construction, $v''$ has degree $2$. Let $u''$ be an arbitrary neighbor of $v''$. Let $v$ and $u$ be, respectively, the preimages of $v''$ and $u''$ in $G$. Note that $v$ and $u$ are neighbors in $G$. We store $2$ bits on $v''$ and $1$ bit on $u''$ as follows.
\begin{itemize}
	\item The first bit of $v''$ and the first bit of $u''$ are set to $1$.
	\item The second bit of $v''$ is set to $1$ if the cycle is oriented from $v''$ to $u''$, and to $0$ otherwise.
\end{itemize}
Observe that bit-holding nodes that correspond to different nodes in $S$ are at distance at least $3\alpha - 2 \ge 2\alpha + \gamma^3 \beta - 2 \ge 2 \alpha + 1$, and hence each $\alpha$-radius neighborhood contains at most $2$ bit-holding nodes, as required.
A proper orientation for the cycles of $G'$ can be recovered as follows.
\begin{itemize}
	\item Nodes compute $G'$ (without communication), and spend $r$ rounds to check which of the cycles that pass through them have length at most $r$. 
	\item The edges of cycles of length at most $r$ can be consistently oriented without any advice. For example, a rule can be to find the node with the largest ID in the cycle, orient outgoing the edge that connects it with its neighbor that has larger ID, and then orient the rest of the edges consistently.
	\item For each cycle of $G'$ of length strictly larger than $r$ (and hence that it belongs to $G''$), we can find a pair of neighboring nodes in $G$ within distance $r$ that have both images in $G''$. Their bits encode the orientation of the cycle.
\end{itemize}
Thus, we obtain a balanced orientation for $G$ by spending $O(r)$ rounds. An example is depicted in \Cref{fig:orientations}. 
\begin{figure}
	(a)\includegraphics[width=0.6\textwidth,angle=90]{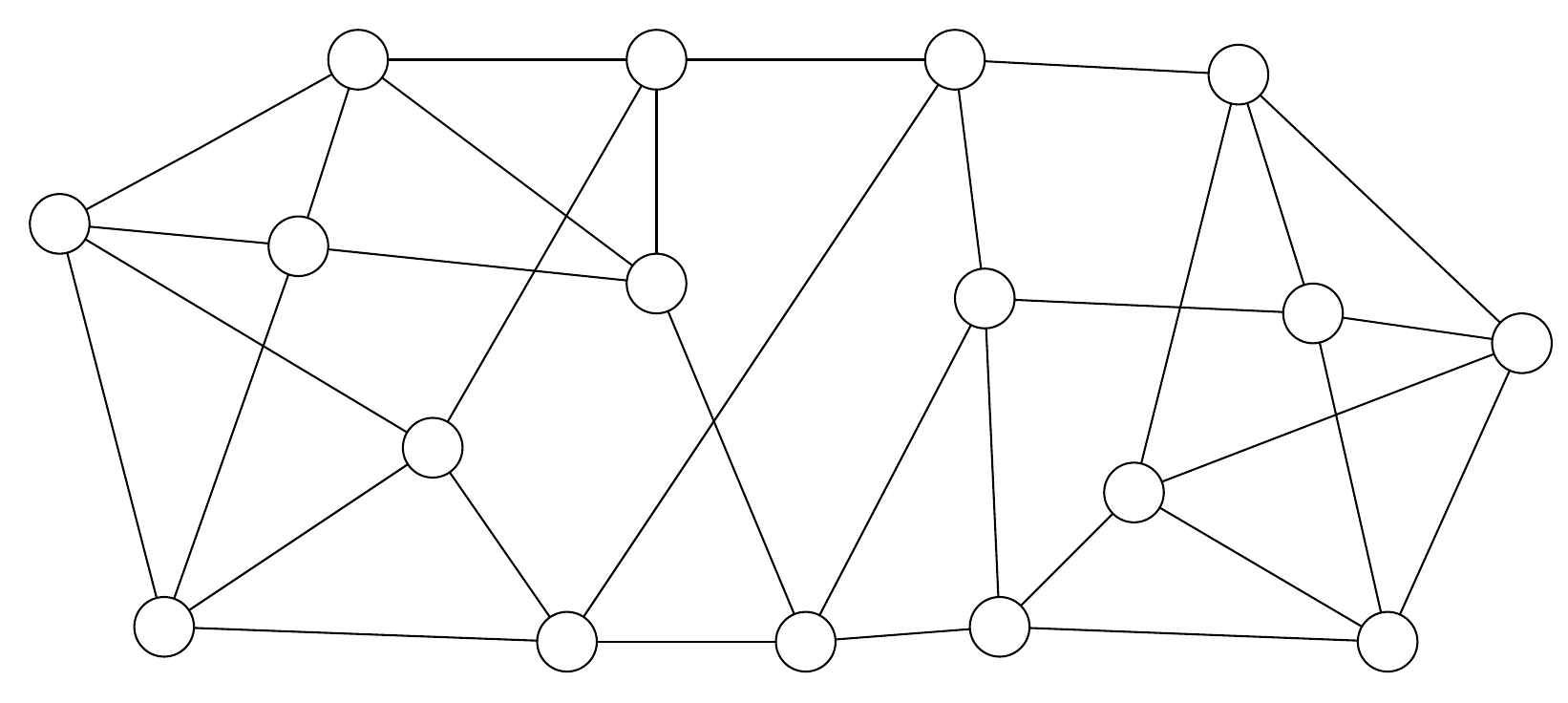}%
	\hspace{\stretch{1}}%
	(b)\includegraphics[width=0.6\textwidth,angle=90]{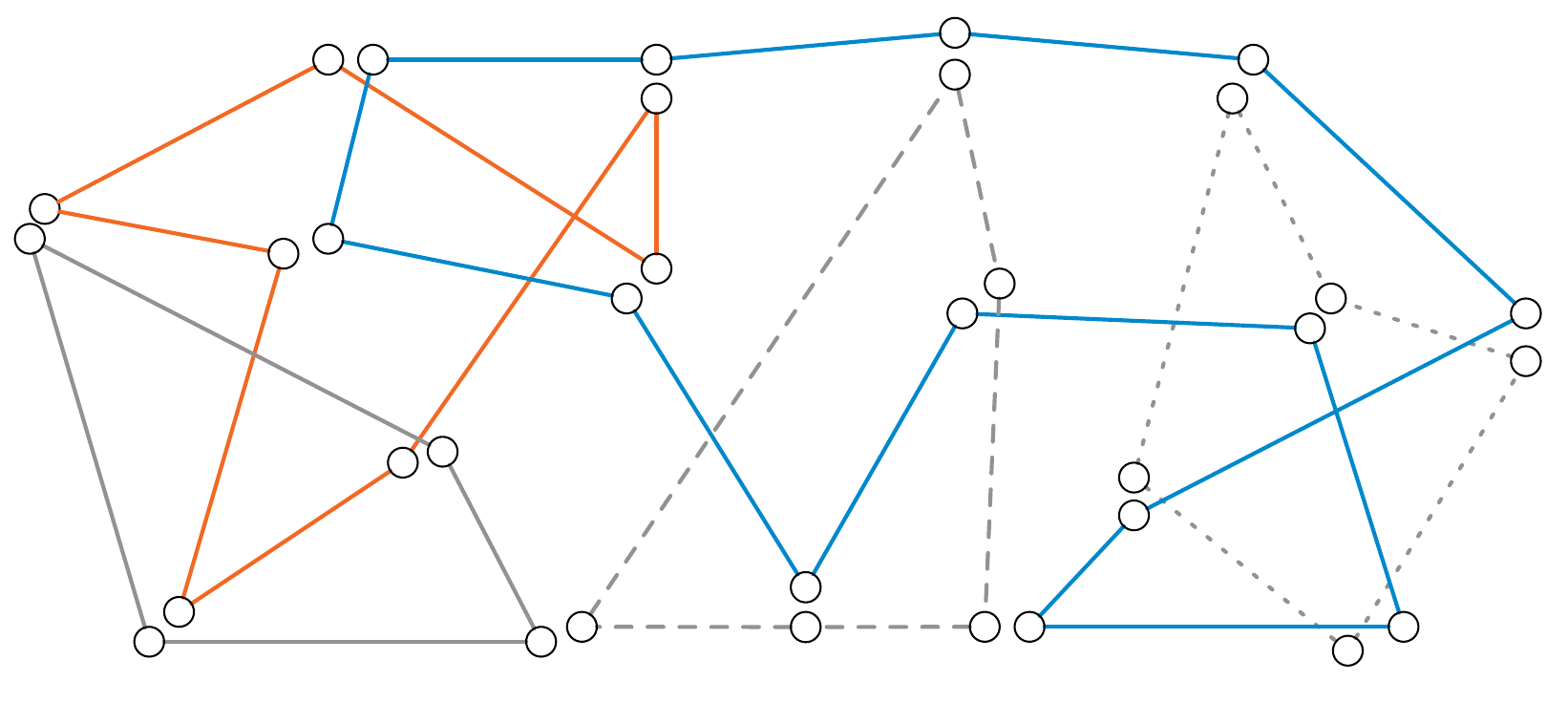}%
	\hspace{\stretch{1}}%
	(c)\includegraphics[width=0.6\textwidth,angle=90]{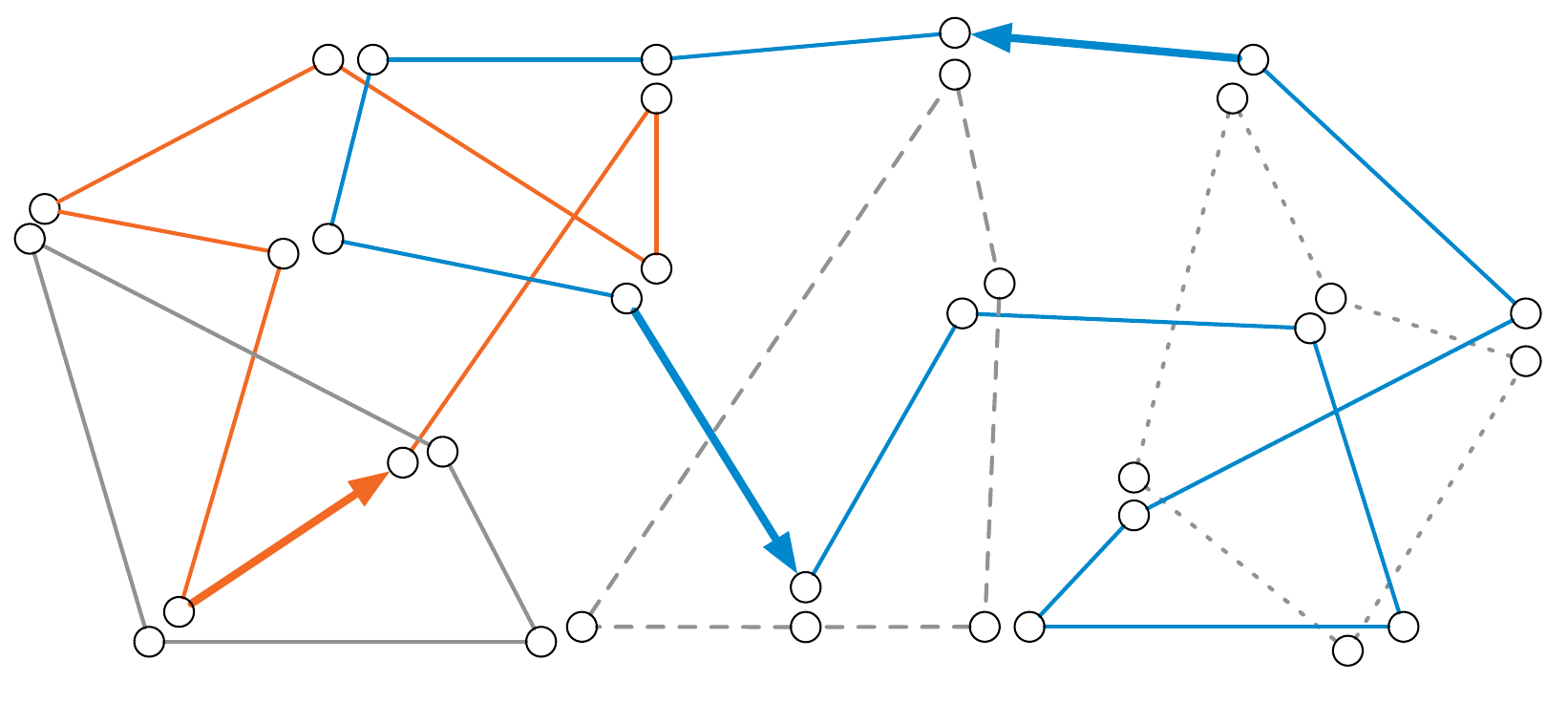}
	\caption{(a)~The original graph $G$. (b)~The graph $G'$, obtained by decomposing $G$ into cycles. By removing the three gray cycles (solid, dashed, and dotted), which are considered short and can be oriented without advice, we obtain $G''$.  (c)~For each directed edge, exactly one node belongs to $S''$. Such nodes receive advice encoding the orientation of these edges. We use the Lovász Local Lemma to prove that it is possible to choose $S''$ such that their corresponding images in $G$ are far enough from each other, and such that all nodes can recover the orientation of all the cycles they belong to by inspecting $G$ up to a bounded distance. }	\label{fig:orientations}
\end{figure}

The only remaining thing to be done is to prove the existence of $S''$ when $r = \Delta^{O(\alpha)}$. Let $\bar{r} = \varepsilon r$, for a small-enough constant $\varepsilon$ to be fixed later.
We start by computing a set $\bar{S}$ of nodes of $G''$ that satisfies the following.
\begin{itemize}
	\item All nodes of $G''$ are within distance $2 \bar{r}$ from a node in $\bar{S}$.
	\item The distance of any pair of distinct nodes in $\bar{S}$ is at least $2 \bar{r} + 1$.
\end{itemize}
Note that such a set can be computed greedily.
For small-enough $\varepsilon$, this set clearly satisfies Property 1, but it may not satisfy Property 2 and Property 3.
We now use the Lovász Local Lemma to show that the nodes of $\bar{S}$ can be moved along the cycles (to a distance of at most $\bar{r}$) in order to satisfy Property 2 and Property 3. Note that, for small-enough $\varepsilon$, Property 1 remains satisfied.

To each node $\bar{v}$ of $\bar{S}$ we assign a random variable, with value in $\{1,\ldots,\bar{r}\}$, that indicates how much we shift each node towards an arbitrary fixed direction. To each node $v$ of $G$ we assign an event that is bad when $v$ does not satisfy Property 2 or Property 3 after the shifting. Let $d$ be an upper bound on the number of events that depend on the same event, where two events are independent if they do not depend on at least one common variable. Observe the following.
\begin{itemize}
	\item The number of events that can depend on the same variable is bounded by $\Delta^{3\alpha+1} \bar{r}$. In fact, after shifting, a node can land in $\bar{r}$ possible different positions, and in each position, it affects at most $\Delta^{3\alpha+1}$ nodes of $G$, since in each $3\alpha$-radius neighborhood there are at most $\Delta^{3\alpha+1}$ nodes.
	\item The number of variables that affect the same event is bounded by $\frac{\Delta}{2} \Delta^{3\alpha+1}$. In fact, if we consider a node $v$ of $G$, in its $3\alpha$-radius neighborhood there are at most $\Delta^{3\alpha+1}$ nodes. For each node $u$ among these, there are at most $\Delta/2$ paths that pass through $u$ in $G'$. For each path $P$ that passes through $u$, since nodes of $\bar{S}$ are at distance at least $2 \bar{r} + 1$, and each node in $\bar{S}$ is shifted by at most $\bar{r}$, only one node of $\bar{S}$ that lies on $P$ can reach $u$ after being shifted.
\end{itemize}
We can upper bound $d$ as $\Delta^{3\alpha+1} \bar{r} \cdot \frac{\Delta}{2} \Delta^{3\alpha+1} \le \bar{r} \Delta^{6 \alpha + 3}$.
Moreover, we can upper bound the probability that a bad event happens on a node $v$ as follows. As mentioned before, the number of nodes that can be shifted into the $3\alpha$-radius neighborhood of $v$ is upper bounded by $\Delta^{3\alpha+1} \Delta/2$. The probability that a shifted node lands in the $3\alpha$-radius neighborhood of a node can be upper bounded by $\frac{\Delta^{3\alpha+1}}{\bar{r}}$, since a shifted node could land in $\bar{r}$ possible nodes, but there are at most $\Delta^{3\alpha+1}$ nodes in the $3\alpha$-radius neighborhood of $v$. A bad event happens on $v$ when at least two nodes are shifted in the $3\alpha$-radius neighborhood of $v$. Hence, by a union bound over all possible pairs, we can upper bound the failure probability on $v$ as \[p \le \binom{\Delta^{3\alpha+1} \Delta/2}{2} \left(\frac{\Delta^{3\alpha+1}}{\bar{r}}\right)^2 \le \frac{\Delta^{12 \alpha + 6}}{\bar{r}^2}.\] 
Note that $p d \le \Delta^{18\alpha + 9}  / \bar{r}$. Hence, we can pick $r \in \Delta^{O(\alpha)}$ such that $e p d \le e \Delta^{18\alpha + 9}  / \bar{r} < 1$. By the Lovász Local Lemma, this implies that an assignment of shifts satisfying the requirements exists.

\subparagraph{Extension to all degrees.} 
We can extend the proof of \Cref{lem:balanced-orientation} to the case in which nodes are allowed to have an odd degree and obtain an orientation in which each node $v$ satisfies $|\deg_{\mathrm{in}}(v) - \deg_{\mathrm{out}}(v)| \le 1$, as follows. Nodes can compute a virtual graph $G'$ that is a collection of paths and cycles, such that each node in $G$ is the endpoint on at most one path passing through it. Then, nodes can compute $G''$ by discarding paths that are shorter than $r$. After this, we use the same advice schema as in \Cref{lem:balanced-orientation} to obtain consistently oriented paths. By mapping the orientation of the paths to an orientation of the edges of $G$, we obtain $|\deg_{\mathrm{in}}(v) - \deg_{\mathrm{out}}(v)| \le 1$, as required.
Thus, we obtain the following corollary.
\begin{corollary}\label{cor:balanced-orientation-all}
	Let $\mathcal{G}_\Delta$ be the family of graphs of maximum degree $\Delta$. Let $\Pi_\Delta$ be the problem of orienting the edges of a given graph $G \in \mathcal{G}_\Delta$ such that each node $v$ satisfies $|\deg_{\mathrm{in}}(v) - \deg_{\mathrm{out}}(v)| \le 1$, where $\deg_{\mathrm{in}}(v)$ and $\deg_{\mathrm{out}}(v)$ are, respectively, the number of incoming and outgoing edges of $v$. Then, there exists a $(\mathcal{G}_\Delta, \Pi_\Delta, \gamma_0,A,T)$-composable advice schema, where $\gamma_0 = O(1)$, $A(c,\gamma) = \Theta(\gamma^3)$, and $T(\alpha,\Delta) = \Delta^{O(\alpha)}$.
\end{corollary}
\begin{corollary}
	Let $\mathcal{G}_\Delta$ be the family of graphs of maximum degree $\Delta$. Let $\Pi_\Delta$ be the problem of orienting the edges of a given graph $G \in \mathcal{G}_\Delta$ such that each node $v$ satisfies $|\deg_{\mathrm{in}}(v) - \deg_{\mathrm{out}}(v)| \le 1$.
	Then, there exists a uniform fixed-length sparse $(\mathcal{G}_\Delta, \Pi_\Delta,1,\Delta^{O(1)})$-advice schema.
\end{corollary}

\subparagraph{Extension to splittings on bipartite graphs.}
Let $\mathcal{G}_\Delta$ be the family of graphs maximum degree $\Delta$, where all nodes have even degree. Let $\Pi$ be the problem of computing an edge coloring of $G$ with $2$ colors, say red and blue, such that, for each node $v$ of degree $d$, $v$ has the same number of incident red and blue edges.  This problem is called \emph{splitting}. We show that, for the problem $\Pi$, there exists a composable schema.

Observe that, if we are given a balanced orientation and a $2$-coloring of the nodes, then we can compute a splitting by coloring red the edges outgoing from white nodes and blue the edges outgoing from black nodes. Moreover, if we modify the schema for balanced orientations so that, on each bit-holding node, we add one bit that denotes the color of the nodes, then the nodes are also able to recover a $2$-coloring of the graph. Thus, we obtain the following.
\begin{corollary}\label{cor:splitting}
Let $\mathcal{G}_\Delta$ be the family of bipartite graphs maximum degree $\Delta$, where all nodes have even degree.
	Let $\Pi_\Delta$ be the splitting problem on graphs of $\mathcal{G}_\Delta$.
	Then, there exists a $(\mathcal{G}_\Delta, \Pi_\Delta, \gamma_0,A,T)$-composable advice schema, where $\gamma_0 = O(1)$, $A(c,\gamma) = \Theta(\gamma^3)$, and $T(\alpha,\Delta) = \Delta^{O(\alpha)}$.
\end{corollary}
\begin{corollary}
Let $\mathcal{G}_\Delta$ be the family of bipartite graphs maximum degree $\Delta$, where all nodes have even degree.
	Let $\Pi_\Delta$ be the splitting problem on graphs of $\mathcal{G}_\Delta$.
	Then, there exists a uniform fixed-length sparse $(\mathcal{G}_\Delta, \Pi_\Delta,1,\Delta^{O(1)})$-advice schema.
\end{corollary}

\subparagraph{\boldmath Extension to $\Delta$-edge coloring on bipartite $\Delta$-regular graphs, when $\Delta$ is a power of $2$.}
Let $\mathcal{G}_\Delta$ be the family of bipartite $\Delta$-regular graphs, where $\Delta$ is a power of $2$. Let $\Pi$ be the problem of computing a $\Delta$-edge coloring. We show that, by recursively using the schema for splitting, we obtain the following corollaries.
\begin{corollary}
	Let $\mathcal{G}_\Delta$ be the family of bipartite $\Delta$-regular graphs, where $\Delta$ is a power of $2$. Let $\Pi$ be the problem of computing a $\Delta$-edge coloring.
	Then, there exists a $(\mathcal{G}_\Delta, \Pi_\Delta, \gamma_0,A,T)$-composable advice schema, where $\gamma_0 = O(\Delta)$, $A(c,\gamma) = \Theta((
	\Delta \gamma)^3 \log \Delta)$, and $T(\alpha,\Delta) = \Delta^{O(\alpha)}$.
\end{corollary}

\begin{corollary}
	Let $\mathcal{G}_\Delta$ be the family of bipartite $\Delta$-regular graphs, where $\Delta$ is a power of $2$. Let $\Pi$ be the problem of computing a $\Delta$-edge coloring.
	Then, there exists a uniform fixed-length composable $(\mathcal{G}_\Delta, \Pi,1,f(\Delta))$-advice schema for some computable function $f$.
\end{corollary}

For each $1 \le i \le \log \Delta$, we define the problem $\Pi_i$ as the problem of computing an edge coloring of $G$ with $2^i$ colors, such that each color-class induces a $\Delta/2^i$-regular subgraph. Observe that, by assumption, for each $i$, $\Delta/2^i$ is an integer. 
Note that $\Pi_1$ is exactly the splitting problem, and solving $\Pi_{i+1}$ given a solution for $\Pi_i$ is exactly the problem of computing a splitting for each subgraph given by the solution for $\Pi_i$. Observe that $\Pi_{\log \Delta}$ is exactly the problem of computing a $\Delta$-edge coloring.
Thus, we can solve $\Pi_1$ given advice for splitting, and we can then solve $\Pi_2$ given a solution for $\Pi_1$ and two separate advices for splitting, one for each subgraph given by $\Pi$. In total, in order to solve $\Pi_{\log \Delta}$, we thus need $O(\Delta)$ advices for splitting. Thus, we obtain the claimed results by applying \Cref{lem:compose} with $k = O(\Delta)$ on the schema given by \Cref{cor:splitting}.

\section{\boldmath\texorpdfstring{$\Delta$}{Delta}-coloring of \texorpdfstring{$\Delta$}{Delta}-colorable graphs}\label{ssec:delta-col}

Let $G$ be a graph with maximum degree $\Delta$. We show that any such graph that is $\Delta$-colorable can be colored with $\Delta$ colors in $O(\log \Delta \cdot \sqrt{\Delta \log \Delta})$ rounds using $1$ bit of initial advice. More precisely, we design a composable advice schema for the $\Delta$-coloring problem, and then we apply \Cref{lem:composable-to-1bit} to obtain a schema that uses $1$ bit per node.

\begin{theorem}\label{thm:delta_coloring_scheme}
	Let $\mathcal{G}_\Delta$ be the family of $\Delta$-colorable graphs. Let $\Pi_\Delta$ be the $\Delta$-vertex coloring problem. Then, there exists a $(\mathcal{G}_\Delta,\Pi_\Delta,\gamma_0,A,T)$-composable advice schema, where $\gamma_0 = O(1)$ and $A(c,\gamma) = \Theta(\gamma^3)$ and $T(\alpha, \Delta) = \max(O(\alpha^2 \log \Delta), O(\alpha \log \Delta \cdot \sqrt{\Delta \log \Delta}))$.
\end{theorem}
\begin{corollary}
	Let $\mathcal{G}_\Delta$ be the family of $\Delta$-colorable graphs. Let $\Pi_\Delta$ be the $\Delta$-vertex coloring problem. Then, there exists a uniform fixed-length $(\mathcal{G}_\Delta,\Pi_\Delta,1,T(\alpha, \Delta))$-advice schema with $T(\alpha, \Delta) = \max(O(\log \Delta), O(\log \Delta \cdot \sqrt{\Delta \log \Delta}))$.
\end{corollary}

The remaining part of this section is dedicated to proving \cref{thm:delta_coloring_scheme}. The formal definition of advice schema and composability are presented in \Cref{def:advice-schema} and \Cref{def:composable}, respectively. To recall, in order to prove that a composable schema exists, for any given $\gamma \ge \gamma_0$ (and hence in our case a large-enough constant), and for any given $\alpha \ge A(c,\gamma)$, we need to provide a schema of distributing advice bits to the vertices of the graph such that, in any ball of radius $\alpha$ there are at most $\gamma_0$ bit-holding nodes, and bit-holding nodes must have at most $c \alpha / \gamma^3$ bits each. Moreover, given such a bit assignment, there must exist an algorithm that can solve $\Delta$-coloring in $\max(O(\alpha^2 \log \Delta), O(\alpha \log \Delta \cdot \sqrt{\Delta \log \Delta})$ time.

To that end, we split the problem of $\Delta$-coloring into three subproblems, one of them can be solved with a distributed algorithm of bounded locality (i.e., with a runtime that solely depends on $\Delta$), and for the two remaining ones we design composable schemas. The three independent parts of our algorithm are as follows.
\begin{enumerate}
\item First, the algorithm computes an $O(\Delta^2)$-vertex coloring in $O(\log \Delta)$ time, using an algorithm based on network decomposition which leverages some additional advice.
\item Then, the algorithm reduces the number of colors to $\Delta +1$, by using a standard distributed coloring algorithm, which runs in $O(\sqrt{\Delta \log \Delta})$ rounds when an initial $O(\Delta^2)$-coloring is given \cite{FHK16,BarenboimEG18,MausT22}.
\item Finally, the algorithm reduces the number of colors from $\Delta+1$ to $\Delta$, inspired by distributed algorithms from \cite{PS92, GHKM18}, leveraging the advice bits to terminate in just $O(\log \Delta \cdot \sqrt{\Delta \log \Delta})$ rounds. 
\end{enumerate}
The second step is a simple application of distributed algorithms with no advice, while the first and third algorithms rely on some given advice. Below, we formally state the problems to be solved in steps 1 and 3,  and then we design composable schemas for them. Later, we will compose the schemas into a single one by applying \Cref{lem:compose}, obtaining \Cref{thm:delta_coloring_scheme}.

\subsection{\texorpdfstring{\boldmath $O(\Delta^2)$}{O(Delta²)}-vertex coloring}\label{subsection:initial_coloring}
In this section, we consider the following problem.
\begin{problem}[$O(\Delta^2)$-coloring] \label{problem:initial_coloring}
	Let $G = (V, E)$ be a graph with maximum degree $\Delta$. The goal is to find an assignment of colors $C: V \rightarrow \{1, 2, \dots, O(\Delta^2) \}$ such that there are no two vertices $u, v$ such that $\{u, v\} \in E \wedge C(u) = C(v)$.
\end{problem}
\noindent For this problem, we prove the following lemma.
\begin{lemma}\label{lem:problem:initial_coloring}
	Let $\Pi_\Delta$ be the $O(\Delta^2)$-coloring problem. Then, there exists a $(\mathcal{G}_\Delta,\Pi_\Delta,\gamma_0,A,T)$-com\-pos\-able advice schema, where $\gamma_0 = O(1)$, $A(c,\gamma) = \Theta(\gamma^3)$, and $T(\alpha,\Delta) = O(\alpha^2 \log \Delta)$.
\end{lemma}
\begin{proof}
We fix $\beta = \gamma_0 = 3$, and we show that, for any $\gamma \ge \gamma_0$, for any constant $c$, and for any $\alpha = \max\{\gamma^3 \beta / c, \gamma^3 \beta\}$, there exists a variable-length $(\mathcal{G}_\Delta,\Pi_\Delta,\beta,O(\alpha^2 \log \Delta))$-advice schema satisfying that in each $\alpha$-radius neighborhood there are at most $\gamma_0$ bit-holding nodes. 

The algorithm roughly follows the approach that can be used to color a graph via a network decomposition, which can be described as follows. Let us consider a partition of vertices into disjoint clusters, each cluster having only $\Delta^{O(\alpha^2 \log\Delta)}$ adjacent clusters and a small diameter. In order to find a proper $\Delta^{O(\alpha^2 \log\Delta)}$ coloring, we can use a proper $\Delta+1$ coloring of vertices in each cluster, and a proper $\Delta^{O(\alpha^2 \log\Delta)}$ coloring of the cluster graph (i.e., a coloring of the graph where each cluster is contracted to a point, or, in other words, an assignment of colors to clusters such that adjacent clusters have different colors). The color of a vertex $v$ in a cluster $C$ is the tuple $($color of $v$ in $C$, color of $C)$. The first part of the tuple ensures that neighboring vertices in the same cluster have different colors, and the second part ensures that vertices in two adjacent clusters have different colors. The total number of colors we use is still $\Delta^{O(\alpha^2 \log\Delta)}$. The number of colors we use can be then reduced to $O(\Delta^2)$ by using a few rounds of Linial's coloring algorithm \cite{Linial92}.

\begin{lemma}\label{lem:linial_step}\cite{Linial92}
Let us consider a graph with maximum degree $\Delta$ and some coloring of its vertices with colors $1, 2, \dots, c$. There exists $2$ round distributed LOCAL algorithm that finds a proper $O(\Delta^2)$ coloring, if $c \in \poly(\Delta)$, or an $O(\Delta^2 \log c)$ coloring otherwise.
\end{lemma}

Given that the total number of colors before we start using the reduction described in \cref{lem:linial_step} is $\Delta^{O(\alpha^2 \log\Delta)}$, we need $O(\log^* \alpha)$ rounds to reduce the number of colors down to $O(\Delta^2)$. 

In the remaining part of the proof of \cref{lem:problem:initial_coloring} we show that we can define an adequate clustering, a distribution of advice bits, and an algorithm that using these advice bits finds a proper $\Delta^{O(\alpha^2 \log\Delta)}$ coloring of the cluster graph.

In the following, the degree of a cluster is the number of edges with exactly one endpoint in the cluster.
The idea of the advice schema of this section is similar to the idea behind the schema for graphs with sub-exponential growth. Essentially, we show that we can define a clustering with clusters that have sufficiently many internal vertices to encode the color of the cluster.  However, to provide such a clustering, we rely on a slightly different idea, namely: we split all clusters into $O(\alpha^2 \log \Delta)$ buckets, and let the $i$th bucket contain every cluster whose degree is in the interval $[\Delta^{i-1}, \Delta^i)$. Each bucket receives its own disjoint palette of colors; the $i$th bucket receives a palette of size $\Delta^i$. The clusters from the $i$th bucket will be colored only with those colors, and as such the number of bits describing such a color is at most $(i+1) \log \Delta$. For any $\Delta >1$, the total number of colors in bucket $i$ and all smaller buckets is $\sum_{j=1}^{i} \Delta^j < \Delta^{i+1}$. 

In the following, by $\alpha$-independent set we denote a set of nodes that are at pairwise distance at least $\alpha$.
Generally, our goal is to design the clustering in a way that guarantees that each cluster assigned to bucket $i$ contains an $\alpha$-independent set of nodes that are capable of storing all the required $(i+1) \log \Delta$ bits. However, we also need to have some special marker denoting cluster centers, and we need to take into account that vertices from the $\alpha$-independent set cannot be too close to the border of the cluster. Nevertheless, these are some low-level details that can be easily taken care of by using clusters with radius strictly larger than $\alpha+1$.

To define the clustering, we proceed as follows. Let $r = 100 \alpha^2 \log \Delta$. We pick an $(r,r)$-ruling set $I$ from $G$. Then, we assign each vertex from $G$ to the closest vertex from $I$ (breaking ties in an arbitrary consistent manner, e.g., by joining the cluster of the smallest ID neighbor). That concludes the first part of our construction. Currently, we have clusters satisfying that all nodes that are within distance $r/2 - 1 = 50 \alpha^2 \log \Delta - 1$ from a cluster center $v$ are guaranteed to belong to the cluster of $v$, and such that the degree of the clusters is bounded from above by $\Delta ^{r+1}$. However, for our advice schema, we need the clusters with degrees larger than $\Delta^{10 \alpha}$ to have at least $100 \alpha^2 \log^2 \Delta+2$ internal vertices that are at least $\alpha$ hops away from the border of the cluster. 

To find such a clustering without the high-degree low-volume clusters, we remove some of the vertices from $I$, which we call the refining process. To that end, let us consider clustering $C$ we just built around the nodes from $I$. Let $I'$ be the set of centers of clusters from $C$. In our refining process, we gradually remove the clusters from $C$:
\begin{itemize}
\item initially $I'$ is equal to $I$, but it stops being true as soon as we start removing clusters from $C$,
\item during the process, $C$ is only a partial clustering, i.e., not all vertices belong to a cluster.
\end{itemize}
We deem a cluster to be \emph{broken} when it has degree at least $\Delta^{10 \alpha}$ but does not have at least $\Delta^{9 \alpha-2}$ vertices of distance at least $\alpha$ from its border. In the refining process, we consider all broken clusters one by one, and we `repair' each broken cluster by destroying its neighbors, and possibly destroying other nearby clusters. Whenever we consider a broken cluster $B$ that is still in $C$, we remove from $C$ all other clusters ($\neq B$) that contain at least one vertex whose distance from the border of $B$ is smaller than $2\alpha+1$. Thus, the set $I'$ still contains the center $i_B$ of said broken cluster, but has lost one or more other cluster centers.

In the end, we are left with a set of clusters that have survived the refining process. Furthermore, we have a set of vertices that no longer belong to any cluster as their cluster was removed from $C$ by the refining process. Those vertices join the closest high-degree cluster that is still present in $C$, and we later call this part of the refining process the re-partition step.

\begin{lemma}\label{lem:volumes}
In the new clustering, broken clusters that survived the refining process have at least $\Delta^{9 \alpha-2}$ nodes that are at least $\alpha$ hops from their new border.
\end{lemma}
\begin{proof}
For any broken cluster $B$ that is not removed from $C$, the refining process removed all clusters (in particular all high-degree clusters) with vertices at a distance at most $2\alpha+2$ from the border of $B$. As such, the re-partition step assigns to $B$ all vertices that are at distance at most $\alpha$ from B's border, as any other cluster is at least $\alpha+1$ hops away, and $B$ is at most $\alpha$ hops away. Thus, all vertices that were originally in $B$ are now at least $\alpha$ hops away from their new border. In conclusion, $B$ contains at least $\Delta^{10\alpha}$ nodes that are at least $\alpha$ hops away from their new border.
\end{proof}

On a high level, our advice schema uses marker $111$ to denote vertices in $I$, marker $11$ to denote each cluster center (vertices from $I'$) and then assigns $1$ bit on a sparse set of vertices of the cluster, to encode the color of the cluster.
First, we describe the distribution of bit-holding vertices inside clusters. Then, we describe the algorithm that using this advice can recover the color of each cluster in 
$O(\alpha^2 \log \Delta)$ rounds.

\subparagraph{\boldmath Clusters with degree at most $\Delta^{10 \alpha}$.} 
If the degree of a cluster is smaller than $\Delta^{10 \alpha}$, we need $10 \alpha \log \Delta$ bits of advice to encode the color of this cluster. If the cluster has a radius strictly smaller than $r/2 - 1$, then it means that the whole graph is contained in the cluster, and hence no advice is needed to compute a coloring. Hence, assume that the radius is at least $r/2-1$. Consider a BFS tree starting from the center $v$ of the cluster. By assumption, there exists a path that lies on the BFS tree, that does not contain $v$ and any of its neighbors, that does not contain nodes that have edges connected to nodes outside the cluster and all nodes at distance at most $\alpha$ from them, and that contains at least $r/2 - 4 - \alpha$ nodes. By picking an $(\alpha,\alpha)$-ruling set on such a path, we obtain a set of size at least $(r/2-4 - \alpha) / (2\alpha+1) = (50\alpha^2 \log \Delta - 4 - \alpha) / (2 \alpha+1) > 10 \alpha \log \Delta$ vertices that can be used in our advice schema. We sort these nodes by their IDs, and we assign to them, in order, the bits of the color of the cluster.

\subparagraph{\boldmath Clusters with degree larger than $\Delta^{10 \alpha}$.} 
By \cref{lem:volumes} we know that all clusters of degree larger than $\Delta^{10 \alpha}$ have at least $\Delta^{9 \alpha-2}$ internal vertices that are $\alpha$ hops away from their border. They either had it in the clustering around $I$, or they are clusters that survived the refining process and \cref{lem:volumes} guarantees those internal vertices exist.   

In such a large set of vertices, there exists an $\alpha$-independent set of size at least $\Delta^{9 \alpha-2} / \Delta^{\alpha+1} \allowbreak > \Delta^{7\alpha}$, for $\Delta>1$.  Such an independent set can be found e.g. by a greedy algorithm. Adding a constraint that vertices cannot be too close to the center of the cluster takes away $1$ vertex from the set (the greedy algorithm used for construction can start from choosing the center). Therefore, we have at least $\Delta^{7\alpha}$ vertices, which for $\Delta>1$ and sufficiently large $\alpha$ is larger than $50 \alpha^2 \log^2 \Delta+2$.

\subparagraph{Decoding algorithm.}
Firstly, the algorithm identifies the clustering -- all vertices with advice indicating that they belong to $I$ initiate a BFS computation, each vertex $v$ joins cluster of the initiator $i$ of the first BFS that reached $v$ (breaking ties in favor of the smallest neighbor ID).
This gives the initial partition into clusters. 
Then, the clusters act upon the advice denoting whether they are in $I'$ - if they do not belong to $I'$, their vertices join the closest high-degree cluster.

That results in the clustering as used by our advice schema. Then, the vertex can identify all bit-holding vertices, and read bits in the same order that was used by our encoding schema. This allows a center to identify a binary representation of the color of the cluster in a proper coloring of the cluster graph.

Finally, as mentioned, to compute a coloring of vertices, each cluster-center can compute a coloring of the vertices in its cluster (it knows the topology of the cluster), and broadcast the computed colors to all vertices in the cluster. Overall, the time needed for computation is proportional to the radius of the clusters, which is $O(\alpha^2 \log \Delta)$.
\end{proof}

\subsection{\texorpdfstring{\boldmath $\Delta$}{Delta}-coloring}\label{ssec:delta_coloring}

Below, we present a composable advice schema that given $\Delta+1$ coloring of $\Delta$-colorable graph, finds a $\Delta$-coloring.

\begin{problem} \label{problem:reduce_to_delta}
We are given a $\Delta$-colorable graph $G$ with maximum degree $\Delta$,  with a proper $\Delta+1$ coloring $C$. The goal is to compute a proper $\Delta$-coloring of $G$.
\end{problem}

\begin{lemma}\label{lem:reduce_to_delta}
	Let $\mathcal{G}_{\Delta}$ be the set of all $\Delta$-colorable graphs with vertices colored according to some proper $\Delta+1$ coloring. Let $\Pi_\Delta$ be the $\Delta$-vertex coloring problem, given that on the input vertices are given a color from some proper $\Delta+1$ coloring. Then, there exists a $(\mathcal{G}_\Delta,\Pi_\Delta,\gamma_0,A,T)$-composable advice schema, where $\gamma_0 = O(1)$, $A(c,\gamma) = \Theta(\gamma^3)$, and $T(\alpha,\Delta) = O(\alpha \log \Delta \cdot \sqrt{\Delta \log \Delta})$.
\end{lemma}

The remaining part of \cref{ssec:delta_coloring} is dedicated to the proof of \cref{lem:reduce_to_delta}. 

The algorithm in our schema relies on ideas similar to those used in \cite{PS92}. On a somewhat abstract level, both algorithms start with uncoloring all vertices with color $\Delta+1$. Then, both have two steps -  first they perform some local computations that reduce the set uncolored vertices to a set of uncolored vertices that are far apart. Then, given that the uncolored vertices are far apart, they can be simultaneously handled in a fairly simple way. The difference between our algorithm and that of \cite{PS92} is in what is considered \emph{far apart} and what can be handled \emph{in a fairly simple way}.

In the remaining part of the proof of \cref{lem:reduce_to_delta}, we first recall parts of the algorithm from \cite{PS92}, and only afterwards, we describe our advice schema. Even though we do not implement the algorithm from \cite{PS92} directly, we recall it first as we use some building blocks of that algorithm to define the behavior of our algorithm and claim that it is indeed correct.

\subparagraph{\boldmath Some details of distributed $\Delta$-coloring algorithm \cite{PS92}.}
In \cref{lem:simple_dist_brooks}, we formulate the observation from \cite{PS92} that is the most relevant for designing our $\Delta$-coloring advice schema.

\begin{lemma}\label{lem:simple_dist_brooks}
Let us consider a graph $G$ and a partial proper $\Delta$-coloring, in which uncolored vertices form a $2c \log_{\Delta} n$-independent set $I_1$. Then, there exists a set of vertices $X$ in $G$ such that:
\begin{itemize}
	\item each vertex from $I$ has exactly one vertex form $X$ at distance at most $c \log_{\Delta} n$,
	\item vertices from $X$ either have a degree smaller than $\Delta$ or have two neighbors with the same color in the given partial coloring, none on the shortest path from $I$ to $X$,
	\item shifting colors along the path, in the direction from $X$ to $I$ and recoloring vertices from $X$ results in proper $\Delta$-coloring.
\end{itemize}
\end{lemma}

The set $I$ from \cref{lem:simple_dist_brooks} in the informal description above is the set of vertices that are \emph{far apart}. \cref{lem:simple_dist_brooks} essentially says that if we can reduce the problem of extending partial $\Delta$-coloring (with no uncolored neighboring vertices) to a variant of the same problem where any two uncolored vertices are at least $2c \log_{\Delta} n$ hops from each other, then it can be solved in additional $c \log_{\Delta} n$ rounds. Then, the authors of \cite{PS92} provide the required reduction, which concludes their $\Delta$-coloring algorithm.

The reduction starts with choosing an $\Omega(\log_{\Delta} n), O(\log^2_{\Delta} n)$ ruling set $R$, on the set of uncolored vertices (both vertices and distance constraints are held only on uncolored vertices, but the distance is defined by shortest paths using all edges). Then, each vertex finds a shortest path to the closest vertex from set $R'$, which is $R$ extended by vertices of degree smaller than $\Delta$. Now, using these shortest path trees, the algorithm recolors  vertices, considering vertices in a leaf-to-root order. Whenever an uncolored vertex is considered, it uncolors its parent, and then chooses a color from the palette of colors not used by its neighbors.

The technical problem here is that neighboring uncolored vertices may need to choose color simultaneously, and as such some symmetry breaking is needed. To that end, the authors of \cite{PS92} use a randomized approach, that in expected $O(\log n)$ rounds can process one layer of the trees. The authors of \cite{GHKM18} handle an analogous problem by a slightly different approach, and employ \cref{thm:list_coloring} to handle that issue.
\begin{theorem}[\cite{FHK16,BarenboimEG18,MausT22}]\label{thm:list_coloring} 
There is a deterministic distributed algorithm that given proper $O(\Delta^2)$ vertex coloring solves the $({\deg}+1)$-list coloring problem in time $O(\sqrt{\Delta \log \Delta})$.
\end{theorem}

\subparagraph{Our advice schema.}
While it's possible to design a schema for the \cref{problem:reduce_to_delta} in a direct way, considering the problems of recoloring vertices from $I$ and all other uncolored vertices separately allows us to focus on one difficulty at a time. For each of the two subproblems we design a composable advice schema, and the final schema from \cref{lem:reduce_to_delta} is obtained by composing them using \cref{lem:compose}. 

First, we show the part of our algorithm that reduces the general problem to the problem of coloring vertices from some $\Omega(\log_\Delta n)$ independent set $I$. To recall, in order to handle this subproblem, the algorithm from \cite{PS92} solves a sequence of $O(\log^2 n)$ list coloring problems. By a sequence of $O(\log^2 n)$ problems, we mean that we cannot solve them simultaneously, as input for some of those problems depends on the output of other problems, and those dependencies can form chains of length $O(\log^2 n)$. We leverage advice in a way, that allows us to solve some of those problems simultaneously, by encoding the answers to a subset of those problems, in a way that allows running more of the list coloring algorithms in parallel.

\begin{problem} \label{problem:reduce_to_roots}
We are given a $\Delta$ colorable graph $G$ with maximum degree $\Delta$, with a proper $\Delta+1$ coloring $C$. The goal is to compute a proper partial $\Delta$-coloring of $G$, in which any uncolored vertices are at a distance at least $2c \log_{\Delta} n$, where $c$ is the constant $c$ from \cref{lem:simple_dist_brooks}.
\end{problem}

\begin{lemma}\label{lem:reduce_to_roots}
	Let $\mathcal{G}_{\Delta}$ be the set of all $\Delta$-colorable graphs with vertices colored according to some proper $\Delta+1$ coloring. Let $\Pi_\Delta$ be the partial $\Delta$-vertex coloring problem described in \cref{problem:reduce_to_roots}. Then, there exists a $(\mathcal{G}_\Delta,\Pi_\Delta,\gamma_0,A,T)$-composable advice schema, where $\gamma_0 = O(1)$, $A(c,\gamma) = \Theta(\gamma^3)$, and $T(\alpha,\Delta) = O(\alpha \log \Delta \cdot \sqrt{\Delta \log \Delta})$.

\end{lemma}
\begin{proof}

We fix $\beta =3, \gamma_0 = 2$, and we show that, for any $\gamma \ge \gamma_0$, for any constant $c$, and for any $\alpha = \max\{\gamma^3 \beta / c, \gamma^3 \beta\}$, there exists a variable-length $(\mathcal{G}_\Delta,\Pi_\Delta,\beta,O(\alpha \log \Delta \cdot \sqrt{\Delta \log \Delta}))$-advice schema satisfying that in each $\alpha$-radius neighborhood there are at most $\gamma_0$ bit-holding nodes.

The general idea is to consider the same shortest path spanning forest used in \cite{PS92}. However, instead of processing the trees directly, we choose a set of relay vertices, and regular uncolored vertices use them to define another set of shortest path trees, later used for recoloring.

A proper set of relay vertices can be found as follows. Let us consider an $\Omega(\log_{\Delta} n)$, $O(\log^2_{\Delta} n)$ ruling set $R$, and the shortest path spanning forest $F$ as in the distributed algorithm \cite{PS92}. To define the relay vertices, we consider vertices in $F$ in layers.

The initial layer is just the set of vertices in $R$. Then, each layer is a $(2\alpha+22) \log \Delta$ maximal independent set of vertices, at distance $(2\alpha+22)\log \Delta$ from the vertices in the last layer. Such selection guarantees that:
\begin{itemize}[noitemsep]
	\item for each relay vertex there is some space to deposit advice bits,
	\item relay vertices from layer $i$ are not too far from the relay vertices from layer $i-1$,
	\item each uncolored vertex is close to a relay vertex.
\end{itemize}

\emph{Advice schema:} to provide all needed advice, it is enough for us to mark all the relay vertices with $11$, all vertices from the ruling set with $111$, and additionally encode the color in the resulting coloring of each of those vertices. Since the vertices are  $(2\alpha+22) \log \Delta$ hops away from each other, each can use the part of a path between the nodes of length $(\alpha+11) \log \Delta$ to store one bit per vertex, leaving $2\alpha$ hops of free vertices between the bit holding vertices. Overall, that gives more than $\log \Delta$ bits, which is sufficient to encode the color of the marked vertex.

\emph{Algorithm:} Our algorithm can first recolor all vertices in local neighborhoods of relay vertices. To that end, all vertices explore their $O(\alpha \log \Delta)$ neighborhoods and reconstruct the shortest path trees around the relay nodes. Then, the algorithm uses the same technique as in \cite{PS92} (process vertices on paths, leaf to root), but instead of a randomized step that breaks the symmetry, it uses \cref{thm:list_coloring} as in the algorithm from \cite{GHKM18}. In total, this step is the most time-consuming part of the algorithm and takes $O(\alpha \log \Delta \cdot \sqrt{\Delta \log \Delta})$ time. The result is that all non-relay nodes that are not in $R$ were colored.

Then, the relay vertices can assign to themselves a new color (decoded from the advice bits scattered in its $O(\alpha \log \Delta)$ hop neighborhood), which leaves us with the task of re-coloring the vertices on the paths between the relay nodes. Each relay node knows its ancestor in the previous layer of relays, as such the vertices on the paths that need recoloring can be identified, and new colors can be assigned. 
\end{proof}

\subparagraph{Fixing coloring of root vertices.}
Finally, let us consider the last subproblem, of extending partial $\Delta$-coloring into proper $\Delta$-coloring, given that the uncolored vertices are at a distance larger than $\Omega(\log_{\Delta} n)$. The standard algorithms use a somewhat brute-force approach, that for each uncolored vertex finds a sequence of recolorings that if applied, results in a proper $\Delta$-coloring. By \cref{lem:simple_dist_brooks}, a set of such nonoverlapping sequences can be found by exploring $O(\log_{\Delta} n)$ hop neighborhoods of all uncolored vertices. Once again, one can see this sequence of recolorings as a sequence of dependent problems and the advice can be used to both define the sequence and make the dependency chains shorter.

\begin{problem} \label{problem:fix_root_colors}
We are given a graph $G$ with partial $\Delta$-coloring $C$, such that any two vertices $u,v$ such that $C(v) = \bot$ and $C(u) = \bot$ are at distance larger than $2 c \log_{\Delta} n$, where $c$ is a constant from \cref{lem:simple_dist_brooks}. The goal is to compute a proper $\Delta$-coloring of $G$.
\end{problem}

\begin{lemma}
	Let $\mathcal{G}_{\Delta}$ be the set of all $\Delta$-colorable graphs with vertices partially colored according to some proper $\Delta$-coloring, in which any two uncolored vertices are at a distance larger than $2 c \log_{\Delta} n$. Let $\Pi_\Delta$ be the problem described in \cref{problem:reduce_to_roots}. Then, there exists a $(\mathcal{G}_\Delta,\Pi_\Delta,\gamma_0,A,T)$-composable advice schema, where $\gamma_0 = O(1)$, $A(c,\gamma) = \Theta(\gamma^3)$, and $T(\alpha,\Delta) = O(\alpha)$.	

\end{lemma}
\begin{proof}
	We fix $\beta =3, \gamma_0 = 2$, and we show that, for any $\gamma \ge \gamma_0$, for any constant $c$, and for any $\alpha = \max\{\gamma^3 \beta / c, \gamma^3 \beta\}$, there exists a variable-length $(\mathcal{G}_\Delta,\Pi_\Delta,\beta,O(\alpha))$-advice schema satisfying that in each $\alpha$-radius neighborhood there are at most $\gamma_0$ bit-holding nodes.

	By \cref{lem:simple_dist_brooks} we know that there exists a set of vertices $X$ such that each uncolored vertex in $G$ has a vertex from $X$ such that $G$ can be properly colored by shifting the colors along the shortest paths from uncolored vertices to vertices from $X$ and recoloring vertices in $X$ with a free color. Our composable advice schema puts:
\begin{itemize}
\item markers $111$ on each of the vertices from $X$,
\item markers $11$  on every $2\alpha+10$th vertex on the shortest paths from the uncolored vertices to vertices from $X$ (leaving one block of unmarked vertices of length slightly longer, but at most two times longer, to accommodate for the fact that the length of the path may be not divisible by $2\alpha+10$), 
\item marker $1$, on a neighbor of each vertex marked with $11$, in the direction of $X$ along the shortest path.
\end{itemize}
Such a schema deposits at most $3$ bits per vertex and has at most $2$ bit holding vertices in any ball of radius $\alpha$.

\emph{Decoding algorithm:} to compute the paths, each vertex with no color or with marker $11$ gathers its $2\alpha+22$ hop neighborhood. For each vertex, such information is sufficient to identify $2\alpha+22$ predecessors and successors on the shortest path between uncolored vertices and $X$, and a direction towards $X$. This information can be then broadcasted by each vertex marked by $11$ to all vertices on the path. Then, the path vertices can update their color accordingly. 

As for vertices with marker $111$, they wait for other vertices to recolor and then recolor themselves. The property of vertices was that either they have a low degree (and then they always can recolor themselves with a color no larger than $\Delta$), or that they have two neighbors of the same color, none being on the selected shortest paths from $I$ to $X$. As such, after recoloring the vertices from $X$ still have at least one free color to choose from (as they still have degrees smaller than $\Delta$ or two neighbors with the same color).
\end{proof}

\section{3-coloring 3-colorable graphs}\label{ssec:3-col-3-col}
In this section, we devise an algorithm that colors any $3$-colorable graph with $3$ colors in $\poly(\Delta)$ communication rounds, given $1$ bit of advice per node. In more detail, we prove the following theorem.
\begin{theorem}\label{thm:3coloring}
	Let $\mathcal{G}$ be the family of $3$-colorable graphs. Let $\Pi$ be the problem of computing a $3$-coloring.
	Then, there exists a uniform fixed-length $(\mathcal{G}, \Pi,1,\Delta^{O(1)})$-advice schema.
\end{theorem}
Throughout this section, we assume that the input graph $G$ is $3$-colorable.
We start by giving an informal overview of the encoding schema and the $3$-coloring algorithm.

On a high level, the idea of the encoding schema is to fix a greedy $3$-coloring (with colors $1, 2, 3$), mark each node of color $1$ with bit $1$, and then add bit $1$ in a few other places to indicate which $2$-coloring (with colors $2$ and $3$) is chosen in each large component obtained by removing nodes of color $1$.
(For small components, i.e., those of diameter upper bounded by some large enough function from $\poly(\Delta)$, we do not need to give advice, as the nodes can just consistently choose one of the two $2$-colorings by gathering the whole component.)
One challenge here is to make sure that $1$-bits that indicate that the respective node is of color $1$ (which we call $1$-bits \emph{of type $1$}) are distinguishable from $1$-bits that indicate the $2$-coloring of a large component (which we call $1$-bits \emph{of type $23$}).
Another difficulty is how to encode the chosen $2$-colorings without compromising the schema for distinguishing between the two kinds of $1$-bits.
Also, it should be noted that inside each large component, we need to distribute sufficiently many $1$-bits so that for each node inside the component, its distance to one of those $1$-bits is upper bounded by some function from $\poly(\Delta)$.
In fact, we need something slightly stronger, namely, that each node inside a large component is in at most $\poly(\Delta)$ distance (inside the component!) to sufficiently many of those $1$-bits to infer the $2$-coloring chosen by the encoding schema for the respective component.

Our solution for distinguishing between $1$-bits of type $1$ and type $23$ is simple: a $1$-bit at some node $v$ is of type $1$ if and only if $v$ has at most one neighbor that is assigned bit $1$.
To indicate the $2$-coloring of a large component $C$, we will assign $1$-bits to nodes inside the component such that the following properties are satisfied for the subset $S \subseteq V(C)$ of nodes of $C$ assigned a $1$-bit:
\begin{enumerate}
	\item Each node of $C$ is in distance at most $\poly(\Delta)$ (inside the component) to some node of $S$.
	\item The nodes of $S$ can be partitioned into groups such that inside each group all nodes are ``close together'', whereas any two nodes of $S$ from different groups are ``far apart'' (making the groups distinguishable).
	\item In each group, the nodes of the group form either one or two connected components, and if the number of connected components is $1$, the smallest-ID node in the component has color $2$, whereas if the number of connected components is $2$, the smallest-ID node among all nodes in both components has color $3$.   
\end{enumerate}

With this $1$-bit encoding of (parts of) the fixed greedy coloring, the design of a $3$-coloring algorithm becomes fairly natural.
First, each node $u$ collects a sufficiently large neighborhood to determine whether it has color $1$ in the greedy coloring (in which case $u$ outputs color $1$), and, if not, whether it is in a ``small'' or a ``large'' component of nodes of colors $2$ and $3$.
If $u$ is in a small component, it (and all other nodes in the component) simply outputs the color it would receive in a simple fixed deterministic coloring schema of the component.
If $u$ is in a large component, it determines a close-by group of $1$-bits, infers from the number of connected components in that group the color of the smallest-ID node of that group in the greedy $3$-coloring, and then outputs the color that it would receive in the unique $2$-coloring of its component that respects the color of the aforementioned smallest-ID node.

By design, the described $3$-coloring algorithm recovers the fixed greedy $3$-coloring for all nodes of color $1$ and all large connected components, while it produces a correct $2$-coloring with colors $2, 3$ (that may or may not coincide with the fixed greedy $3$-coloring) on the small connected components, yielding a proper $3$-coloring of $G$.

Before describing the encoding schema and the $3$-coloring algorithm more formally, we prove a technical lemma that will be useful for the design of the encoding schema.
The lemma essentially states that in $O(\Delta)$ distance of each node in a component of nodes of colors $2, 3$ in the fixed greedy $3$-coloring, there are $1$ or $2$ nodes in the component with the property that if the node(s) is/are assigned bit $1$, then the assigned $1$-bits are of type $23$, and for no node of color $1$ the number of $1$-bits in its $1$-hop neighborhood increases by more than $1$.
\begin{lemma}\label{lem:single-or-double}
	Let $G$ be a $3$-colorable graph and assume that a proper greedy $3$-coloring (with color $1, 2, 3$) of the nodes of $G$ is given.
	Let $C$ be a maximal connected component of the subgraph induced by the nodes of colors $2$ and $3$.
	Assume that the diameter of $C$ is at least $2\Delta$ and consider any node $v \in V(C)$.
	Then at least one of the following holds:
	\begin{enumerate}
		\item There exists a node $w \in V(C)$ satisfying $\dist_C(v, w) \leq \Delta$ such that $w$ has at least two neighbors of color $1$ in $G$.
		\item There exist two neighbors $x, y \in V(C)$ satisfying $\dist_C(v, x) \leq \Delta$ and $\dist_C(v, y) \leq \Delta$ such that $x$ and $y$ do not have a common neighbor of color $1$ in $G$. 
	\end{enumerate}
\end{lemma}
\begin{proof}
	Assume there is no node $w$ with the stated property (otherwise we are done).
	This in particular implies that any node in distance at most $\Delta$ from $v$ in $C$ has exactly one neighbor of color $1$ in $G$ (as the given $3$-coloring is a greedy $3$-coloring).
	For $v$, let $u$ denote this unique neighbor.
	Consider any path $v = v_0, v_1, \dots, v_{\Delta}$ of length $\Delta$ starting in $v$.
	(Such a path must exist as otherwise the diameter of $C$ would be strictly smaller than $2\Delta$.)
	As $u$ has degree at most $\Delta$ and $v$ is a neighbor of $u$, at least one node from $\{ v_1, \dots, v_{\Delta} \}$ is not a neighbor of $u$.
	This implies that there is some index $0 \leq i \leq \Delta - 1$ such that exactly one of $v_i$ and $v_{i + 1}$ is a neighbor of $u$.
	As both $v_i$ and $v_{i + 1}$ have exactly one neighbor of color $1$ in $G$, it follows that $v_i$ and $v_{i+1}$ constitute nodes $x$ and $y$ with the properties described in the lemma.
\end{proof}

Now we are set to formally describe our encoding schema.
While it might not be immediately obvious why some objects that the schema computes along the way exist, we will show subsequently in \Cref{lem:3colwelldef} that they indeed exist.

\subparagraph{The encoding schema.}
Fix an arbitrary greedy $3$-coloring $\varphi \colon V(G) \rightarrow \{ 1, 2, 3 \}$ of the nodes of $G$ (i.e., each node of color $i$ has neighbors of all colors $< i$).
Assign bit $1$ to all nodes of color $1$.

Let $G_{2,3}$ denote the graph induced by the nodes of colors $2$ and $3$.
Let $G^*_{2,3}$ denote the subgraph of $G_{2,3}$ obtained by removing all maximal connected components of diameter\footnote{Whenever we refer to the diameter of a component $C$, we refer to its \emph{strong diameter}, i.e., to the diameter of the component in $C$, not in $G$.} at most $4000 \Delta^9$ from $G_{2,3}$.

For each maximal connected component $C$ of $G^*_{2,3}$, do the following.

Compute a $(2000\Delta^9, 2000\Delta^9)$-ruling set $R_C$ (where distances are with respect to $C$, not $G$).
For each node $r \in R_C$, do the following.

Choose a set $Q_r$ with the following properties.
\begin{enumerate}
	\item $|Q_r| = 12 \Delta^6$.
	\item For any node $v \in Q_r$, we have $\dist_C(r, v) \leq 600 \Delta^9$.
	\item For any two distinct nodes $v_1, v_2 \in Q_r$, we have $\dist_C(v_1, v_2) \geq 50 \Delta^3$.
\end{enumerate}
For each node $v \in Q_r$, do the following.

Select either a node $w$ or two nodes $x, y$ as described in \Cref{lem:single-or-double}, applied to $v$.
Denote the set of selected nodes by $S_v$ (i.e., either $S_v = \{ w \}$ or $S_v = \{ x, y \}$).
Denote by $T_v$ the set of all nodes of $V(C)$ that
\begin{enumerate}
	\item do not share a neighbor of color $1$ with some node from $S_v$, and
	\item are neither neighbors of nor identical to nodes in $S_v$.
\end{enumerate}
Choose a node $v'$ and a path $P_v$ of length $\Delta$ starting in $v'$ such that
\begin{enumerate}
	\item all nodes of $P_v$ are contained in $T_v$, and
	\item $\dist_C(v, v') \leq 20 \Delta^3$.
\end{enumerate}

Select either a node $w'$ or two nodes $x', y'$ as described in \Cref{lem:single-or-double},
applied to $v'$, with the additional constraint that the selected node or a pair of nodes must be contained in $P_v$.
Denote the set of selected nodes by $S'_v$ (i.e., either $S'_v = \{ w' \}$ or $S'_v = \{ x', y' \}$).

Let $\mathcal C(G^*_{2,3})$ denote the set of all maximal connected components $C$ of $G^*_{2,3}$.
For each $C \in \mathcal C(G^*_{2,3})$, and each $r \in R_C$, select one node $v_{r, C} \in Q_r$ such that each node of color $1$ in $G$ has at most one neighbor in the set
\[
\bigcup_{C  \in \mathcal C(G^*_{2,3}), r \in R_C} \left( S_{v_{r,C}} \cup S'_{v_{r,C}} \right) .
\]

Now, for each component $C \in \mathcal C(G^*_{2,3})$ and each ruling set node $r \in R_C$, assign $1$-bits to a subset of $S_{v_{r, C}} \cup S'_{v_{r, C}}$ as follows.
Let $s$ denote the smallest-ID node in $S_{v_{r, C}} \cup S'_{v_{r, C}}$, and let $X_s \in \{ S_{v_{r, C}}, S'_{v_{r, C}} \}$ denote the one of the two sets containing $s$.
If $s$ is of color $2$, then assign bit $1$ to all nodes in $X_s$.
If $s$ is of color $3$, then assign bit $1$ to all nodes in $S_{v_{r, C}} \cup S'_{v_{r, C}}$.

Finally, assign bit $0$ to all nodes that have not been assigned bit $1$ so far.

\subparagraph{Well-definedness of the encoding schema.}
The following lemma shows that the encoding schema is well-defined, i.e., that all objects that the encoding schema computes along the way do indeed exist.

\begin{lemma}\label{lem:3colwelldef}
	The encoding schema is well-defined.
\end{lemma}
\begin{proof}
	We show the existence of the objects the encoding schema computes along the way in the order in which they appear in the schema (omitting discussions of trivially existing objects).
	The existence of a $(2000\Delta^9, 2000\Delta^9)$-ruling set follows from the fact that any maximal independent set on the power graph $C^{2000\Delta^9}$ is a $(2000\Delta^9, 2000\Delta^9)$-ruling set on $C$ (and the fact that maximal independent sets trivially exist on any graph as they can be found by a simple greedy algorithm).
	
	Next, consider the set $Q_r$ with the mentioned properties (for some fixed $r$ in some component $C$).
	Since $C$ is a component of $G^*_{2,3}$, we know that $C$ has diameter larger than $4000\Delta^9$, which implies that there is a node $r_{600\Delta^9}$ that is at distance exactly $600\Delta^9$ from $r$ in $C$.
	Let $(r = r_0, r_1, r_2, \dots, r_{600\Delta^9})$ be a shortest path from $r$ to $r_{600\Delta^9}$.
	Then the set $Q_r := \{ r_{i \cdot 50\Delta^3} \mid 1 \leq i \leq 12\Delta^6 \}$ satisfies the properties specified in the encoding schema.
	
	Next, the existence of a node $w$ or two nodes $x, y$ as described follows from \Cref{lem:single-or-double}.
	
	Now, consider node $v'$ and path $P_v$.
	We begin the proof of their existence by bounding the number of nodes in $V(C) \setminus T_v$.
	Observe that there are at most $2\Delta^2$ nodes that share a neighbor of color $1$ with some node from $S_v$ since $|S_v| \leq 2$.
	Observe further that $|S_v| \leq 2$ also implies that there are at most $2\Delta + 2$ nodes that are neighbors of or identical to some node in $S_v$.
	Hence, $|V(C) \setminus T_v| \leq 2\Delta^2 + 2\Delta + 2 \leq 6\Delta^2$.
	Now consider a node $v_{20\Delta^3}$ that is at distance exactly $20\Delta^3$ from $v$ in $C$ (whose existence again follows from the fact that $C$ has diameter larger than $4000\Delta^9$).
	Let $(v = v_0, v_1, v_2, \dots, v_{20\Delta^3})$ be a shortest path from $v$ to $v_{20\Delta^3}$.
	Since $|V(C) \setminus T_v| \leq 6\Delta^2$, there must be a subpath $(v_i, v_{i+1}, \dots, v_{i + \Delta})$ of length $\Delta$ such that all nodes on the subpath are contained in $T_v$.
	Therefore, by setting $v' := v_i$ and $P_v := (v_i, v_{i+1}, \dots, v_{i + \Delta})$, we ensure that $v'$ and $P_v$ satisfy the properties stated in the encoding schema, concluding the existence proof for $v'$ and $P_v$.
	
	To show the existence of a node $w'$ or two nodes $x', y'$ as described in the schema, it suffices to observe that the proof of \Cref{lem:single-or-double} also works with the additional restriction that the selected node $w'$ or the selected two nodes $x', y'$ are contained in $P_v$ (as the proof, in fact, starts by finding such a path).
	
	Finally, the last and most interesting part of proving the well-definedness of the encoding schema consists in showing that the desired nodes $v_{r,C}$ satisfying the property specified in the schema exist.
	We show this existence by using the Lovász Local Lemma (LLL).
	Define an LLL instance as follows.
	
	For each $C \in \mathcal C(G^*_{2,3})$, and each $r \in R_C$, let $X_{r, C}$ denote the uniformly distributed random variable with value set $Q_r$.
	For each node $u$ of color $1$, let $A_u$ denote the event that there are two different pairs $(C, r) \neq (C', r')$ with $C, C' \in \mathcal C(G^*_{2,3})$ (where possibly $C = C'$), $r \in R_C$, and $r' \in R_{C'}$ such that $u$ has both a neighbor in $S_{X_{r, C}} \cup S'_{X_{r, C}}$ and a neighbor in $S_{X_{r', C'}} \cup S'_{X_{r', C'}}$.
	
	Next, we bound the number of random variables an event $A_u$ depends on.
	Observe that for each set $Q_r$ (where $r \in R_C$ for some $C \in \mathcal C(G^*_{2,3})$), each node $v \in Q_r$, and each node $w \in S_v \cup S'_v$, we have $\dist_C(r,w) \leq \dist_C(r,v) + \dist_C(v,w) \leq 600 \Delta^9 + 20 \Delta^3 + \Delta \leq 700 \Delta^9$ (by \Cref{lem:single-or-double} and the distance properties stated in the description of the encoding schema).
	Hence, for any two different pairs $(C, r) \neq (C', r')$ (as specified above), and any two nodes $u \in Q_r$, $v \in Q_{r'}$, we have $(S_u \cup S'_u) \cap (S_v \cup S'_v) = \emptyset$: either we have $C \neq C'$, in which case $S_u \cup S'_u \in V(C)$ and $S_v \cup S'_v \in V(C')$ are subsets of the node sets of different components, or $C = C'$ and $r \neq r'$, in which case the computed distance upper bound of $700 \Delta^9$, together with the fact that the nodes in $R_C$ have a pairwise distance of at least $2000 \Delta^9$ in $C$, ensures that any node in $(S_u \cup S'_u)$ has a distance (in $C$) of at least $600 \Delta^9$ from any node in $(S_v \cup S'_v)$.
	By the definition of the $A_u$, it follows that each event $A_u$ depends on at most $\Delta$ random variables (as $u$ has at most $\Delta$ neighbors).
	
	Vice versa, we claim that for any random variable $X_{r, C}$, there are at most $(12\Delta^6 \cdot 4 \cdot \Delta) = 48\Delta^7$ different events $A_u$ that depend on $X_{r, C}$: this simply follows from the facts that $12\Delta^6$ is an upper bound for the number of sets of the form $S_v \cup S'_v$ for some $v \in Q_r$, $4$ is an upper bound for the number of nodes in any such $S_v \cup S'_v$, and $\Delta$ is an upper bound for the number of events each such node can affect.
	We conclude that the dependency degree $d$ of our LLL instance is upper bounded by $d \leq \Delta \cdot 48\Delta^7 = 48 \Delta^8$.
	
	Now we bound the probability that an event $A_u$ occurs.
	Observe that for any two distinct nodes $u \neq v$ contained in the same $Q_r$, and any two nodes $\hat{u} \in S_u \cup S'_u$ and $\hat{v} \in S_v \cup S'_v$, we have $\dist_C(\hat{u}, \hat{v}) > 0$, due to $\dist_C(u, v) \geq 50\Delta^3$, $\dist_C(u, \hat{u})\leq 20\Delta^3 + \Delta \leq 21 \Delta^3$, and $\dist_C(v, \hat{v})\leq 20\Delta^3 + \Delta \leq 21 \Delta^3$ (where $C$ denotes the component containing $r$).
	Hence, $(S_u \cup S'_u) \cap (S_v \cup S'_v) = \emptyset$ for any two such $u, v$.
	It follows that for each event $A_u$, and each random variable $X_{r,C}$ that $A_u$ depends on, the probability that $u$ has a neighbor in $S_{X_{r, C}} \cup S'_{X_{r, C}}$ is at most $\Delta/|Q_r| = 1/(12\Delta^5)$.
	This implies that for any two distinct random variables $X_{r,C}, X_{r',C'}$ the probability that $u$ has a neighbor in both $S_{X_{r, C}} \cup S'_{X_{r, C}}$ and $S_{X_{r', C'}} \cup S'_{X_{r', C'}}$ is upper bounded by $1/(144\Delta^{10})$.
	Recall that $A_u$ depends on at most $\Delta$ random variables; hence, by union bounding over all the (less than $\Delta^2$) pairs of random variables that $A_u$ depends on, we obtain an upper bound of $\Delta^2/(144\Delta^{10}) = 1/(144\Delta^8)$ for the probability $p$ that $A_u$ occurs.
	
	Now we are ready to apply the Lov\'asz Local Lemma.
	By the bounds established above, we know that $epd \leq e \cdot 1/(144\Delta^8) \cdot 48\Delta^8 < 1$.
	By the Lov\'asz Local Lemma, this implies that there exists an assignment to the random variables $X_{r,C}$ such that none of the events $A_u$ occurs.
	In other words, it is possible to select, for each $C \in \mathcal C(G^*_{2,3})$, and each $r \in R_C$, a node $v_{r,C}$ such that for each node $u$ of color $1$ there exists at most one pair $(r,C)$ such that $u$ has a neighbor in $( S_{v_{r,C}} \cup S'_{v_{r,C}})$.
	Moreover, for any fixed pair $(r,C)$ and any node $u$ of color $1$, at most one neighbor of $u$ is contained in $( S_{v_{r,C}} \cup S'_{v_{r,C}})$, by the guarantees provided by \Cref{lem:single-or-double} and the construction of $S'_v$ (which ensures that $S'_v \subseteq T_v$).
	Hence, the aforementioned selection of nodes $v_{r,C}$ indeed satisfies the property specified in the encoding schema, concluding the proof.
\end{proof}

Before formally describing the algorithm that uses the computed advice for computing a $3$-coloring, we prove a useful property related to the encoding schema.

\begin{lemma}\label{lem:twoones}
	Let $\varphi \colon V(G) \rightarrow \{ 1, 2, 3\}$ be the greedy coloring fixed at the beginning of the encoding schema, and let $b(v)$ denote the bit that the encoding schema assigns to a node $v$.
	Then $\varphi(v) = 1$ if and only if the following two properties hold: $b(v) = 1$ and there is at most one neighbor $w$ of $v$ with $b(w) = 1$.
\end{lemma}
\begin{proof}
	Consider a node $v$ with $\varphi(v) = 1$.
	Then, the encoding schema guarantees that $b(v) = 1$.
	Moreover, observe that any node $u$ with $\varphi(u) \neq 1$ and $b(u) = 1$ is contained in
	\[
	\bigcup_{C  \in \mathcal C(G^*_{2,3}), r \in R_C} \left( S_{v_{r,C}} \cup S'_{v_{r,C}} \right) ,
	\]
	and that each node of color $1$ has at most one neighbor in this set.
	Hence, $v$ has at most one neighbor $w$ with $b(w) = 1$.
	
	Now, for the other direction, consider a node $v$ with $\varphi(v) \neq 1$.
	If $b(v) \neq 1$, we are done; thus assume $b(v) = 1$.
	As above, we obtain that $v$ is contained in
	\[
	\bigcup_{C  \in \mathcal C(G^*_{2,3}), r \in R_C} \left( S_{v_{r,C}} \cup S'_{v_{r,C}} \right) .
	\]
	By the definitions of $S_{v_{r,C}}$ and $S'_{v_{r,C}}$ (and the guarantees specified in \Cref{lem:single-or-double}) and the choice of exactly which nodes of color $\neq 1$ receive a $1$-bit, this implies that $v$ has at least two neighbors of color $1$ or $v$ has a neighbor $u$ of color $\neq 1$ that also receives bit $1$.
	As we are done in the former case, assume the latter, i.e., that $v$ has a neighbor $u$ with $\varphi(u) \neq 1$ and $b(u) = 1$.
	Since $\varphi$ is a greedy coloring and $\varphi(u) \neq 1$, we know that $v$ has a neighbor $w$ of color $1$ (which therefore satisfies $b(w) = 1$).
	As $\varphi(u) \neq 1$, we have $w \neq u$, which implies that $v$ has two neighbors that receive bit~$1$.
\end{proof}

\subparagraph{The 3-coloring algorithm.}
The algorithm $\mathcal A$ for computing a $3$-coloring using the advice provided by the encoding schema proceeds as follows.
Each node that receives bit $1$ and has at most one neighbor receiving bit $1$ outputs color $1$.
Let us call the set of these nodes $W$.
Each node $u \in V(G) \setminus W$ collects its $(4000\Delta^9 + 1)$-hop neighborhood in $V(G) \setminus W$.
If the maximal connected component $C(u)$ in $V(G) \setminus W$ containing $u$ is of diameter at most $4000\Delta^9$, then $u$ outputs the color it would receive in the unique proper $2$-coloring of the nodes of $C(u)$ that assigns color $2$ to the node in $C(u)$ of smallest ID.
Otherwise, $u$ selects a node $w$ in $C(u)$ that receives bit $1$ and satisfies $\dist_{C(u)}(u,w) \leq 3000 \Delta^9$.
(Such a node $w$ exists as the encoding schema guarantees that $u$ is in distance at most $2000 \Delta^9$ from some node $r \in R_{C(u)}$, which in turn is in distance at most $600\Delta^9$ from a node $v_{r,C(u)}$ that has distance at most $20 \Delta^3 + \Delta$ to some node in $S_{v_{r,C(u)}} \cup S'_{v_{r,C(u)}}$ that receives bit $1$, and all these distances are in $C$.)
Then, $u$ determines the number $N$ of maximal connected components of nodes that receive bit $1$ in the $(30 \Delta^3)$-hop neighborhood of $w$ in $C(u)$.
Moreover, $u$ determines the node $x$ of smallest ID among all nodes in the $(30 \Delta^3)$-hop neighborhood of $w$ in $C(u)$ that receive color $1$, and computes its distance to $x$ in $C(u)$.
If $N = 1$ and $\dist_{C(u)}(u,x)$ is even, then $u$ outputs color $2$.
If $N = 1$ and $\dist_{C(u)}(u,x)$ is odd, then $u$ outputs color $3$.
If $N > 1$ and $\dist_{C(u)}(u,x)$ is odd, then $u$ outputs color $2$.
If $N > 1$ and $\dist_{C(u)}(u,x)$ is even, then $u$ outputs color $3$.
This concludes the description of $\mathcal A$.
In the following theorem, we show that $\mathcal A$ produces a correct $3$-coloring, and bound its runtime.

\begin{theorem}
	Algorithm $\mathcal A$ outputs a proper $3$-coloring in $O(\Delta^9)$ rounds.	
\end{theorem}
\begin{proof}
	After collecting its $(4000\Delta^9 + 2)$-hop neighborhood, each node $u$ of $G$ can determine whether it outputs color $1$ and, if it does not output color $1$, which nodes of $G$ are contained in its $(4000\Delta^9 + 1)$-hop neighborhood in $C(u)$.
	As this information suffices to compute $u$'s output, $\mathcal A$ terminates in $O(\Delta^9)$ rounds.
	
	The remainder of the proof is dedicated to proving the correctness of $\mathcal A$.
	More precisely, we will show that $\mathcal A$ outputs the greedy coloring $\varphi$ fixed at the beginning of the encoding schema, except on maximal connected components of $G_{2,3}$ of diameter at most $4000 \Delta^9$, on which instead it will output the unique $2$-coloring with colors $2, 3$ that assigns color $2$ to the smallest-ID node of the component.
	We denote this modified coloring by $\psi$.
	As $\psi$ is a proper $3$-coloring, it suffices to show that $\mathcal A$ outputs $\psi$.
	
	Consider first nodes of color $1$ in $\psi$.
	These nodes have color $1$ also in $\varphi$, and by \Cref{lem:twoones} and the definition of $\mathcal A$, it follows that these nodes output color $1$, as desired.
	
	Next, let $u$ be a node with $\psi(u) \neq 1$ for which $C(u)$ has diameter at most $4000\Delta^9$.
	By the definitions of $\mathcal A$ and $\psi$, node $u$ outputs precisely $\psi(u)$.
	
	Finally, let $u$ be a node with $\psi(u) \neq 1$ for which the diameter of $C(u)$ is strictly larger than $4000\Delta^9$.
	Observe that each node in $C(u)$ receives the same color under $\psi$ as under $\varphi$.
	Consider the node $w$ selected by $u$ in $\mathcal A$.
	By the definition of the encoding schema, there exist nodes $r \in R_{C(u)}$ and $v_{r,C(u)} \in Q_r$ such that $w \in ( S_{v_{r,C}} \cup S'_{v_{r,C}} )$.
	As essentially already observed in the proof of \Cref{lem:3colwelldef}, we know that for any $r' \in R_{C(u)}$ satisfying $r' \neq r$, and any node $y \in ( S_{v_{r',C}} \cup S'_{v_{r',C}} )$, we have $\dist_{C(u)}(w,y) \geq 600\Delta^9$.
	By the specification of which nodes receive bit $1$ in the encoding schema, it follows that any node in the $(30 \Delta^3)$-hop neighborhood of $w$ in $C(u)$ that receives bit $1$ is contained in $S_{v_{r,C}} \cup S'_{v_{r,C}}$.
	Vice versa, since any two nodes in $S_{v_{r,C}} \cup S'_{v_{r,C}}$ have distance at most $\Delta + 20 \Delta ^3 + \Delta$ in $C(u)$ and $w \in ( S_{v_{r,C}} \cup S'_{v_{r,C}} )$, all nodes from $S_{v_{r,C}} \cup S'_{v_{r,C}}$ are contained in the $(30 \Delta^3)$-hop neighborhood of $w$ in $C(u)$.
	Hence, it suffices to restrict attention to which nodes from $S_{v_{r,C}} \cup S'_{v_{r,C}}$ receive bit $1$ in the encoding schema in order to determine which color $u$ outputs according to $\mathcal A$.
	
	From the design of the encoding schema (in particular from the fact that $S'_{v_{r,C}} \subseteq T_{v_{r,C}}$), we know that no node in $S_{v_{r,C}}$ is a neighbor of a node in $S'_{v_{r,C}}$.
	Moreover, the graph induced by the nodes in $S_{v_{r,C}}$ is connected, and the graph induced by the nodes in $S'_{v_{r,C}}$ is connected as well.
	Observe further that either all nodes of $S_{v_{r,C}}$ receive bit $1$ or none of them, and the same holds for $S'_{v_{r,C}}$.
	More specifically, the design of the encoding schema ensures that if the smallest-ID node in $S_{v_{r,C}} \cup S'_{v_{r,C}}$ is of color $2$, then there is exactly one maximal connected component of nodes that receive bit $1$ in the $(30 \Delta^3)$-hop neighborhood of $w$ in $C(u)$ (i.e., $N = 1$), and the smallest-ID node among all nodes in the $(30 \Delta^3)$-hop neighborhood of $w$ in $C(u)$ is identical to the smallest-ID node in $S_{v_{r,C}} \cup S'_{v_{r,C}}$ (which implies $\psi(x) = \varphi(x) = 2$).
	Similarly, the encoding schema ensures that if the smallest-ID node in $S_{v_{r,C}} \cup S'_{v_{r,C}}$ is of color $3$, then $N = 2$ and $\psi(x) = \varphi(x) = 3$.
	Now observe that in $\psi = \varphi$ (on $C(u)$), the colors of $u$ and $x$ are identical if $\dist_{C(u)}(u,x)$ is even, and the colors of $u$ and $x$ are different if $\dist_{C(u)}(u,x)$ is odd (simply by the fact that $\psi = \varphi$ is a proper $2$-coloring on $C(u)$).
	Going through the four different cases regarding the size of $N$ and the parity of $\dist_{C(u)}(u,x)$, we conclude that the design of $\mathcal A$ ensures that $u$ outputs $\psi(u)$.
	
	Hence, $\mathcal A$ outputs precisely $\psi$ on the nodes of $G$, which implies that $\mathcal A$ outputs a proper $3$-coloring. 
\end{proof}

\section{Composability}\label{app:composability}
We now prove that, given many composable schemas, it is possible to combine them into a single composable schema.
\begin{lemma}\label{lem:compose}
	Let $\mathcal{P} = \{\Pi_1,\ldots,\Pi_k\}$ be a set of $k$ problems satisfying that, for each problem $\Pi_i$, there exists a $(\mathcal{G},\Pi_i,\gamma_i,A_i, T_i)$-composable advice schema $\mathcal{S}_i$. Each problem $\Pi_i$ is allowed to be defined with the promise that the input contains a solution for a subset of problems in $\{\Pi_1,\ldots,\Pi_{i-1}\}$.

	Let $\Pi$ be the problem that requires to output a solution for each problem $\Pi_1,\ldots,\Pi_k$ when no solution for any $\Pi_i$ problem is given.
	Then, there exists a $(\mathcal{G},\Pi,\gamma_0,A, T)$-composable advice schema $\mathcal{S}$, where:
	\begin{itemize}
		\item $\gamma_0 = \sum_{i=1}^{k} \gamma_i$;
		\item $A(c,\gamma) = \max\{\frac{(2k\gamma)^3 \lceil \log k \rceil}{c}, A_1(c,2 k \gamma),\ldots, A_k(c,2 k \gamma)\}$;
		\item $T$ is defined as follows. Let $D$ be a directed acyclic graph of $k$ nodes, where each node corresponds to a problem in $\mathcal{P}$, and there is a directed edge from $\Pi_i$ to $\Pi_j$ if $\Pi_i$ is defined such that it receives as input a solution for $\Pi_j$. For a directed path $P = (\Pi_{i_1},\ldots,\Pi_{i_d})$ in $D$, let $\mathrm{cost}(\alpha,\Delta,P) := \sum_{1 \le j \le d} T_{i_j}(\alpha,\Delta)$. Then,  $T(\alpha,\Delta) := \max_{P \in D} \mathrm{cost}(\alpha,\Delta,P)$.
	\end{itemize} 
\end{lemma}
\begin{proof}
	We need to prove that, for any constant $c > 0$ and any $\gamma \ge \gamma_0$, for any $\alpha \ge A(c,\gamma)$, there exists $\beta \le c \alpha / \gamma^3$ such that:
	\begin{itemize}
		\item There exists a variable-length $(\mathcal{G},\Pi,\beta,T(\alpha,\Delta))$-advice schema $S$.
		\item For each $G \in \mathcal{G}$, the assignment given by $S$ to the nodes of $G$ satisfies that, in each $\alpha$-radius neighborhood of $G$, there are at most $\gamma_0$ bit-holding nodes.
	\end{itemize}
	Hence, in the following, let $c>0$ be an arbitrary constant, let $\gamma$ be an arbitrary integer that is at least $\gamma_0$, and let $\alpha$ be an arbitrary integer that is at least $A(c,\gamma)$.
	
	For a composable schema $\mathcal{S}$, let $\mathcal{S}(c,\alpha)$ be the schema that is part of the collection $\mathcal{S}$ obtained by using parameters $c$ and $\alpha$. Let $G = (V,E) \in \mathcal{G}$, and let $f_i := \mathcal{S}_i(c,\alpha)$.
	We now define, recursively, $k$ functions $\ell_i$ mapping nodes into bit-strings. Let $\ell_1 := f_1(G)$, and let $O_1$ be the solution that the nodes would compute for $\Pi_1$ when given $\ell_1$ as advice. Let $\mathrm{succ}_D(\Pi_{i})$ be the set of nodes reachable by $\Pi_{i}$ in $D$, that is, the set of problems for which $\Pi_{i}$ is assumed to receive a solution as input. Observe that $\mathrm{succ}_D(\Pi_{i}) \subseteq \{\Pi_1,\ldots,\Pi_{i-1}\}$.
	Let $G_i$ be defined as the graph $G$ where nodes are also labeled with the solutions $\{O_j \mid \Pi_j \in \mathrm{succ}_D(\Pi_{i})\}$. We define  
	$\ell_{i} := f_{i}(G_i)$, and $O_i$ as the solution that nodes would compute for $\Pi_i$ when given the advice $\ell_{i}$.
	
	For each node $v$, let $B(v) = \{ (i,\ell_i(v)) \mid 1 \le i \le k \text{ and } |\ell_i(v)| > 0 \}$.
	Observe that, by the definition of $\gamma_0$ and by construction of $B(v)$, in each $\alpha$-radius neighborhood, there are at most $\gamma_0$ nodes $v$ with non-empty sets $B(v)$. Moreover, by the definition of $\gamma_0$ and $A$, each string $\ell_i(v)$ satisfies $|\ell_i(v)| \le c \alpha / (2k \gamma)^3$. 
	 
	 Each set $B(v)$ can be encoded by using $c \alpha / \gamma^3$ bits, as follows. First, each pair $(i,\ell_i(v))$ is encoded by using at most $\lceil \log k \rceil + 2\lceil \log(c \alpha / (2k \gamma)^3)\rceil + 1 + c \alpha / (2k \gamma)^3$ bits by concatenating the bit-string representation of $i$, followed by $\lceil \log |\ell_i(v)| \rceil$ $1$s, followed by a $0$, followed by the bit-string representation of $|\ell_i(v)|$, followed by $\ell_i(v)$. Observe that such a string can be decoded to recover the pair. Then, the bit-strings of the pairs are concatenated into a single bit-string $\ell_G(v)$, that has length at most $k \lceil \log k\rceil + 4 k c\alpha /(2k \gamma)^3 \le 5 k c\alpha /(2k \gamma)^3 \le c \alpha / \gamma^3$, where the first inequality holds by the definition of $A$. The schema $S$ is defined as the one using, $\ell_G(v)$, for each $G \in \mathcal{G}$.
	 
	 Nodes can solve $\Pi$ as follows. First, each node $v$ decodes its own advice, obtaining a string $\ell_i(v)$ for each $1 \le i \le k$. Then, nodes consider the graph $D$: problems that are sinks in $D$ can be solved in parallel, and after that, nodes can recurse on the subgraph of $D$ induced by problems that still have to be solved. In total, they spend $T(\alpha,\Delta)$ time.
\end{proof}

\begin{lemma}\label{lem:composable-to-1bit}
	Let $\mathcal{S}$ be a $(\mathcal{G},\Pi,\gamma,A, T)$-composable advice schema. Then, there exists a uniform fixed-length sparse $(\mathcal{G},\Pi,1,T')$-advice schema $S'$, where $T':= O(\alpha + T(\alpha,\Delta))$ and $\alpha := A(c,\gamma)$ for a small-enough constant $c$.
\end{lemma}
\begin{proof}
	Let $c$ be a small-enough constant to be fixed later, and let $\alpha = A(c,\gamma)$. Let $S$ be the advice schemas for $\Pi$ obtained from \Cref{def:composable} by using parameters $c$ and $\alpha$. Let $G = (V,E) \in \mathcal{G}$. Let $\ell$ be the bit-assignment of the schema $S$ on the graph $G$. We define the schema $S'$ by providing a function $\ell'$ that maps each node of $V$ into a single bit.
	
	We define a clustering of the bit-holding nodes as follows.
	Let $H \subseteq V$ be the set of bit-holding nodes. Take an arbitrary node $v \in H$ and remove it from $H$. Initialize a cluster $C$ as $C = \{v\}$. Then, repeatedly perform the following operation, until nothing new is added to $C$: take a node $u$ in $H$ that is at distance at most $d := \alpha/(10\gamma)$ from at least some node in $C$, remove $u$ from $H$, add $u$ to $C$. Then, recurse on the remaining nodes of $H$ to create new clusters.

	We first prove that each cluster $C$ has size bounded by $\gamma$. Assume, for a contradiction, that there exists a cluster $C$ that contains strictly more than $\gamma$ nodes, and consider the set of nodes $C' \subseteq C$ of size exactly $\gamma+1$ obtained while constructing $C$. Since each node of $C'$ has a neighbor in $C'$ at distance at most $d$, we get that all nodes in $C'$ must be at pairwise distance at most $\gamma d = \alpha/10$. 
	Consider an arbitrary node $v \in C'$. We obtain that, in its $\alpha$-radius neighborhood there are at least $\gamma+1$ nodes, which is a contradiction with the definition of composable advice schema.

	In the following, we define the function $\ell'$.
	At first, define $\ell'$ by assigning a $1$ to each bit-holding node, and $0$ to all the other nodes. Then, for each cluster $C$, we operate as follows. Let $N(C)$ be the set of nodes at distance at most $d / 4$ from the nodes in $C$. Observe that, for two different clusters $C$ and $C'$, it holds that the nodes in $N(C)$ are at distance at least $d/2$ from all the nodes in $N(C')$.
	For each node $v \in N(C)$, let $x_v$ be the distance between $v$ and its nearest node in $C$. There are two cases: 
	\begin{itemize}
		\item There is a node $z$ satisfying $x_z \ge d/8 + 10$. In this case, we will modify the assignment of bits in a suitable way.
		\item There is no node $z$ satisfying $x_z \ge d/8 + 10$. In this case, we will leave the assignment as is. In fact, in this case, the whole graph has diameter bounded by $O(\alpha)$, and hence the nodes can solve $\Pi$ by brute force, without the need of advice.
	\end{itemize}
	
	In the following, we present the bit-assignment for the clusters of the former case. Let $v_1,\ldots, v_k$ be the bit-holding nodes of $C$, sorted by their ID. Recall that the $k \le \gamma$, and that each bit-holding node has a bit-string of length at most $c \alpha / \gamma^3$.
	Let $B = (L_1,\ldots,L_k)$, where $L_{i} = \ell(v_i)$. We encode the array $B$ by using at most $4 c \alpha / \gamma^2$ bits, as follows. Each $L_i$ is encoded by writing $\lceil \log |L_i|  \rceil$ $1$s, followed by a $0$, followed by the bit-string representation of $|L_i|$, followed by $L_i$, which requires at most $4 c \alpha / \gamma^3$ bits. Then, we concatenate the obtained bit-strings, in ascending order of $i$, obtaining a string $L$. Note that, given $L$, it is possible to recover $B$.
		
	Then, we create a new bit-string, by replacing each $0$ of $L$ with a sequence of $\gamma+1$ $1$s followed by a $0$, and each $1$ of $L$ with a sequence of $\gamma+2$ $1$s followed by a $0$. Let $L'$ be the obtained bit-string. Note that this string has length at most  $(\gamma+3) 4 c \alpha / \gamma^2 \le 16 c \alpha / \gamma$, and that it is possible to recover $L$ given $L'$.
	
	Recall that $d/8 = \alpha/(80\gamma)$. We fix $c$ small enough so that $16 c \alpha / \gamma < d/8$. Consider the path $P = (v_1,\ldots,v_{d/8})$ obtained as follows. Node $v_1$ is $z$ (that is, the node satisfying $x_z \ge d/8 + 10$), and node $v_{i+1}$ is an arbitrary neighbor of $v_i$ satisfying $x_{v_{i+1}} = x_{v_{i}} - 1$. We pad the string $L'$ to be of length exactly $d/8$ by adding $0$s. Then, for each $i$, we set $\ell'(v_i)$ as the $i$th bit of $L'$ (possibly overriding the $0$ that has been previously assigned). 
	
	We now prove that, given the assignment provided by the function $\ell'$, it is possible for the nodes to solve $\Pi$. First, each node $v$ gathers its $\alpha$-radius neighborhood $\bar{N}(v)$.
	By looking at $\bar{N}(v)$, node $v$ can see which nodes within distance $\alpha$ are marked $1$. Among them, node $v$ can recognize which nodes correspond to bit-holding nodes, and which nodes correspond to nodes that encode bits of~$L'$:
	\begin{itemize}
		\item Nodes that form connected components of size at most $\gamma$ correspond to bit-holding nodes.
		\item Nodes that form components of size $\gamma+1$ or $\gamma+2$ correspond to nodes that are encoding bits of $L'$. 
	\end{itemize} 
	Then, each bit-holding node $v$ can compute which bit-holding nodes within distance $\alpha$ from $v$ belong to the same cluster $C$ of $v$. Then, $v$ can sort the nodes of $C$ by their IDs. Then, bit-holding nodes consider the nodes within distance $d/8 + 10$ from $C$ and their assigned bits, and in this way they are able to recover $L'$. If no $L'$ is present, it means that the whole graph has diameter $O(\alpha)$, and hence all the nodes can solve $\Pi$ by brute force. Otherwise, nodes can decode $L'$ to obtain the array $B$, and hence the advice for $\Pi$. 
	Then, in $T(\alpha,\Delta)$ time, nodes can solve $\Pi$.
	
	Finally, observe that, by making $c$ small enough (and possibly increasing $\alpha$ suitably), we can make the ratio between assigned $1$s and assigned $0$s an arbitrarily small constant.
\end{proof}

\end{document}